\documentclass[11pt,letterpaper]{article}
\usepackage[margin=1in]{geometry}

\usepackage[sort]{cite}

\usepackage{mdframed}
\usepackage[colorlinks,linkcolor=blue]{hyperref}

\usepackage{colortbl}
\usepackage{mdframed}

\usepackage{paralist}
\usepackage{pdfpages}
\usepackage{enumitem}
\usepackage{float}
\usepackage{booktabs}
\usepackage[override]{cmtt}
\usepackage{color,xcolor}
\usepackage{graphicx,color,eso-pic}
\usepackage{amsmath,amssymb}
\usepackage{amsmath,amsthm,amstext,amsfonts,amssymb,latexsym}
\usepackage{amsbsy}
\usepackage{stmaryrd}

\usepackage{pifont}

\newcommand\smalltriangledown{%
   {\vcenter{\hbox{${\scriptscriptstyle\triangledown}$}}}}

\usepackage[
	lambda,
	advantage,
	operators,
	sets,
	adversary,
	landau,
	probability,
	notions,
	logic,
	ff,
	mm,
	primitives,
	events,
	complexity,
	asymptotics,
	keys]
{cryptocode}

\usepackage{tikz}
\usetikzlibrary{arrows,automata}
\usetikzlibrary{positioning}

\newenvironment{MyItemize}[1]{\begin{itemize}\setlength{\itemsep}{0.1cm}
\setlength{\parskip}{-0.05cm} #1}{\end{itemize}}

\renewcommand{\highlightkeyword}[2][\ ]{\ensuremath{\textrm{#2}}#1}

\newcommand{\mcal}[1]{\ensuremath{\mathcal {#1}}}

\renewcommand{\poly}{\ensuremath{{{\sf poly}}}\xspace}

\renewcommand{\TC}{\ensuremath{{\bf TC}}}
\newcommand{\SP}{\ensuremath{{\bf SP}}}
\newcommand{\LC}{\ensuremath{{\bf LC}}}

\newcommand{\idx}{\ensuremath{\mathsf{idx}}\xspace}

\newcommand{\err}{\ensuremath{{\sf err}}\xspace}
\newcommand{\labels}{\ensuremath{{\mcal{L}}}\xspace}

\newcommand{\czigzag}{\ensuremath{C_{\rm zigzag}}\xspace}

\definecolor{darkgreen}{rgb}{0,0.5,0}
\definecolor{lightblue}{RGB}{0,176,240}
\definecolor{darkblue}{RGB}{0,112,192}
\definecolor{lightpurple}{RGB}{124, 66, 168}
\definecolor{grey}{RGB}{139, 137, 137}
\definecolor{maroon}{RGB}{178, 34, 34}
\definecolor{green}{RGB}{34, 139, 34}
\definecolor{types}{RGB}{72, 61, 139}
\definecolor{gold}{rgb}{0.8, 0.33, 0.0}

\definecolor{darkgray}{gray}{0.3}

\newcommand{\skiptext}[1]{}

\newcommand{\splitter}{{segmenter}\xspace}
\newcommand{\Splitter}{{Segmenter}\xspace}

\usepackage{boxedminipage}
\usepackage{cite}
\usepackage{url}
\usepackage{graphicx}
\usepackage[bf,justification=centering]{caption}
\usepackage[justification=centering]{subcaption}
\usepackage{color}
\usepackage{xspace}
\usepackage{multirow}
\usepackage{amsmath,amstext,amssymb,amsfonts,latexsym}
\usepackage{wrapfig}

\usepackage{algorithm}

\definecolor{darkred}{rgb}{0.5, 0, 0}
\definecolor{darkgreen}{rgb}{0, 0.5, 0}
\definecolor{darkblue}{rgb}{0,0,0.5}

{
\theoremstyle{definition}
\newtheorem{Definition}{Definition}
}

\newcommand\markx[2]{}

\renewcommand{\path}{\ensuremath{\mathsf{path}}\xspace}

\newcommand{\N}{\mathbb{N}}

\newcommand{\ignore}[1]{}

\renewcommand{\paragraph}[1]{\vspace{5pt}\noindent\textbf{#1}}

\newcounter{task}

\newenvironment{proofof}[1]{\noindent \textbf{Proof of #1:}}{\hfill \qed}

\newtheorem{thm}{Theorem}[section]      

\newtheorem{theorem}[thm]{Theorem}

\newtheorem{lemma}[thm]{Lemma}

\newtheorem{corollary}[thm]{Corollary}
\newtheorem{fact}[thm]{Fact}
\newtheorem{proposition}[thm]{Proposition}

{
\theoremstyle{definition}
\newtheorem{Remark}{Remark}
}

\newcommand{\elaine}[1]{{\footnotesize\color{magenta}[Elaine: #1]}}
\newcommand{\gnote}[1]{{\footnotesize\color{blue}[Gilad: #1]}}
\newcommand{\weikai}[1]{{\footnotesize\color{green}[WK: #1]}}

\renewcommand{\elaine}[1]{}
\renewcommand{\gnote}[1]{}
\renewcommand{\weikai}[1]{}

\newcounter{cnt:challenge}

\ignore{

}

\newcommand{\I}{\ensuremath{{\bf I}}\xspace}

\begin{document}

\newcommand\relatedversion{}
\renewcommand\relatedversion{\thanks{The full version of the paper can be accessed at \protect\url{https://arxiv.org/abs/2102.11489}}} 



 \title{\bf Optimal Sorting Circuits for Short Keys
 \footnote{Author ordering is randomly generated.}}

 \author{Wei-Kai Lin \\ Cornell \\ {\tt wklin@cs.cornell.edu}
 \and Elaine Shi
 \\ CMU \\
 {\tt runting@cs.cmu.edu}}
\date{}

\maketitle



\begin{abstract}

A long-standing open question in the algorithms and complexity literature
is whether there exist sorting circuits of size $o(n \log n)$.
A recent work by Asharov, Lin, and Shi (SODA'21) showed that if 
the elements to be sorted have short keys 
whose length $k = o(\log n)$,
then one can indeed overcome the $n\log n$ barrier 
for sorting circuits, by leveraging non-comparison-based techniques. 
More specifically, Asharov et al.~showed
that there exist $O(n) \cdot \min(k, \log n)$-sized 
sorting circuits for $k$-bit keys, ignoring $\poly\log^*$ factors.
Interestingly, the recent works 
by Farhadi et al. (STOC'19)
and Asharov et al. (SODA'21) also showed that the above result is 
essentially optimal for every key length $k$,  
assuming that the famous Li-Li network coding conjecture holds.
Note also that proving any {\it unconditional} super-linear circuit lower bound
for a wide class of problems is beyond the reach of current techniques.

Unfortunately, the 
approach taken by previous works to achieve optimality in size
somewhat crucially relies on sacrificing the depth:  
specifically, 
their circuit is super-{\it poly}logarithmic in depth even for
1-bit keys.
Asharov et al.~phrase it as an open question 
how to achieve optimality both in size and depth. 
In this paper, we 
close this important gap in our understanding.
We construct a sorting circuit of size $O(n) \cdot \min(k, \log n)$
(ignoring $\poly\log^*$ terms) and depth $O(\log n)$.
To achieve this, our approach departs significantly  
from the prior works.
Our result can be viewed as a generalization of the landmark
result by Ajtai, Koml\'os, and Szemer\'edi 
(STOC'83), 
simultaneously in terms of size and depth. 
Specifically, for $k = o(\log n)$, 
we achieve asymptotical improvements in size over the AKS sorting circuit,
while preserving optimality in depth.

\ignore{
Although proving any super-linear circuit lower bound 
for a large class of problems {\it unconditionally} is beyond 
the reach of current techniques, the recent works
of Farhadi et al. (STOC'19) and Asharov et al. (SODA'21)
proved a conditional lower bound:
assuming that the famous Li-Li network coding conjecture
is true, any circuit that sorts elements with $k$-bit keys
must have size $\Omega(n) \cdot \min(k, \log n)$.
}

\end{abstract}


\pagestyle{plain}

\section{Introduction}


Sorting circuits have been investigated 
for a long time in the algorithms and complexity theory
literature, and it is almost surprising that we still
do not fully understand sorting circuits.
Suppose we want to sort an input array with $n$ elements, each
with a $k$-bit comparison key and a $w$-bit payload.
A long-standing open question is whether there exist
circuits with $(k + w) \cdot o(n \log n)$ boolean gates where each gate
is assumed to have constant fan-in and constant 
fan-out~\cite{isthereoramlb}.
The recent works of Farhadi et al.~\cite{sortinglbstoc19} (STOC'19)
showed that assuming the famous Li-Li network 
coding conjecture~\cite{lilinetcoding}, 
it is impossible to construct sorting
circuits of size $(k + w) \cdot o(n \log n)$
when there is no restriction on the key length $k$.
Given this conditional lower bound, we seem to have hit another wall. 
However, shortly afterwards,  
Asharov, Lin, and Shi~\cite{soda21} 
showed that we can indeed overcome the $n \log n$
barrier for {\it short keys}, specifically, when $k = o(\log n)$.
More specifically, Asharov et al.~showed
that an array containing $n$ elements
each with a $k$-bit key and a $w$-bit payload
can be sorted in a circuit of size
$(k+w) \cdot O(n) \cdot \min(k, \log n)$ 
(ignoring $\poly\log^*$ terms);
moreover, Asharov et al.~\cite{soda21}
prove that this is optimal for every choice of $k$.

Asharov et al.~\cite{soda21}'s result 
moved forward our understanding on sorting circuits, 
since it achieved asymptotical improvements for short keys 
relative to the landmark result
by Ajtai, Koml\'os, and Szemer\'edi~\cite{aks} (STOC'83), who  
constructed sorting circuits containing $O(n \log n)$
{\it comparator}-gates.
As Asharov et al.~\cite{soda21} point out, 
an $o(n \log n)$ sorting circuit for short keys  
might have eluded the community
earlier due to 
a couple natural barriers.
First, an $o(n \log n)$ sorting circuit is 
impossible in the {\it comparator-based} model
even for $1$-bit keys --- this follows
partly due to the famous 0-1 principle which
was described in Knuth's textbook~\cite{knuthbook}.
Indeed, Asharov et al.~\cite{soda21}
is the first to show how to leverage {\it non-comparison-based}
techniques to achieve a non-trivial sorting result
in the circuit model.
Earlier, non-comparison-based sorting
was investigated in the Random Access Machine (RAM) model
to achieve almost linear-time sorting~\cite{sortlinearandersson,Kirkpatricksort,hansort00,hansort01,thorupsort}
but it was unknown how non-comparison-based 
techniques can help in the circuit model.
The second natural barrier pertains to the {\it stability}
of the sorting algorithm. 
Stability requires that elements
with the same key should preserve 
the same order as they appear in the input.
Recent works~\cite{osortsmallkey,lbmult} have shown that 
an $o(n \log n)$-sized {\it stable} sorting circuit is impossible
even for $1$-bit keys, 
if we either assume the Li-Li network coding conjecture~\cite{lilinetcoding}
or assume that the circuit 
follows the so-called {\it indivisibility} model (i.e., 
the circuit does not perform encoding or computation
on the elements' payloads. 
Therefore, to achieve their result, 
Asharov et al.~\cite{soda21}
had to forgo both the comparator-based restriction as well as the
stability requirement.

Despite the progress, 
Asharov et al.~\cite{soda21}'s result 
is nonetheless unsatisfying --- to achieve optimal circuit size, 
they pay a significant price in terms of depth: 
their circuit
is $(\log n)^{\omega(1)}$ in depth even for 1-bit keys.
In fact, as written, the depth of their circuit is 
{\it super-linear} --- 
however, with some work, it is possible to leverage existing techniques~\cite{paracompact}
to improve their depth to 
$(\log n)^{O(\log (\log^* n))}$, which grows 
asymptotically faster than any poly-logarithmic function.
We are not aware of any known technique that 
can improve the depth to 
even polylogarithmic, even for $1$-bit keys, 
while still preserving the $o(n \log n)$ circuit size.
\elaine{todo: calculate their depth more carefully}
\weikai{
I believe $(\log n)^{O(\log\log^* n)}$ will bound the depth:
using a tight compaction of depth $(\log n)^{c_1\log\log^* n}$,
the depth of median is $D(n)=D(n/5)+(\log n)^{c_1\log\log^* n}+D(7n/10)$,
which solves to $D(n) \le (\log n)^{c_2\log\log^* n}$ for constant $c_2 > c_1$,
and then two-parameter recursive sorting incurs another $\log n$.
}

We therefore ask the following natural question, which
was also phrased as the main open question in the 
work by Asharov et al.~\cite{soda21}:
\begin{itemize}[leftmargin=5mm,topsep=2pt]
\item[]
{\it Can we construct sorting circuits for short keys 
optimal both in size and depth?}
More concretely, 
can we sort $n$ elements each with a $k$-bit
key and $w$-bit payload in a circuit
of size $(k+w) \cdot O(n) \cdot \min(k, \log n)$
and of logarithmic depth?
\end{itemize}
If we could achieve the above, we would get a result
that strictly generalizes AKS~\cite{aks} (taking 
both circuit size and depth into account). 
Independently and concurrently to this work,
Kouck{\'{y}} and Kr{\'{a}}l~\cite{KK21}
also improved the depth to $O(\log^3 n)$;
we will summerize and compare their results later in Section~\ref{sec:additional_results},
but the above question remains open even after their work.

\subsection{Our Main Result}
\label{sec:mainresult}

We answer the above question affirmatively except
for an extra $\poly\log^*$ factor in the circuit size.
We explicitly construct a sorting circuit
for short keys that is optimal in size 
modulo $\poly\log^*$ factors, and optimal in depth, 
as stated in the following theorem:

\begin{theorem}[Optimal sorting circuits for short keys]
Suppose that $n > 2^{4k + 7}$.
There is a constant fan-in, constant fan-out 
boolean circuit that correctly sorts 
any array containing $n$ elements each with a $k$-bit key 
and a $w$-bit payloads, whose 
size is $O(n k (w+k)) \cdot \max(1, \poly(\log^* n - \log^* (w+k)))$
and whose depth is $O(\log n + \log w)$.
\label{thm:intro-sort-circ}
\end{theorem}
The circuit size is {\it optimal} upto $\poly\log^*$ 
factors {\it for every $k$}
due to a lower bound by Asharov et al.~\cite{soda21} (assuming
either the invisibility model or the Li-Li network coding conjecture).
Furthermore, $\Omega(\log n)$ depth is necessary even for $1$-bit
keys, as implied by the lower bound 
of Cook et al.~\cite{orlbpram};
moreover, the $\log w$ part of the depth is needed
even for propagating the comparison result to 
all bits of the output.
Our sorting circuit leverages {\it non-comparison-based} techniques,
and moreover it 
does {\it not} preserve stability --- as mentioned
earlier, forgoing the comparison-based restriction
and the stability requirement is inherent even for the 1-bit key special case.

\subsection{Technical Highlights}

\paragraph{Blueprint and challenges.}
To get our main results, we need two major stepping stones: 
\begin{enumerate}[leftmargin=5mm,]
\item 
{\it Linear-sized, logarithmic-depth compaction circuit.}
First, we solve the problem for the 1-bit special case. We show how to get a 
1-bit sorting circuit (also called a compaction circuit)
that is linear in size (modulo $\poly\log^*$ factors) and logarithmic in depth.
In comparison, the prior state-of-the-art~\cite{soda21} is also  
linear in size (modulo $\poly\log^*$ factors) but suffers
from $(\log n)^{\omega(1)}$  depth.
\item 
{\it $1$-bit to $k$-bit upgrade.} 
Next, our goal is to upgrade $1$-bit sorting to $k$-bit sorting.
Since any $o(n \log n)$-sized circuit 
that sorts $1$-bit keys
inherently cannot be stable~\cite{lbmult,soda21}, 
we cannot use classical techniques such as Radix sort 
to get $k$-bit sorting from $1$-bit sorting. 
To date, the only known technique for accomplishing 
the $1$-bit to $k$-bit upgrade {\it without relying on stability}
was a clever two-parameter recursion trick 
suggested by Lin, Shi, and Xie~\cite{osortsmallkey}. Unfortunately, their approach
incurs at least polylogarithmic depth.
We propose a brand new paradigm for performing the $1$-bit to $k$-bit upgrade,
elaborated below.
\end{enumerate}

We now explain at a very high level the novel ideas
that allow us to overcome these challenges.

\paragraph{Technical highlight: a brand new 1-bit to $k$-bit upgrade.}
Given a linear-sized, logarithmic-depth circuit that sorts
$1$-bit keys, we want to leverage it to construct
a $k$-bit sorting circuit that is $O(k)$ times larger in size,
and without blowing up the depth.
As mentioned, the only known prior technique~\cite{osortsmallkey}
for performing this upgrade inherently suffers from poly-logarithmic depth,
and it seems unlikely that we can hope to overcome this depth barrier 
if we stick to the known technical frameworks.
Therefore, our approach completely departs from prior works.

Our novel idea 
lies in using the famous  AKS
construction in a non-blackbox manner. Specifically, we 
propose a new building block called a {\it nearly ordered segmenter}
which can be constructed by running the beginning $6k = o(\log n)$ layers of the 
AKS circuit.
We prove that such a nearly ordered segmenter 
can {\it partially sort} an input array in the following sense:
if we divide the outcome into $2^{3k}$ segments,
then, inside each segment, at most  
$\frac{1}{2^{8k}}$ fraction of elements do not belong to the current segment.
This new abstraction ``nearly ordered segmenter'' is of independent interest 
and may be useful in other applications.

Now, imagine that we apply a 
nearly ordered segmenter to an input array with a small number of distinct keys
(specifically, $2^k$ distinct keys), 
resulting in $2^{3k}$ segments, where for each segment,
only a small fraction of elements are in the wrong segment.
Since there are only $2^k$ distinct keys but
as many as $2^{3k}$ segments, 
most of the segments would have only a single key had the array been
completely sorted.
This means that if we apply the 
nearly ordered segmenter to an input array with only $2^k$ distinct keys, then
we can prove something even stronger about the outcome: in fact, 
only a small fraction of the elements are {\it misplaced} in the sense
that they do not belong to the current {\it position} (had the array been completely
sorted).

If we could somehow extract these misplaced elements, sort them,
and then route the sorted result back into the misplaced positions, then we could  
fully sort the input!
Indeed, this is what we do, and we accomplish this 
with the help of the compaction circuit.
How to use the compaction circuit to correct the remaining errors
turns out to be very much non-trivial too. 
There are two main technical challenges: first, 
even identifying 
which elements are misplaced (subject
to the desired performance bounds) is non-trivial;
second, after we determine which set of 
possibly misplaced elements to extract, sort, and route back,
we cannot directly use compaction to perform the extraction and 
route-back because the compaction circuit is {\it unstable}!
We discuss how to overcome these technical challenges
in Section~\ref{sec:roadmap} and the subsequent formal sections.

\paragraph{Technical highlight: linear-size, logarithmic-depth compaction circuit.}
To get this result, 
we need fairly sophisticated and novel techniques. 
At a high level, to avoid
suffering from the super-polylogarithmic depth of Asharov et al.~\cite{soda21}, 
we first construct various building
blocks that can be regarded as relaxations of (tight) compaction. 
Specifically, by relaxing compaction along
several different axes, we define several new, intermediate abstractions, 
each of which will play a role in the
final construction. We show that the relaxed abstractions can be realized in sub-logarithmic or logarithmic
depth. We then gradually bootstrap these building blocks into stronger ones, and the final tight compaction
circuit is achieved through multiple steps of bootstrapping. We defer 
the details to Section~\ref{sec:roadmap}.

\subsection{Additional Result in the Oblivious PRAM Model}
\label{sec:additional_results}
Along the way towards getting our main result (Theorem~\ref{thm:intro-sort-circ}),
we also get an intermediate result for the oblivious Parallel RAM (PRAM) model: 
we show how to construct a deterministic, oblivious PRAM algorithm
that sorts short keys, optimal in both total work and depth
(and this time without the extra $\poly\log^*$ factors). As we explain below,
even this intermediate result is interesting in its own right.
Note also that 
this intermediate oblivious PRAM 
result does not directly give our circuit result --- 
partly, this is because on a PRAM, word-level operations 
on $\log n$-bits can be accomplished with unit cost; but there is no such free lunch
in the circuit model.  
Specifically, for the Oblivious PRAM model,
an optimal compaction 
algorithm linear in total work and logarithmic 
was known~\cite{paracompact} and we could directly rely on that in the
$1$-bit to $k$-bit upgrade.
Unfortunately, the circuit counterpart of this result is unknown, 
and getting the circuit counterpart of this result is highly non-trivial as
our paper shows.

A deterministic algorithm 
in the oblivious PRAM model is a PRAM algorithm whose
memory access patterns do not depend on the input 
(except the input size).
We show that indeed, one can obliviously sort 
$n$ elements
each with a $k$-bit key in $O(n) \cdot \min(k, \log n)$ total
work and $O(\log n)$ depth, assuming that each element
can be stored in $O(1)$ memory words.
The total work is optimal assuming either
the indivisibility model or the Li-Li network coding
conjecture~\cite{osortsmallkey,soda21}, and the depth
is optimal unconditionally even for 1-bit keys~\cite{orlbpram}.

\begin{theorem}[Sorting short keys on an oblivious PRAM]
There exists a deterministic oblivious parallel algorithm
that sorts any input array
containing $n$ elements each with a $k$-bit key
in $O(n) \cdot \min(k, \log n)$
total work and $O(\log n)$ depth, assuming that each element
can be stored in $O(1)$ words\footnote{Note that the theorem statement
for oblivious PRAM does {\it not}
 have an extra $\poly\log^*$ blowup in total work.}.
\label{thm:intro-sort-opram}
\end{theorem}

Prior to our work, it was known that   
$n$ elements with $k$-bit keys 
can sorted by a {\it randomized} oblivious algorithm 
in $O(k n \frac{\log \log n}{\log k})$
work and polylogarithmic depth~\cite{osortsmallkey}.
It is possible to improve the total work to  
$O(k n)$ and get rid of the randomization 
by combining techniques from Lin et al.~\cite{osortsmallkey}
and Asharov et al.~\cite{paracompact}.
However, to the best of our knowledge, existing techniques 
are stuck at polylogarithmic depth. 
To attain the above result,  
our techniques depart significantly from the prior 
works~\cite{osortsmallkey,soda21}.

\paragraph{Concurrent work of Kouck{\'{y}} and Kr{\'{a}}l~\cite{KK21}.}\footnote{
  This manuscript was written in 2020, 
  and then we posted it in February 20201 (arXiv:2102.11489),
  four days after Kouck\'y and Kr\'al.
}
In the independent and concurrent work,
Kouck{\'{y}} and Kr{\'{a}}l construct a sorting circuit of size 
$O(nk(w+k)\cdot(1+\log^*n - \log^*(w+k)))$ and depth $O(\log^3 n)$
for $k\le (\log n)/11$ bits.
Compared to 
our result, their circuit depth is still poly-logarithmic, 
since they directly adopt the two-parameter recursion trick
by Lin et al.~\cite{osortsmallkey} 
to upgrade from $1$-bit sorting to $k$-bit sorting. 
As mentioned, this framework inherently suffers from at least
poly-logarithmic depth. We got around this issue by proposing
a brand new framework 
for this 1-bit to $k$-bit upgrade.
On the other hand, Kouck{\'{y}} and Kr{\'{a}}l tightened 
the $\poly\log^*$ factor to $\log^*$. 
In addition, Kouck{\'{y}} and Kr{\'{a}}l consider another variant of sorting circuit
that sorts $n$ integers each of $k$ bits \emph{without payload},
commonly referred to as {\it integer sorting}.
For integer sorting, 
they claim circuit size $O(nk^2)$ and depth $O(\log n + k \log k)$.
Our paper mainly focuses on sorting with payload.


\section{Technical Roadmap}
\label{sec:roadmap}

We give an informal technical overview of our ideas in this section.

\subsection{Sorting Short Keys on an Oblivious PRAM (a.k.a. $1$-bit to $k$-bit Upgrade)}
\label{sec:roadmap-opram}
As an intermediate stepping stone, we first consider how to sort $k$-bit keys 
on an Oblivious PRAM in $O(n) \cdot \min(k, \log n)$ total work and $O(\log n)$ depth.
Without loss of generality, 
we assume that $k < \frac{1}{8}\log n$ in the following exposition
where $n$ denotes the length of the array 
to be sorted; since if $k \geq \frac{1}{8}\log n$, we can simply run AKS~\cite{aks}
to sort the array.
We also assume that $n$ is a power of $2$; if not, we can pad
it with elements with $\infty$ keys to the next power of $2$.
We assume that each element can be stored in $O(1)$ memory words.

In the Oblivious 
PRAM model, Asharov et al.~\cite{paracompact} showed how to get an optimal compaction algorithm that is linear
in total work and logarithmic in depth
(even though the optimal counterpart in the circuit model 
is not known prior to our work).
Our goal is to upgrade the 1-bit sorting (i.e., compaction) 
to $k$-bit sorting. 
As mentioned, we cannot build upon existing approaches for this upgrade
since they incur polylogarithmic depth~\cite{osortsmallkey,soda21}. 
We therefore suggest a brand new approach.


\subsubsection{New Abstraction: Nearly Orderly Segmenter}
We propose a new abstraction called an {\it $(\eta, p)$-orderly 
segmenter}, where $\eta \in (0, 1)$ indicates how sorted
the resulting array is, and $p$ 
denotes the number of segments.
\ignore{
Henceforth, we use the notation $p$ to 
denote the number of segments, 
and we use $\eta \in (0, 1)$ to denote the 
fraction of elements 
in each segment that are misplaced.
}
An array $A := A_1 || A_2 || \ldots || A_p$,
represented as the concatenation of $p$ equally sized partitions
denoted $A_1, A_2, \ldots, A_p$, 
is said to be $(\eta, p)$-orderly iff  
in each of the $p$ segments, at most $\eta$ fraction of the elements
belong to the \emph{wrong segment}
if the array were to be fully sorted. 
An $(\eta, p)$-orderly segmenter 
receives an input array whose length is divisible by $p$, 
and outputs a permutation of the input array that 
is $(\eta, p)$-orderly. 

We then show how to construct a deterministic, oblivious 
$(2^{-8k}, 2^{3k})$-orderly segmenter
that requires $O(nk)$ total work and $O(k)$ depth.
The construction involves partially executing  
the AKS algorithm~\cite{aks}.
Recall that the full AKS algorithm would execute for a total of $\log n$ cycles.
In each cycle, the following is repeated for $O(1)$ number of times: 
partition the array into disjoint partitions where  
each partition may not be a contiguous region 
in the original array, and apply an $\epsilon$-near-sorter 
to each partition in parallel where $\epsilon \in (0, 1)$ is a sufficiently
small constant\footnote{An $\epsilon$-near-sorter is a constant depth 
comparator circuit described by Ajtai et al.~\cite{aks}, which we will
formally define in the subsequent technical sections.}.
Our key observation is the following:

\begin{quote}
\textbf{Observation.} 
If we execute the AKS algorithm 
not for the full $\log n$ cycles, but only for $6k$ cycles, 
it gives a $(2^{-8k}, 2^{3k})$-orderly segmenter.
\end{quote}

The proof of the above statement is rather technical  
since it requires us to {\it use the properties of AKS in a non-blackbox manner}.
We defer the proof to the formal technical sections.

\paragraph{Applying a nearly ordered segmenter to an array with few
distinct keys produces an almost sorted array.}
One helpful intuition is the following:  
if we run AKS for only $o(\log n)$ cycles, in general, 
we cannot guarantee sortedness within 
each segment of length $n/2^{o(\log n)}$.
Specifically, had the number of distinct keys been large, 
running AKS for only $o(\log n)$ cycles
could produce an outcome that is far from sorted (i.e.,
a large number of elements are misplaced in the \emph{wrong position}
even they are in their correct segments).

Fortunately, 
our input array has relatively few distinct keys --- specifically, at most 
$2^k$ distinct keys.
This means that if the array were fully sorted,
then almost every segment consists of identical keys
except for $2^k$ segments.
In this case, applying a $(2^{-8k}, 2^{3k})$-orderly segmenter 
results in an array that is close to fully sorted, i.e., 
only a small $O(1/2^{2k})$ fraction 
\elaine{what is the fraction?}\weikai{filled}%
of elements are misplaced in the wrong position.
Given this crucial observation, what remains to be done
is to {\it extract} the misplaced elements,  
{\it sort} them, and {\it route them  
back} into the original positions while preserving the sorted order. 
We now discuss how to accomplish this goal --- doing so turns out to be 
rather non-trivial, and we first need to 
construct some new building blocks which we describe next.

\subsubsection{Additional New Building Blocks}
Henceforth, let $K := 2^k$.
We will need the following new building blocks to be able to
extract, sort, and route back the remaining errors.

\paragraph{${\bf SlowSort}^K(A)$:} an inefficient oblivious sort algorithm --- when
given an array $A$ of length $m$ with at most $K := 2^k$ distinct keys,
the algorithm outputs a sorted permutation of $A$. 
We would like to accomplish ${\bf SlowSort}^K(A)$ 
in $O(mK)$ total work and $O(\log m + k)$ depth, since later 
we will apply ${\bf SlowSort}^K(A)$
to arrays of size $n/K$ where $n$ is the length of the larger array 
we need to sort.

It turns out that even this slow version is somewhat non-trivial to construct.
The most obvious idea, that is, relying on AKS~\cite{aks}, does not work.
AKS would have incurred $O(m \log m)$ 
work; and for small choices of $K$, $\log m$ could be larger than $K$.

We instead make $K$ copies of the input array:
in the $u$-th copy, we want to put elements
with the key $u \in [0, K-1]$ into the right positions, where
as all other elements should be fillers.
If we can accomplish this, 
we can sort $A$ by performing a coordinate-wise $K$-way selection
among the $K$ arrays.

Specifically, let $s_u$ be the number of elements smaller than $u$.
In the sorted array, elements with the key $u$ should appear
in positions $s_u + 1, s_u + 2, \ldots, s_{u+1}$.
Now, in the $u$-th copy, we preserve all the elements with the key $u$
but replace all other elements with fillers.
We mark exactly $s_u$ fillers with the key $-\infty$ and 
the mark rest of fillers with the key $\infty$.
Now, the $u$-th copy of the problem boils down to sorting 
$m$ elements with 3 different keys. 
We show that this can be accomplished 
in linear time and logarithmic depth, 
if we leverage 
the linear-work, logarithmic depth 
oblivious compaction~\cite{paracompact} algorithm (we defer
the details of the construction to subsequent technical sections). 

\paragraph{${\bf FindDominant}(A)$:}
let $A$ be an input array containing $n$ elements
each with a $k$-bit key, and let $\epsilon \in (0, 1/2)$.
We say that an array $A$ of length $n$ is $(1 - \epsilon)$-uniform iff
except for at most $\epsilon n$
elements, all other elements in $A$ have the same key --- henceforth
this key is said to be the {\it dominant} key.
We will need an oblivious algorithm 
${\bf FindDominant}(A)$ which finds the 
dominant key among an $(1-2^{-8k})$-uniform input array $A$ 
containing $n$ elements;
further, we want to accomplish this in $O(n)$ total work and
$O(\log n + k)$ depth.
We construct an oblivious algorithm for solving this problem 
that is reminiscent of 
Blum et al.~\cite{blumselect}'s median-finding algorithm, and 
moreover the algorithm employs ${\bf SlowSort}^K$ as a building block 
--- see the 
subsequent technical section for details.

\subsubsection{Sorting a $(2^{-8k}, 2^{3k})$-Orderly Array}
Let $A$ be a $(2^{-8k}, 2^{3k})$-orderly array containing
$n$ elements with $k$-bit keys, and recall that $K := 2^k$.
If $A$ were to be fully sorted, then 
among the $K^3$ segments, at most $K$ segments can have multiple keys;
and all remaining segments must have only a single key. 
Since $A$ is $(2^{-8k}, 2^{3k})$-orderly, it means that 
all but $K$ segments of $A$ must be $(1-2^{-8k})$-uniform, i.e., 
all but $2^{-8k}$
fraction of elements have the same key.

Our goal is to extract, sort, and then route back all the misplaced elements
that do not belong to the right position.
However, it turns out that even identifying which elements are 
misplaced (within the desired performance bounds) is challenging.
Rather than identifying the set of misplaced elements precisely,
we will instead identify a {\it superset}
of the misplaced elements, 
including 1) every segment that is not $(1-2^{-8k})$-uniform,
and 2) roughly $O(2^{-8k})$ fraction of elements from each 
$(1-2^{-8k})$-uniform segments.
Observe that for any $(1-2^{-8k})$-uniform segment, 
at most $2^{-8k}$ fraction of elements with the dominant key can be misplaced;
and moreover, elements whose keys are not dominant might be misplaced too. 
To identify a small superset of misplaced elements 
in each segment, we must first identify the dominant key of each segment.
To this end, we make use of the aforementioned ${\bf FindDominant}$
building block which correctly identifies the dominant key
as long as the segment is $(1-2^{-8k})$-uniform. 

\paragraph{Strawman algorithm.}
Based on these ideas, let us first look at a 
flawed strawman algorithm that makes use of compaction to extract
the misplaced elements, and route them back after they are sorted:
\begin{mdframed}
\begin{center}
{\bf Flawed strawman idea: sorting a $(2^{-8k}, 2^{3k})$-orderly array}
\end{center}
\begin{enumerate}[leftmargin=5mm,itemsep=1pt,topsep=2pt]
\item 
Each segment decides if it is $(1-2^{-8k})$-uniform or not. 
That is, each segment calls ${\bf FindDominant}$
to find its dominant key. 
If the segment is indeed $(1-2^{-8k})$-uniform, 
${\bf FindDominant}$ is guaranteed to return the correct dominant key;
else an arbitrary result may be returned.

\item 
\label{step:roadmap-extract1}
Use oblivious compaction to extract 
1) all segments in $A$ that 
are not $(1-2^{-8k})$-uniform, and 2) from each $(1-2^{-8k})$-uniform 
segment:  extract all elements whose keys differ from the dominant key,
and extract $2^{-8k} \cdot (n/K^3)$ 
elements with the dominant key
where $n/K^3$ is the segment size.
We can show that the number of extracted elements
is upper bounded by $3n/K^2$; and there is a way to 
pad the extracted array 
with fillers 
to a fixed length of $3n/K^2$ to hide how long it actually is.

Note that the extracted elements contain all the
elements that belong to incorrect segments,
but possibly some additional elements too.  
The invariant we want to maintain here is that 
{\it all remaining elements
must belong to the right segment and all segments are uniform}.

\item 
Call ${\bf SlowSort}^K$ to sort the extracted array 
and reverse route the result back to the original array.
\label{step:roadmap-slowsortextracted}

\item 
At this moment, all elements fall into the correct segment, 
but if a segment has multiple keys, it may not be sorted internally.
Fortunately, we know that at most $K$ segments can be multi-keyed. 
Therefore, we use oblivious compaction to extract these $K$ segments,
call ${\bf SlowSort}^K$ 
to sort within each extracted segment, reverse 
route the result back, and output the final result.
\label{step:roadmap-sortwithin}
\end{enumerate}
\end{mdframed}

This algorithm almost works except for one subtle issue
that breaks correctness: the linear-work, logarithmic-depth 
oblivious compaction algorithm~\cite{paracompact}
is {\it not stable} and in fact this is inherent~\cite{osortsmallkey,lbmult}.
This means that in the Step~\ref{step:roadmap-extract1} above, 
the extracted elements do not preserve the order 
in which they appear in the input array.

\paragraph{The fix: granularity switching.}
Note that we do not need full stability, we just
need to make sure that the extracted elements 
are ordered based on their segment numbers --- this way,
after we sort the extracted elements and route the sorted elements back into  
the original positions, 
every element will land in the correct segment.

Therefore, to fix this problem, one na\"ive idea is to use ${\bf SlowSort}^{K^3}$
to sort the extracted array once again
based on which segment each element belongs to,  
but this would be too costly
since it would incur $K^3 \cdot 3n/K^2 = O(nK)$ work.
We propose a {\it granularity-switching} idea.
Specifically, we switch to a 
more coarse-grained partitioning scheme 
at this point: we instead view the array as 
$K^2$ super-segments, where each super-segment
is the concatenation of $K$ original segments.
Therefore, we use ${\bf SlowSort}^{K^2}$
to sort the extracted array whose length is $3n/K^2$, and this incurs
$O(n)$ work and $O(\log n + k)$ depth.
At this moment, we can follow through with 
Steps~\ref{step:roadmap-slowsortextracted}
and \ref{step:roadmap-sortwithin}, with the following modifications:
\begin{enumerate}[leftmargin=5mm,itemsep=1pt,topsep=2pt]
\item 
The reverse-routing in Step~\ref{step:roadmap-slowsortextracted}
now needs to reverse the decision of the ${\bf SlowSort}^{K^2}$
instance as well as the compaction.
\item At the end of Step~\ref{step:roadmap-slowsortextracted}, every
element now belongs to the correct super-segment. 
Therefore, Step~\ref{step:roadmap-sortwithin} 
now works on the super-segments rather than the segments. 
\end{enumerate}

We defer a  
detailed description of our final algorithm
to the subsequent formal sections.


\subsection{Linear-Sized, Logarithmic-Depth Compaction Circuit}
\label{sec:roadmap-compact-circ}

So far, we can get an oblivious PRAM algorithm that sorts $k$-bit keys 
in $O(n) \cdot \min(k, \log n)$ total work and $O(\log n)$ depth ---
to achieve this, we  
critically make use of an oblivious PRAM algorithm for compaction
that is linear in total work and logarithmic in depth.
Eventually, we want to get the circuit counterpart of this result.
A critical missing link is a  
compaction circuit optimal in both size and depth. 
Even though we know how to construct an optimal compaction
algorithm on an Oblivious PRAM~\cite{paracompact} and this may seem
tantalizingly close to a compaction circuit, unfortunately the oblivious PRAM
result does not  
directly translate to the circuit model --- if one tries to 
directly convert
the oblivious PRAM algorithm 
to the circuit model, it results in 
a circuit $O(nw + n \log n)$ in size~\cite{soda21}. 
Partly, this is because word-level
operations on $\log n$ bits 
can be accomplished in unit cost on a PRAM but there is no such free lunch
in the circuit model.
The recent work of Asharov, Lin, and Shi~\cite{soda21}
showed how to obtain a compaction circuit that is  
$O(nw)$ size (ignoring $\poly\log^*$ terms), 
but their circuit depth (as written) is linear.


We now describe 
how to get a compaction circuit 
that is not only optimal in size upto $\poly\log^*$ factors, 
but also optimal in depth.
To accomplish this goal, we go through 
several steps of bootstrapping that takes us from weaker primitives
to stronger primitives.
Specifically, we need several intermediate abstractions ---
all of these abstractions can be viewed in some way
as a relaxation of (tight) compaction; 
but each relaxation is of an incomparable nature.
We will first define all these intermediate abstractions, and
then we explain our blueprint 
for getting an optimal tight compaction circuit.


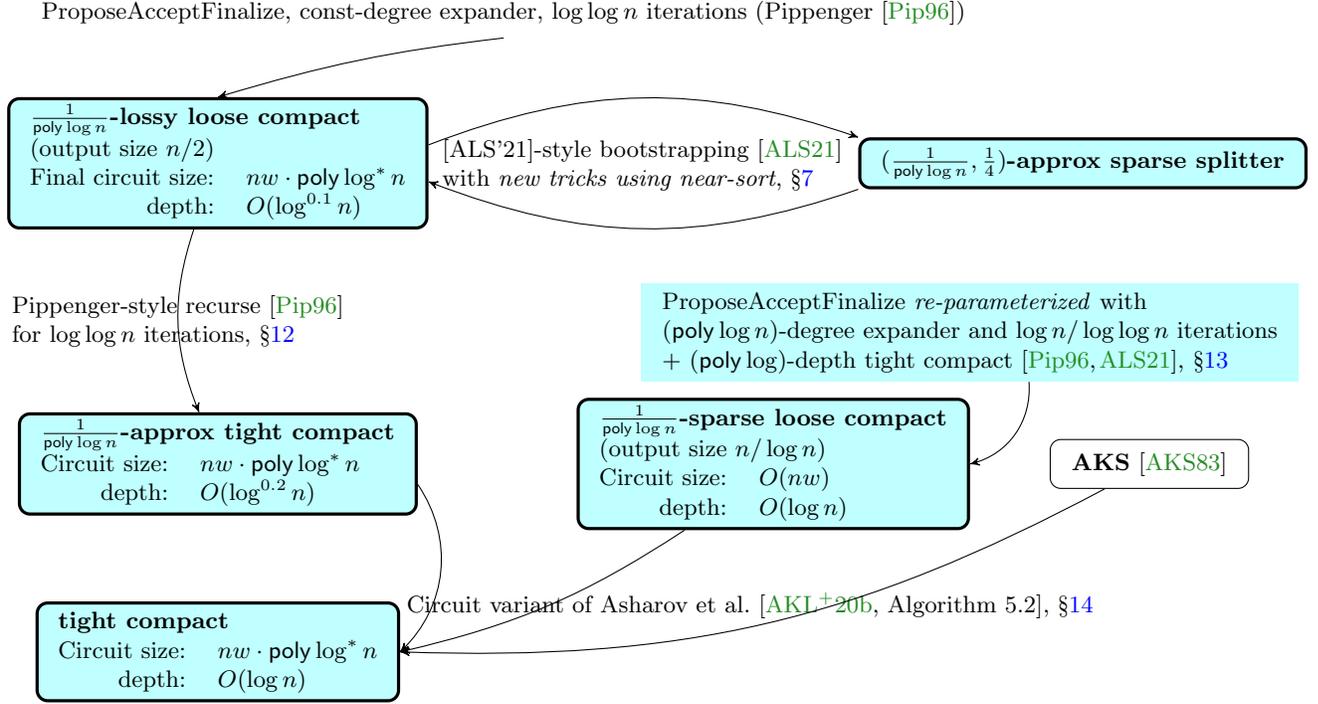
\begin{figure}[t]
\begin{center}
\begin{tikzpicture}[->,>=stealth']
\footnotesize

\definecolor{aqua}{rgb}{0.75, 1.0, 1.0}

\tikzstyle{state} = [rectangle,
           rounded corners,
           draw=black, very thick,
           minimum height=2em,
           inner sep=2pt,
           text centered,]

\tikzstyle{box} = [rectangle,
           minimum height=2em,
           inner sep=2pt,
           text centered,]

\node[box,
] (PAF) 
{
\begin{tabular}{l}
 ProposeAcceptFinalize, const-degree expander, 
 $\log\log n$ iterations (Pippenger~\cite{selfroutingsuperconcentrator})
\end{tabular}
};

\node[state,
  below of=PAF,
  node distance=2cm,
  anchor=east,
  xshift=-1cm,
  fill=aqua,
] (LOSSYLOOSE) 
{
\begin{tabular}{rl}
  \multicolumn{2}{l}{
    \textbf{$\frac{1}{\poly\log n}$-lossy loose compact}
  }\\
  \multicolumn{2}{l}{
    (output size $n/2$)
  }\\
  Final circuit size: & $nw \cdot \poly\log^* n$\\
  depth: & $O(\log^{0.1} n)$\\
\end{tabular}
};

\node[state,
  right of=LOSSYLOOSE,   
  node distance=11.5cm,  
  fill=aqua,
  ] (APPROXSPARSE) 
{
\begin{tabular}{l}
 \textbf{$(\frac{1}{\poly\log n}, \frac{1}{4})$-approx sparse splitter}\\
\end{tabular}
};

\node[state,
  below of=LOSSYLOOSE,
  node distance=4cm,
  fill=aqua,
] (APPROXTIGHT) 
{
\begin{tabular}{rl}
  \multicolumn{2}{l}{
    \textbf{$\frac{1}{\poly\log n}$-approx tight compact}
  }\\
  Circuit size: & $nw \cdot \poly\log^* n$\\
  depth: & $O(\log^{0.2} n)$\\
\end{tabular}
};

\node[box,
  below of=APPROXSPARSE,
  node distance=2.25cm,
  xshift=-1.5cm,
  fill=aqua,
] (PAFNEW) 
{
\begin{tabular}{l}
 ProposeAcceptFinalize \emph{re-parameterized} with \\
 ($\poly \log n$)-degree expander and $\log n / \log\log n$ iterations\\
  + ($\poly\log$)-depth tight compact~\cite{selfroutingsuperconcentrator,soda21},  
  \S\ref{sec:sparselc}\\
\end{tabular}
};

\node[state,
  below of=PAFNEW,
  node distance=1.75cm,
  anchor=east,
  fill=aqua,
] (SPARSELOOSE) 
{
\begin{tabular}{rl}
  \multicolumn{2}{l}{
    \textbf{$\frac{1}{\poly\log n}$-sparse loose compact}
  }\\
  \multicolumn{2}{l}{
    (output size $n/\log n$)
  }\\
  Circuit size: & $O(nw)$\\
  depth: & $O(\log n)$\\
\end{tabular}
};

\node[state,
  thin,
  right of=SPARSELOOSE,
  node distance=5cm,
] (AKS) 
{
\begin{tabular}{l}
 \textbf{AKS}~\cite{aks}\\
\end{tabular}
};

\node[state,
  below of=APPROXTIGHT,
  node distance=2.5cm,
  fill=aqua,
] (TIGHT) 
{
\begin{tabular}{rl}
  \multicolumn{2}{l}{
    \textbf{tight compact}
  }\\
  Circuit size: & $nw \cdot \poly\log^* n$\\
  depth: & $O(\log n)$\\
\end{tabular}
};

\path 
(PAF.south)      edge[bend right=5] (LOSSYLOOSE.north)
(LOSSYLOOSE.5)  edge[bend left=20]
    node[anchor=south,below,yshift=-0.4cm]{
    \begin{tabular}{l}
    [ALS'21]-style bootstrapping~\cite{soda21} \\
    with \emph{new tricks using near-sort}, \S\ref{sec:llc}
    \end{tabular}}
  (APPROXSPARSE.north west)
(APPROXSPARSE.south west)      edge[bend left=20] (LOSSYLOOSE.355)
(PAFNEW.320)      edge[bend left=40] (SPARSELOOSE.east)
(LOSSYLOOSE)      edge[bend right=20]
  node[]{
    \begin{tabular}{l}
    Pippenger-style recurse~\cite{selfroutingsuperconcentrator} \\
    for $\log\log n$ iterations, \S\ref{sec:approxtc}  
    \end{tabular}
  } 
  (APPROXTIGHT)
(APPROXTIGHT)      edge[bend left=40]  (TIGHT.east)
(SPARSELOOSE)      edge[bend left=10] 
  node[anchor=west,xshift=-2cm]{
    Circuit variant of Asharov et al.~\cite[Algorithm 5.2]{paracompact}, \S\ref{sec:tc}
  }
  (TIGHT.east)
(AKS)      edge[bend left=15] (TIGHT.east)
;
\end{tikzpicture}
\end{center}
\caption{Blueprint of our tight compaction circuit. 
The light blue parts represent our new abstractions or constructions.
Rounded boxes represent the intermediate abstractions. }
\label{fig:compaction}
\end{figure}

\paragraph{Intermediate abstractions.}
We rely the following intermediate abstractions --- 
among them, the approximate splitter 
and the sparse loose compactor are new abstractions.
\weikai{removed ``the approximate tight compaction''}

\begin{itemize}[leftmargin=6mm,itemsep=1pt,topsep=1pt]
\item {\it Lossy loose compaction}.
Let $\alpha \in (0, 1)$. 
Given an array of length $n$ 
containing at most $n/128$ {\it real} 
elements and all remaining elements are {\it fillers},
an $\alpha$-lossy loose compactor 
compresses the array by a half, losing at most $\alpha n$  
real elements in the process.

\item {\it Approximate splitter}.
Let $\beta \in (0, 1/4]$ and let $\alpha \in (0, 1)$.
An $(\alpha, \beta)$-approximate splitter solves the following problem:
we are given an input array 
containing $n$ elements
each marked with a 1-bit label indicating whether the element
is {\it distinguished} or not. It is promised that at most $\beta \cdot n$
elements in 
the input are distinguished.
We want to output a permutation 
of the input array, 
such that at most $\alpha n$ distinguished elements
are not contained in the first $\floor{\beta n + n/64}$
positions of the output.

\item {\it Approximate tight compaction}.
Let $\alpha \in (0, 1)$.
Given an input array containing $n$ elements,
each with a $1$-bit key,
an $\alpha$-approximate tight compactor outputs
a permutation 
of the input array, 
such that at most $\alpha \cdot n$ elements 
in the output are misplaced. Here,
the $i$-th element in the output is 
said be misplaced iff 
its key disagrees with what the $i$-th smallest key in the input array.

\item {\it Sparse loose compactor}.
Let $\alpha \in (0, 1)$. An array of length $n$ is said to 
be $\alpha$-sparse
if there are at most $\alpha n$ real elements in it 
and the rest are all fillers.
A sparse loose compactor performs exactly the same task
as a lossy loose compactor, except that 
1) the input array is promised to be 
$1/(\log n)^{C_\smalltriangledown}$-sparse
for some fixed constant $C_\smalltriangledown > 8$;
2) we now want to compress the array to $n/\log n$ length; 
and 3)  
we do not want to lose any real elements in the compressed output array. 
\end{itemize}

\paragraph{Blueprint of our compaction circuit.}
The entire construction is fairly sophisticated. To help understanding,
we depict the blueprint in Figure~\ref{fig:compaction}.
We explain the high-level ideas below and 
give a more detailed exposition in the remainder of this section.

\begin{enumerate}[leftmargin=5mm,itemsep=1pt]
\item 
{\it Using a repeated bootstrapping 
trick to get a 
$1/\poly\log n$-lossy loose compactor of linear size and sub-logarithmic depth.}
The work of Asharov et al.~\cite{soda21} suggests a {\it repeated
bootstrapping} idea that 
upgrades an inefficient compaction circuit of size $O(n w + n\log n)$ 
to an efficient compaction circuit of size  
$O(nw \cdot \poly(\log^*n - \log^*(w+k)))$.
An inefficient compaction circuit 
of size $O(n w + n \log \log n)$ can be obtained
from Pippenger~\cite{selfroutingsuperconcentrator} with some additional work ---
but the resulting circuit has poly-logarithmic depth $d = \poly\log n$. 
\elaine{double check}
Moreover, if the inefficient compaction circuit has depth $d$,
then the resulting efficient compaction circuit  
would have depth $d^{\log (\log^* n)}$.
To avoid this depth blowup, 
we apply their repeated bootstrapping 
idea not directly to (tight) compaction, 
but to the weaker abstraction,
$1/\poly\log n$-lossy loose compactor. 
Given Pippenger's ideas~\cite{selfroutingsuperconcentrator}, we can construct
an initial inefficient $1/\poly\log n$-lossy loose compactor
that is $O(n w + n \log\log n)$ in size, 
\elaine{is this correct?}\weikai{loglog consistent with depth}%
but whose depth is only $O(\log \log n)$.
In this way, even after this repeated bootstrapping, we can cap the depth
at $O(\log^{0.2} n)$.
Some technicalities arise to adapt Asharov et al.'s repeated bootstrapping idea 
to our case: we will need to make use of a new
abstraction called an approximate splitter which we define above;
additionally, we also need to make use of $\epsilon$-near-sorters
in our new repeated bootstrapping.
We defer the technical details to Sections~\ref{sec:llc}, \ref{sec:splitterfromlc}, \ref{sec:lcfromsplitter},
and \ref{sec:optllc}. 

\item 
{\it Upgrade to a $1/\poly\log n$-approximate tight compactor of linear size
and sub-logarithmic depth.}
Pippenger~\cite{selfroutingsuperconcentrator} 
showed how to get a tight compactor from a loose compactor.
We will use Pippenger's ideas to obtain 
a $1/\poly\log n$-approximate tight compactor 
of linear size and sub-logarithmic depth from the aforementioned  
$1/\poly\log n$-lossy loose compactor.
To avoid depth blowup, we need to stop the Pippenger-style recursion early,
and cap it at $O(\log \log n)$ iterations. 
The resulting tight compactor is not perfect and still has $1/\poly\log n$ 
fraction of misplaced elements, partly
because the loose compactor 
we started with is lossy, and partly because we stopped the recursion early.
We defer the details to Section~\ref{sec:approxtc}. 
\ignore{
Once we obtain 
$1/\poly\log n$-lossy loose compactor of linear size (ignoring $\poly\log^*$ terms) 
and sub-logarithmic depth, 
we can upgrade it to a $1/\poly\log n$-approximate tight compactor 
also of linear size and sub-logarithmic depth, using
techniques from Pippenger~\cite{selfroutingsuperconcentrator}
stopping the recursion early at $O(\log \log n)$ iterations.
}

\item 
{\it Constructing a sparse loose compactor.}
Given the $1/\poly\log n$-approximate tight compactor, our remaining job 
is to correct the remaining $1/\poly\log n$ fraction of errors.
To correct the remaining errors, 
we are again inspired by Pippenger~\cite{selfroutingsuperconcentrator}:
we correct a large fraction of these errors using expander
graphs, and then extract the remaining errors using a {\it loose compactor} 
(into a half-sized array), 
we then recurse on the extracted 
array to correct all remaining errors, and route the corrected 
elements back into their original positions.

The problem is that the recursive extraction incurs another logarithmic factor in depth
while a (non-lossy) loose compactor already takes logarithmic depth.
Fortunately, a crucial observation that helps here is
that the errors are {\it sparse} --- specifically, at most  
$1/\poly\log n$ elements are errors. 
We show how to construct
a (non-lossy) sparse loose compactor, that compresses a sparse array 
containing at most $n/\poly\log n$ real elements, to 
a size of $n /\log n$ (rather than just half of $n$), 
without losing any real elements in the process ---
notice that we avoid the recursive extraction and the extra logarithmic factor in depth.
To accomplish this, 
our key insight is to use a {\it $\poly\log n$-degree expander graph}
rather than a constant-degree expander as in Pippenger's approach --- 
this way, the depth of the overall extraction is reduced to
$O(\log n)$, and this is critical in achieving
small depth.
To make this idea fully work, we additionally need a slightly inefficient 
tight compactor circuit  
that achieves linear work and polylogarithmic depth --- we show how to get such
a tight compactor with some modifications to 
Asharov's construction~\cite{soda21}. 
We defer the details to Section~\ref{sec:sparselc}.

\item {\it Putting everything all together.} 
Finally, as mentioned, given the 
$1/\poly\log n$-approximate tight compactor of linear size
and sub-logarithmic depth, 
a sparse loose compactor of linear size and logarithmic depth, 
we can leverage Pippenger's ideas~\cite{selfroutingsuperconcentrator}
to get a tight compaction circuit of linear size and logarithmic depth.
Moreover, this construction also makes use of the AKS sorting network~\cite{aks}.
We defer the details to Section~\ref{sec:tc}.
\end{enumerate}

\elaine{TODO: fix all the broken refs}

\subsection{Sorting Circuit for Short Keys}
With our algorithms in 
Sections~\ref{sec:roadmap-opram} and \ref{sec:roadmap-compact-circ}, 
and with some extra work, one can 
get a sorting circuit for short keys
that satisfies Theorem~\ref{thm:intro-sort-circ}.
The technicalities 
here 
are mostly how to efficiently convert some of the algorithmic building
blocks used by the oblivious PRAM sorting algorithm  
to the circuit model. We defer the details to the subsequent formal sections. 


\subsection{Additional Related Work}

Since the landmark AKS result~\cite{aks},
various works have attempted to simplify it and/or reduce
the concrete constants~\cite{akspaterson,aksjoel,zigzag}. Notably,
the recent ZigZag sort of Goodrich (STOC'14)~\cite{zigzag} 
took a rather different approach than the original AKS; unfortunately,
its depth is asymptotically worse than AKS.
None of these works achieved 
theoretical improvements over AKS, and all of them 
considered the comparator-based model.

As mentioned, the special case of sorting 1-bit keys 
is also called compaction, 
which is trivial to accomplish on a (non-oblivious) RAM.
A line of work was concerned about the circuit complexity
of compaction~\cite{yaoselectnet,sortminmem,jimboselect,pippengerselect};
but all earlier works focused on the comparator-based model.
Due to the famous 0-1 principle described as early 
as in Knuth's textbook~\cite{knuthbook}, there is an 
$\Omega(n \log n)$ lower bound for compaction  
with comparator-based circuits.
Several works have considered compaction 
in other incomparable models of computation 
as explained below (but none of them 
easily translate to a circuit result).
Leighton et al.~\cite{leightonselection}
show how to construct comparison-based, {\it probabilistic} 
circuit families for compaction, 
with $O(n \log \log n)$ comparators; again, here we require that 
for every input, an overwhelming
fraction of the circuits in the family  
can give a correct result on the input. 
\ignore{
One interesting observation 
is that although deterministic, comparator-based {\it selection} circuits
have a $\Omega(n \log n)$
lower bound~\cite{leightonselection,sortminmem}, probabilistic circuit families 
do not inherit the same lower bound.
\elaine{why?}%
}%
Subsequent works~\cite{odsmitchell,osortsmallkey} have improved 
Leighton's result by removing
the restriction that the circuit family must be parametrized
with the number of 0s
without increasing the asymptotical overhead.
These works also imply that compaction 
can be accomplished with in $O(n \log \log n)$
time on a randomized Oblivious RAM~\cite{odsmitchell,osortsmallkey}.

Asharov et al.~considered how to accomplish  
compaction on {\it deterministic} Oblivious RAMs in linear work~\cite{optorama},
but their construction is sequential in nature.
Their work was subsequently extended~\cite{paracompact}
to a PRAM setting achieving optimality in both work and depth;  
but a counterpart of such an optimal compaction result in the circuit 
model was not known earlier.
Dittmer and Ostrovsky improve
its concrete constants by introducing randomness back~\cite{DittmerOstrovsky20}.
Interestingly, linear-time oblivious compaction played a 
pivotal role in the construction of an optimal Oblivious RAM (ORAM) 
compiler~\cite{optorama}, 
a machine that translates a RAM program to a functionally-equivalent 
one with oblivious access patterns. 
Specifically, earlier ORAM compilers relied on oblivious sorting which requires 
$\Omega(n\log n)$ time 
either assuming the indivisibility model~\cite{osortsmallkey} or 
the Li-Li network coding conjecture~\cite{sortinglbstoc19}; whereas more recent
works~\cite{panorama,optorama} observed that with 
with a lot more additional work, we could replace oblivious sorting
with the weaker compaction primitive.

Besides Pippenger's self-routing
super-concentrator\cite{selfroutingsuperconcentrator},
Arora, Leighton, and Maggs~\cite{alm90}
considered a self-routing {\it permutation} network. Their construction
does not accomplish sorting.
Further, converting their non-blocking network to a permutation
circuit would require at least $\Omega(n \log^2 n)$ gates~\cite{bmm}.
Pippenger's work~\cite{selfroutingsuperconcentrator}
adopted some techniques from the Arora et al.~work~\cite{alm90}.



\section{Nearly Orderly \Splitter}
\label{sec:segmenter}

\subsection{Notations}
\paragraph{Array and multiset notations.}
Whenever we say an {\it array}, we mean an ordered array.
Throughout the paper, we may assume that the array
to be sorted has length $n$ that is a power of $2$ ---
in case not, we can always round it up to the nearest 
power of $2$ by padding $\infty$ elements, incurring
only constant blowup in array length and consuming at most one additional bit in 
terms of key length.

Given an array $A$, the notation ${\it mset}(A)$ denotes
multiset formed by elements in $A$.
Suppose that 
$A$ and $A'$ are two arrays, 
then $A || A'$ denotes the array formed
by concatenating $A$ and $A'$. 
For $m \in \N$, we use the notation $[m] := \{1, 2, \ldots, m\}$.
Suppose that $1 \leq s \leq t \leq |A|$, 
we use the notation $A[s:t]$ 
to denote the length-$(t-s+1)$ segment of the array 
$A$ from $s$-th element to the $t$-th element. We define 
the short-hand notations $A[:t] := A[1:t]$ and $A[s:] := A[s:|A|]$.
Unless otherwise noted $\log$ means 
$\log_2$.

\paragraph{Binary tree notations.}
Given a complete binary tree with $t$ levels, 
the {\it level} 
of a node is the number of edges from the root to the node.
For example, 
the root is at 
{\it level} $0$;
and the
leaves are at 
{\it level} $t-1$.

The {\it tree distance} of two nodes in a binary tree  
is the length of the shortest path between them.


\subsection{Definitions}

\paragraph{``Misplaced'' elements w.r.t.~segments.}
Let $A$ be an array of length $n$, and let $[s, t] \subseteq [n]$
be a contiguous sub-range of $[n]$.
The number of {\it ``misplaced''} elements in the segment $A[s:t]$, 
denoted $\err(A[s:t])$, is defined as 
as the number of elements residing in $A[s:t]$, however, if $A$ were to be
sorted, ought not to be in $A[s:t]$.\footnote{
  In this section, we abuse ``misplaced'' and refer to the elements
  that are in the wrong \emph{segments} (instead of the wrong \emph{position}
  in the remaining of this work).
}
More formally, 
$${\err}(A[s:t]) = 
\left|{\it mset}(A[s:t])
- {\it mset}(B[s:t]) \right|.
$$
where $B = {\sf sorted}(A)$ denotes the sorted
version of $A$, and recall that ${\it mset}(A[s:t])$ 
denotes the multiset formed by elements in $A[s:t]$.
As a special case, if 
${\it mset}(A[s:t]) = {\it mset}(B[s:t])$, then 
the ${\err}(A[s:t]) = 0$.
(Notice that )

\paragraph{Nearly orderly \splitter.}
We now define $(\eta, p)$-orderliness and 
an $(\eta, p)$-orderly \splitter.

\begin{Definition}[$(\eta, p)$-orderly]
Let $m$ and $p$ be positive integers, and suppose that $n = mp$. 
Write an array $A$ of length $n$ as 
the concatenation of $p$ equal-sized segments:
$A = A_1 || A_2 || \ldots || A_p$.
We say that $A$   
is $(\eta, p)$-orderly iff  
for each $i \in [p]$, ${\err}(A_i) \leq \eta \cdot |A_i|$.
\end{Definition}

\begin{Definition}[$(\eta, p)$-orderly \splitter]
Let $n := mp$.
An $(\eta, p)$-nearly orderly \splitter (for $n$)
is a circuit that takes an array $A$ of length $n$, and outputs  
a permutation of $A$   
which is $(\eta, p)$-orderly.
\end{Definition}

\subsection{Construction}
We rely on ideas from the AKS sorting network~\cite{aks} to construct 
a nearly orderly \splitter.

At a very high level, 
the AKS algorithm proceeds 
in $O(\log n)$ cycles to sort a length-$n$ input array.
During each cycle $t$,
\begin{enumerate}
\item The algorithm partitions the current array into 
a number of disjoint 
{\it intervals} which are not necessarily equally sized. Henceforth
the term {\it interval} refers to a contiguous subarray.
The number of intervals is geometrically growing with each cycle.
\item 
The algorithm then partitions the intervals  
into groups, and each group contains
a disjoint (but not necessarily contiguous) 
subset of the intervals. It then sorts each group and writes
the sorted array back in place. Further, all the groups
are sorted in parallel.
The above partitioning and sorting procedure  
is repeated three times (and each time the partitioning may be different), 
and then the algorithm enters the next cycle.
\end{enumerate}
At the end of $\log n$ cycles,  
the input array is guaranteed to be sorted.
\elaine{it's exactly log n ?}

It turns out if we repeat the AKS algorithm for 
$6k < \log n$ cycles and stop,
the resulting array will satisfy $(2^{-8k}, 2^{3k})$-orderliness. 
For completeness, below we describe the algorithm 
where we essentially perform AKS for $6k < \log n$ cycles,
we then rely on a technical lemma proven in the AKS 
paper~\cite{aks} to prove that the resulting array  
is indeed $(2^{-8k}, 2^{3k})$-orderly.

\subsubsection{Preliminaries}
\label{sec:segmenter-prelim}
Recall that we would like to 
partition the current array into a number of intervals in each 
AKS cycle $t$.  To understand how the intervals a defined, we will
first define a helper data structure called a $t$-AKS-tree.

\begin{figure*}
\begin{center}
\includegraphics[width=0.75\textwidth]{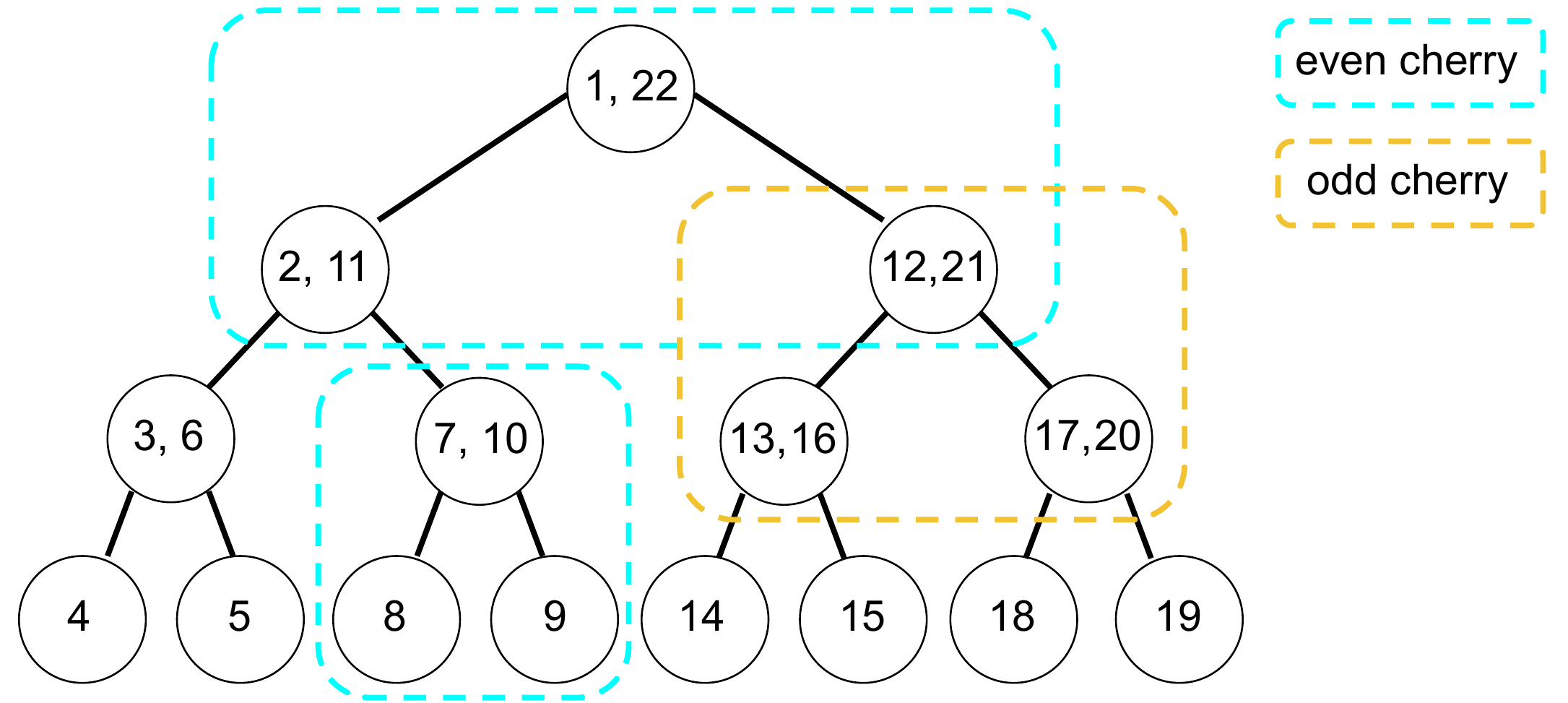}
\end{center}
\caption{$t$-AKS-tree for $t = 3$.}
\label{fig:akstree}
\end{figure*}

\paragraph{$t$-AKS-tree.}
An $t$-AKS-tree is binary tree containing a total of $t + 1$ levels
numbered $0, 1, \ldots, t$, respectively.
Henceforth define $M(t) := 3 \cdot 2^{t} - 2$.
All tree nodes receive either one or two labels from the range 
$[M(t)]$; further, each label is given to exactly  
one tree node. The labeling scheme satisfies the following constraints:
\begin{enumerate}[itemsep=1pt]
\item 
Each leaf receives one label from the range $[M(t)]$;
and each non-leaf node receives two labels from the same range. 
\item For each internal node, every  
label in its right subtree is strictly greater
than every label in its left subtree.
\item 
The set of labels assigned to 
each subtree is a contiguous sub-range 
$[s, t] \subseteq [M(t)]$; and further, the
minimum $s$ and maximum $t$ of the range 
are assigned to the root of the sub-tree.
\end{enumerate}
One can check that the above set of constraints 
uniquely define the labeling on the tree nodes.
In Figure~\ref{fig:akstree}, we give an example of 
a $t$-AKS-tree where $t = 3$.


\paragraph{$t$-AKS-intervals.}
Given an array $A$ of length $n$, we can divide
it into $M(t)$ intervals called $t$-AKS-intervals, i.e., 
$A := A_1 || A_2 || \ldots || A_{M(t)}$, where the length
of each interval $A_i$ 
depends on which level 
the label $i$ shows up in the $t$-AKS-tree.
At a very high level, 
the length geometrically decreases by a factor
of approximately $\gamma := 16$ as the label $i$'s level 
becomes smaller.
\elaine{double-check if 16 works}

We now define the lengths of each $t$-AKS-interval more formally,
following the same approach as in the original AKS 
paper~\cite{aks}.
We first define the following 
numbers
for $t = 1, 2, \ldots, \log n$, and for  $\ell = 1, 2, \ldots, t$:
$$
\begin{array}{l}
\qquad \qquad X_t(\ell) := \left\lfloor \frac{1}{2\gamma}\cdot n \cdot 2^{-t} \cdot 
\gamma^{\ell-t}\right\rfloor, 
\quad Y_t(\ell) := \sum_{j = 1}^\ell X_t(j)  
\end{array}
$$


Let $\ell \in [0, t-1]$ and $j \in [1, 2^\ell]$.
Suppose that the $j$-th node at level $\ell$ in the $t$-AKS-tree 
have the two labels $i$ and $i'$. 
Then, the lengths of the 
two intervals $A_i$ and $A_{i'}$ are defined
as follows.
\[
|A_i| :=  
\begin{cases}
X_t(\ell + 1)   & \text{if $j$ is odd} \\
Y_t(\ell+1) & \text{o.w.}
\end{cases}
\quad 
\text{ and } 
\quad 
|A_{i'}|
 := X_t(\ell+1) +  Y_t(\ell+1) - |A_i|
\]

Finally, in the last level $\ell = t$ in the $t$-AKS-tree, 
each node has only one label. Suppose that the $j$-th node's
label is $i$, then 
the length of the interval $A_i$ is $|A_i| := 
n \cdot 2^{-t} - Y_t(t)$.

\ignore{
Then, the length of $A_i$ is defined as follows:
\[
|A_i| := 
\begin{cases}
\frac{n}{2^{t-1}} \cdot \frac{1}{16^{h(i)}} \cdot (1-\frac{1}{16}) & 
\text{if $h(i) \neq t-1$}  \\
\frac{n}{2^{t-1}} \cdot \frac{1}{16^{h(i)}} & \text{o.w.}  
\end{cases}
\]
One can verify 
that that $\sum_{i \in [M(t)]} |A_i| = n$.
}

\elaine{TO DO: double check this
}

\begin{fact}[Group $t$-AKS-intervals into equally sized segments]
As mentioned, assume that $n := |A|$ is a power of 2. 
Fix any non-leaf level $\ell \in \{0,1,\dots,t - 1\}$ in a
$t$-AKS-tree, 
we can partition $A$ into $2^\ell$ equally sized segments as follows
(where equally sized means that every segment contains
the same number of elements): 
\label{fct:equalseg}
\end{fact}
\begin{mdframed}
\begin{enumerate}[leftmargin=3mm]
\item 
Initially, 
for every node $v$ in the $t$-AKS tree, 
${\labels}(v)$
is defined to be the set of the original 
labels of $v$. 
Specifically, for every non-leaf node $v$,  
${\labels}(v)$ has two labels, 
and for every 
leaf node $v$,  
${\labels}(v)$ has only one label.

\item 
\label{stp:group_seg:assign}
For level $i = 0$ to $\ell-1$, 
for every node $v$ in level $i$ of the $t$-AKS-tree, 
\begin{enumerate}
\item 
let $S \subseteq {\labels}(v)$ be the subset 
of node $v$'s labels 
smaller than every label in ${\labels}(v.{\sf LeftChild})$, and 
let $S' := {\labels}(v) \backslash S$.

\item 
let 
${\labels}(v.{\sf LeftChild}) := {\labels}(v.{\sf LeftChild}) \cup S$;
\item 
let ${\labels}(v.{\sf RightChild}) := {\labels}(v.{\sf RightChild}) \cup S'$;
\end{enumerate}
\item 
For every node $v$ in level $\ell$ of the $t$-AKS-tree:
all $t$-AKS-intervals whose corresponding labels 
are in ${\sf Subtree}(v)$ (including $\labels(v)$)
are grouped together and called one segment, where 
${\sf Subtree}(v)$ means the subtree rooted at $v$.
\end{enumerate}
\end{mdframed}

\ignore{
\begin{fact}
As mentioned, assume that $n := |A|$ is a power of 2. 
Fix any non-leaf level $\ell \in \{0,1,\dots,t - 1\}$ in a
$t$-AKS-tree, and let $L := 2^\ell$.  
Let $i_1 < i_2 < \ldots < i_{L}$ be the ordered set of the larger label (i.e., second label) 
for each odd node in level-$j$ for all $j \le \ell$ of the $t$-AKS-tree 
(e.g., $i_1$ is the second label of 1st node in level-$\ell$,
$i_2$ is that of 1st in level-$(\ell-1)$, $i_3$ is that of 3rd in level-$\ell$,
$i_4$ is that of 1st in level-$(\ell-2)$, and so on). Define $i_0 := 0$.
Then, we can partition $A$ into $L$ equally sized segments,
where  for $j \in [L]$, the $j$-th segment
contains the $t$-AKS-intervals
$A_{i_{j-1} + 1}, A_{i_{j-1} + 2},  \ldots, A_{i_{j}}$.
\ignore{
if we partition $A$ 
into $2^\ell$ equally sized segments, 
then no $t$-AKS-interval in $A$ belongs to two or more segments.
}
\elaine{weikai: is this what you mean?}
\end{fact}
}

The example in Figure~\ref{fig:akstree},
shows the segments for $t = 3$ and level $\ell = 2$. 
In this case, if we partition $A$ into $4$ equally sized segments,
then the segments are:
\[
(A_1, A_2, \ldots, A_6), 
(A_7, A_8, \ldots, A_{11}), 
(A_{12}, A_8, \ldots, A_{16}), 
(A_{17}, A_{18}, \ldots, A_{22}).
\]

\begin{proofof}{Fact~\ref{fct:equalseg}}
This fact is implicit in the AKS paper~\cite{aks}, we prove it explicitly below.
Alternatively, we can consider the following equivalent 
variant of the above algorithm: 
in Step~\ref{stp:group_seg:assign}, \elaine{hard coded ref} 
we do not stop at the end of the $(\ell-1)$-th iteration,
but continue all the way to level $t-1$.
At this moment, for each node $v$ in level 
$\ell$ of the tree: we group together the labels on all leaf nodes 
in ${\sf Subtree}(v)$ --- their corresponding $t$-AKS-intervals
will form one segment.

It is not hard to show through induction that at the end
of the iteration $i = 0$ to $t - 1$ in Step~\ref{stp:group_seg:assign}, \elaine{hard coded refs} 
each node $v$ in level $i + 1$ of the $t$-AKS tree receives
a set of labels from its parent 
which correspond to a total of $Y_t(i)$ elements.
Therefore, at the end of iteration $0$ to $t-2$, 
for each node $v$ in level $i + 1$, 
its labels correspond to 
a total of consecutive $Y_t(i + 1)$ elements.
At the end of the final iteration $i = t-1$, 
each leaf node's labels correspond to $n \cdot 2^{-t}$ consecutive elements.
\end{proofof}


\ignore{to remark: 
In the above intervals, for any level $\ell \in \{0,1,\dots,t\}$ and any $j\in[2^\ell]$,
the location $j\cdot n \cdot 2^{-\ell}$ must be a boundary between 
adjacent intervals $A_i$ and $A_{i+1}$ such that the levels of $A_i$ and $A_{i+1}$ are at most $\ell$
(except for the case $j=2^\ell$ where $A_i$ is the last interval).
This property will be used to split the array into $2^\ell$ equal-sized segments.
}

\paragraph{Even and odd cherries.}
In a binary tree, a {\it cherry} is
defined to be a parent node and its two children.
The {\it even (or odd, resp.) cherries} are those whose parents reside
at an even (or odd, resp.) {\it level}. 

Given $A := A_1 || \ldots || A_{M(t)}$  
written as $t$-AKS-intervals, 
we define 
${\sf EvenCherries}(A_1 || \ldots || A_{M(t)})$
to be a set of {\it disjoint groups} of  
$t$-AKS-intervals. 
Specifically, 
${\sf EvenCherries}(A_1 || \ldots || A_{M(t)})$
is of the form 
$\left\{ G_1, G_2, \ldots, G_d\right\}$
where $d$ is the number of even cherries in a 
$t$-AKS-tree, 
and 
for $i \in [d]$, each group $G_i$ 
corresponds to a distinct even cherry in a $t$-AKS-tree, i.e., 
$G_i$ is one of the following two forms depending
on whether the even cherry touches the leaf level:
\begin{itemize}[itemsep=0pt]
\item 
either $G_i:= A_{j_1} || \ldots || A_{j_6}$ 
where $j_1 < j_2< \ldots < j_6$, 
and moreover, $j_1, \ldots, j_6$ correspond
to the labels of an even cherry in the $t$-AKS-tree that does not 
involve the leaf level;
\item 
or $G_i:= A_{j_1} || \ldots || A_{j_4}$ 
where $j_1 < j_2< \ldots < j_4$.
and moreover, $j_1, \ldots, j_4$ correspond
to the labels of an even cherry 
in the $t$-AKS-tree, involving the leaves this time.
\end{itemize}

The notation ${\sf OddCherries}$ is similarly defined but
replacing ``even'' with ``odd''. 

\elaine{todo: depict cherry in figure}


\paragraph{$\epsilon$-near-sorter.}
Let $\epsilon \in (0, 1)$ be a constant. 
An array $A$ of length $n$ is said to be $\epsilon$-near-sorted, 
iff the following holds
for any $1 \leq k \leq n$:
\begin{enumerate}[leftmargin=5mm,itemsep=1pt]
\item  
$A[1:k+\epsilon n]$ 
contains at least $(1-\epsilon)k$ of the $k$ smallest elements in $A$;
\item 
$A[n-k - \epsilon n + 1 : n]$
contains at least $(1-\epsilon)k$ of the $k$ largest elements in $A$.
\end{enumerate}
In the above, we use the following notations
to deal with boundary conditions: 
for $i > n$, $A[1:i] := A[1:n]$; and for $i < 1$, $A[i:n] := A[1:n]$.

\elaine{move this to building blocks}

An $\epsilon$-near-sorter (for $n$) is a 
circuit containing $O(n)$ comparators 
and of constant depth (dependent on $\epsilon$)
that permutes any input array of length $n$ into one
that is  $\epsilon$-near-sorted.
Earlier works have shown how to construct
such and $\epsilon$-near-sorter using expander graphs~\cite{aks}.
\elaine{double check this defn}



\subsubsection{Nearly Orderly \Splitter Construction}
\label{sec:segmenterdetail}

Our nearly orderly \splitter construction is described below:

\begin{mdframed}
\begin{center}
{\bf Nearly orderly \splitter}
\end{center}

\paragraph{Input:} An array ${\bf I}$ whose length $n$ is a power
of $2$.

\paragraph{Parameters:} 
Let 
$\epsilon \in (0, 1)$ be a sufficiently small constant, and 
let 
$\czigzag  > 1$ be a sufficiently large constant.

\paragraph{Algorithm:}  \ \\

\noindent
Let $A := {\bf I}$ be the current array.

\vspace{5pt}
\noindent
For $t = 1, 2, \ldots, \min(6k, \log n)$: \quad {\tt // $t$-th AKS cycle}
\elaine{change constant}
\begin{enumerate}[leftmargin=6mm,topsep=5pt,itemsep=0pt]
\item 
{\bf Divide into $t$-AKS-intervals.}
Write $A := A_1 || A_2 || \ldots || A_{M(t)}$, where 
$A_1, A_2, \ldots, A_{M(t)}$ are $t$-AKS-intervals.

\item 
Repeat the following $\czigzag$ times:
\begin{enumerate}[leftmargin=3mm,itemsep=1pt]

\item 
\label{step:zig}
{\bf Near-sort even cherries.}
In parallel, apply an $\epsilon$-near-sorter  
to each group of intervals contained in  
${\sf EvenCherries}(A_1 || \ldots ||$ $A_{M(t)})$,
and the result is written back in place (i.e., into the 
$t$-AKS-intervals'  
original positions within $A$).

\item 
{\bf Near-sort odd cherries.}
In parallel, apply an $\epsilon$-near-sorter  
to each group of intervals contained in 
${\sf OddCherries}(A_1 || \ldots ||$ $A_{M(t)})$,
and the result is written back in place.
\end{enumerate}

\item  {\bf Near-sort even cherries.} Repeat Step~\ref{step:zig}
one final time.
\end{enumerate}

\paragraph{Output:} Finally, output $A$.
\end{mdframed}

\begin{theorem}[$(2^{-8k}, 2^{3k})$-orderly-\splitter]
Let $\epsilon \in (0, 1)$
be a suitably small constant, and 
let $\czigzag$
be a suitably large constant.
Then, the above construction is a 
$(2^{-8k}, 2^{3k})$-orderly-\splitter; moreover,
it can be implemented as a comparator-based circuit with 
$O(n) \cdot \min (6k, \log n)$ comparators and 
of $O(1) \cdot \min (6k, \log n)$ depth.
\label{thm:segmenter}
\end{theorem}
\begin{proof}
The proof is presented in Section~\ref{sec:segmenterproof}.
\end{proof}

\subsection{Proof of Theorem~\ref{thm:segmenter}}
\label{sec:segmenterproof}
The size and depth bounds follow in a straightforward manner.
Below we focus on proving that the algorithm
gives a $(2^{-8k}, 2^{3k})$-orderly-\splitter.
To prove this, 
we need to rely on a 
technical lemma proven by Ajtai et al.~\cite{aks}.

\begin{lemma}[Technical lemma due to Ajtai et al.~\cite{aks}] 
Fix any arbitrarily small constant $\alpha \in (0, 1)$
such that $(16 \gamma)^2 \cdot \alpha^2 < 1$.  
\elaine{todo: double check this expression}
There exist a suitably small constant $\epsilon \in (0, 1)$
and a suitably large constant $\czigzag > 1$, 
such that in the above construction, at the end of each cycle 
$t \leq \log n$, the following hold    
for any $t$-AKS-interval $A_i$ where $i \in [M(t)]$:
\begin{enumerate}[leftmargin=5mm]
\item[]
For $r \geq 1$,  $\err_r(A_i) < \alpha^{3r + 27} \cdot |A_i|$, 
where $\err_r(A_i)$ denotes the number of elements
actually in $A_i$, but if the array were sorted,
would land in a $t$-AKS-interval 
that is at a tree-distance at least $r$ away from 
the node labeled with $i$ in the $t$-AKS-tree.
\end{enumerate}

\label{lem:aks}
\end{lemma}

The above Lemma~\ref{lem:aks} is implied by the Theorem  
stated on page 7 of the original AKS paper~\cite{aks} --- we stated
the lemma slightly differently from the 
original AKS paper for our convenience. 
We now use Lemma~\ref{lem:aks} to prove Theorem~\ref{thm:segmenter}.

\medskip
\begin{proofof}{Theorem~\ref{thm:segmenter}}
Recall that $\gamma = 16$.
We will choose $\alpha$ such that 
$4 \cdot (16\gamma)^2 \cdot \alpha^2 = 1$, i.e., $\alpha = 
(32 \gamma)^{-1} = \frac{1}{2^9}$.
Moreover, suppose that we pick $\czigzag$ to be sufficiently large
and $\epsilon \in (0, 1)$ to be sufficiently small
such that Lemma~\ref{lem:aks} is satisfied.
We run the algorithm specified in Section~\ref{sec:segmenterdetail}
with the aforementioned parameters, and let $A$ be the output array.
Without loss of generality, we may assume that $6k < \log n$ 
since otherwise Ajtai et al.~\cite{aks}
proved that the outcome $A$ would be sorted, and this 
would be the easy case.

We now divide $A$ into $2^{3k}$ equally sized segments.
We can equivalently view the $2^{3k}$ equally sized segments as 
being created by the procedure specified in Fact~\ref{fct:equalseg},
where $\ell := 3k$.
Pick an arbitrary segment, say, the $i$-th segment, 
among the $2^{3k}$ equally sized segments.
Henceforth, let $v_{3k, i}$
denote the $i$-th node in level $3k$ of the $6k$-AKS-tree.

Due to the procedure specified in Fact~\ref{fct:equalseg}, 
we know that the $i$-th segment
consists of 
\begin{enumerate}[leftmargin=6mm]
\item 
all $6k$-AKS-intervals whose labels reside in ${\sf Subtree}(v_{3k, i})$
of the original $6k$-AKS-tree 
({\it not} of the tree output by the procedure in Fact~\ref{fct:equalseg});
\item 
a subset of the $6k$-AKS-intervals whose labels reside
in an ancestor node of $v_{3k, i}$ in the $6k$-AKS-tree. 
\end{enumerate}

For convenience, whenever we say the level of a $t$-AKS-interval,
we mean the level of its corresponding label
in the $t$-AKS-tree.
Let $S_{\ell}$ denote the 
the total length of all $6k$-AKS-intervals 
of level $\ell$
contained in the $i$-th segment.
It is not hard to see that for $\ell \in [0, 6k-1]$, 
$S_\ell \leq S_{\ell + 1}/8$, by the definition of the lengths
of the $t$-AKS-intervals.
\ignore{
Therefore, we have that 
\[
\frac{n}{2^{3k}} \leq \frac{1}{1-1/8}S_{6k} \leq 1.2 S_{6k}
\]
where $\frac{n}{2^{3k}}$ denotes the length of the $i$-th segment. 
}

For $\ell \in [3k, 6k]$, a $6k$-AKS-interval 
of level $\ell$ contained in the $i$-th segment
must have tree distance at least 
$\ell - 3k + 1$ from any $6k$-AKS-interval  
not contained in the $i$-th segment.

We use the term ``wrong elements'' to mean
elements that do not belong to 
the $i$-th segment if the array were sorted. 
Let $W_\ell$ denote the total number of wrong elements 
in some $6k$-AKS-interval of level $\ell$ in the $i$-th segment.
By Lemma~\ref{lem:aks}, we have that 
$$W_\ell \leq 
\alpha^{3(\ell-3k+1) + 27} \cdot S_\ell
\leq \alpha^{3(\ell-3k+1) + 27} \cdot 
\frac{S_{6k}}{8^{6k-\ell}}
$$
Therefore, we have that 

\begin{align*}
\frac{\sum_{\ell \in [3k, 6k]} W_\ell}{S_{6k}}
& \leq \alpha^{3 (3k+1) + 27} \cdot 
\left(1 + 
\frac{\alpha^{-3}}{8} 
+  
\left(\frac{\alpha^{-3}}{8}\right)^2 + \ldots
+
\left(\frac{\alpha^{-3}}{8}\right)^{3k} \right) \\
& \leq 
\alpha \cdot \alpha^{9k}
\cdot \left(\frac{\alpha^{-3}}{8}\right)^{3k} \cdot 2   
\leq 2^{-9k}  
\qquad \qquad \qquad \qquad \qquad \qquad \ (\star)
\end{align*}

Moreover, we have that 
\begin{align*}
\frac{\sum_{\ell \in [0, 3k-1]} S_\ell}{S_{6k}}
\leq \frac{1}{8^{3k+1}} 
\cdot \left( 1 + \frac{1}{8} + \ldots + \frac{1}{8^{3k-1}}\right) 
\leq \frac{1}{8^{3k+1}} \cdot 2 
\leq 2^{-9k}
\qquad \ \ (\star \star)
\end{align*}

Combining $(\star)$ and $(\star \star)$, 
we have that 
\[
\frac{\sum_{\ell \in [0, 6k]} W_\ell}{S_{6k}}
\leq 
2^{-9k} \cdot 2 \leq 2^{-8k}
\]
Since $S_{6k}$
is smaller than the total length of the $i$-th segment,
 we have that the fraction of ``misplaced'' elements 
of the $i$-th segment
must be upper bounded by $2^{-8k}$.
\end{proofof}

\ignore{
Suppose that in some cycle $t$, the 
intervals are denoted $I_1, \ldots, I_\ell$ from left to right. 
The algorithm performs the following during the cycle $t$
before moving onto cycle $t+1$:
\begin{enumerate}
\item Defines a partitioning of 
the intervals: $P_1 \cup P_2 \cup 
\ldots \cup P_d = \{I_1, \ldots, I_\ell\}$, 
where each partition $P_j$ contains a disjoint (but not necessarily contiguous) 
subset of of the intervals. 
\item 
For each partition $j \in \ell$, 
suppose that $P_j := \{I_{j_1}, \ldots, I_{j_k}\}$.

\end{itemize}
}



\section{Building Blocks for the Oblivious PRAM Model}
\label{sec:prambldgblock}

In this section, we present some building blocks 
that can be implemented
as deterministic, oblivious parallel algorithms.
This means that the algorithms' memory access patterns
are fixed a-priori and independent of 
the input (once we fix the input's length).

\paragraph{Compaction.}
Compaction (short for ``tight compaction'') 
solves the following problem: given an array 
in which every element is tagged with a 1-bit key, move
all elements tagged with $0$ to the 
front of the array, and move elements tagged with $1$ to the end. 
Asharov et al.~\cite{paracompact}
showed a deterministic algorithm that 
obliviously compacts any array containing $n$
elements each of which encoded as $\ell$ words; and their algorithm 
achieves $O(\ell \cdot n)$ 
total work and $O(\log n)$ depth.

Furthermore, their compactor supports a ``reverse routing'' capability.
Specifically, their compactor can be thought of a 
network consisting of $O(n)$ selector gates 
of depth $O(\log n)$, with $n$ inputs and $n$ outputs. 
Each selector gate takes in a 1-bit 
flag and two input elements that are $\ell$ words long, 
and the flag is used to decide
which of the two input elements to output. 
The first phase of their algorithm, 
takes $O(n)$ work and $O(\log n)$ depth: it  
computes on the elements' 1-bit keys, 
and populates all selector gates' 1-bit flags. 
The second phase of their algorithm then routes 
the input elements to the output layer over this selector network.
This takes $O(\ell \cdot n)$ work and $O(\log n)$ depth.
Since each selector
gate can remember its flag, it is possible to later on route
elements in the reverse direction, from the output layer back to the input
layer.

We stress that Asharov et al.~\cite{paracompact}'s oblivious compaction
algorithm is {\it not stable}, i.e., it does not preserve
the relative order 
of elements with the same key
as they appeared in the input array.
In fact, Lin, Shi, Xie~\cite{osortsmallkey} showed that this is inherent: 
any oblivious algorithm in the indivisibility 
model that achieves {\it stable} compaction must incur $\Omega(n \log n)$ work.
Here, an algorithm in the indivisibility 
model is one that 
does not perform encoding or computation
on the elements' payload strings.
Afshani et al.~\cite{lbmult}
shows that the $\Omega(n \log n)$ lower bound holds
for oblivious, deterministic stable compaction even 
without the indivisibility requirement, but instead assuming
that the Li-Li network coding conjecture holds~\cite{lilinetcoding}.

\paragraph{Distribution.}
Distribution solves the following problem.
We are given an input array ${\bf I}$ of length $n$ in which
each element carries a $w$-bit payload and a 1-bit label
indicating whether the element is {\it real} or a {\it filler}. 
Additionally, we are given a bit-vector ${\bf v}$ of length $n$,
where ${\bf v}[i]$ indicates whether the $i$-th output position
is available to receive a real element.
It is promised that the number of available positions
is at least as many as the number of real elements in ${\bf I}$.
We want to output an array ${\bf O}$ 
such that the multiset of real elements in ${\bf O}$
is the same as the multiset of   
real elements in ${\bf I}$, and moreover
if ${\bf O}[i]$
contains a real element, then it must be that ${\bf v}[i] = 1$, i.e.,
only available positions in the output array ${\bf O}$ can receive
real elements.

The following algorithm accomplishes the aforementioned distribution task
using compaction as a building block:
\begin{mdframed}
\begin{center}
{\bf Distribution}
\end{center}
\begin{enumerate}[leftmargin=5mm,itemsep=1pt]
\item 
Let ${\bf X}$ be an array 
in which all payloads are fillers and each ${\bf X}[i]$ 
is marked with the label ${\bf v}[i]$.
\item 
Now, apply tight compaction to ${\bf X}$ routing all 
entries with $1$-labels to the front, and all entries
with $0$-labels to the end. 
\label{step:routex-distr}
\item 
Apply another instance of tight compaction to the input array
${\bf I}$ routing all real elements to the front and all filler
elements to the end; let the outcome be ${\bf I}'$.
\item 
Next, reverse-route the array ${\bf I'}$ by 
reversing the routing decisions made 
in Step~\ref{step:routex-distr}, 
and output the result.
\end{enumerate}
\end{mdframed}

Therefore, oblivious distribution can be accomplished  
with the same asymptotical overhead 
as oblivious compaction.
Just like compaction, here it also makes sense to consider a 
reverse-routing 
capability of our distribution algorithm.

\paragraph{All prefix sums.}
Given an array $A$ of length $n$, an all-prefix-sum algorithm 
outputs the prefix sums of all $n$ prefixes, 
i.e., $A[:1]$, $A[:2]$, $\ldots$, and $A[:n]$, respectively.
It is promised that the sum of the entire array $A$ can be stored
in $O(1)$ memory words.
It is well-known that there is a deterministic, 
oblivious algorithm that computes all prefix sums in $O(n)$
work and $O(\log n)$ depth~\cite{jajabook}.

\paragraph{Generalized binary-to-unary conversion.}
Imagine that there are $n$ receivers where the $i$-th receiver
is labeled with an indicator bit ${\bf x}[i]$.
We are given an integer $\ell \in \{0,1,\dots,n\}$ expressed
in binary representation, 
and we want to output an array of $n$ bits where the 
$i$-th bit represents
the bit received by the $i$-th receiver. 
We want that the first $\ell$ receivers marked with $1$ 
receive $1$, and all other receivers marked with $1$ 
receive $0$. The receivers marked with $0$ may receive an arbitrary bit.
Note that in the special case 
that all receivers are marked with $1$, then the problem
boils down to converting an integer $\ell \in \{0, 1, \ldots, n\}$ 
expressed in binary representation to a corresponding unary string.

The generalized binary-to-unary conversion problem 
can easily be solved by invoking an all-prefix-sum computation
on an oblivious parallel RAM, taking $O(n)$
total work and $O(\log n)$ depth\footnote{We
explicitly differentiate 
the generalized binary-to-unary conversion
from the all-prefix-sum because it is more convenient
later for our circuit-model results.
In the circuit model,  
the generalized binary-to-unary conversion
can be solved with a circuit $O(n)$ in size and $O(\log n)$
in depth, whereas all-prefix sum requires 
a circuit $O(n \log n)$
in size and $O(\log n)$ in depth (even when the input 
$A$ is a bit array). 
}.


\paragraph{Sorting elements with ternary keys.}
We will need a linear-work, and logarithmic-depth
oblivious algorithm to sort an input array with ternary keys, as stated
in the following theorem.

\ignore{this alg needs to be changed for circuit model: use
b-to-u conversion instead}

\ignore{
We show that constant-length keys can be sorted
in linear work and logarithmic depth. 
We use a two-parameter recursion technique described by  
Lin, Shi, and Xie~\cite{osortsmallkey} 
and Asharov, Lin, and Shi~\cite{soda21}.
}

\begin{theorem}[Sort elements with ternary keys]
There exists a deterministic, oblivious algorithm
that can sort any input array 
$A$ containing $n$ elements 
each with a key from the domain $\{0, 1, 2\}$
in $O(n)$ work and $O(\log n)$ depth.
\label{thm:ternarysort}
\end{theorem}
\begin{proof}
\ignore{The algorithm and analysis are described in Asharov et al.~\cite[Remark 5.3]{paracompact}.
It is a variant of compaction and thus inherits the capability of ``reverse routing''
and takes the same asymptotic complexity in both work and depth}
Consider the following algorithm:

\ignore{TODO: in the circuit later, we need to change this algo}
\begin{mdframed}
\begin{center}
{\bf Ternary-key sorting}
\end{center}
\begin{enumerate}[leftmargin=5mm,itemsep=1pt]
\item
For each key $b \in \{0, 1, 2\}$, 
let $L_b, U_b \in [n]$ denote the 
starting and ending index for $b$ if the array $A$ were to be fully sorted.
We can accomplish  
this by counting for each $b \in \{0, 1, 2\}$ the total number of occurrences
of $b$ in $A$. 

\item 
Relying on oblivious distribution three times, we can 
route all elements with the key $b$ 
to the positions $[L_b, U_b]$ of the output array. 
Output the result.
\end{enumerate}
\end{mdframed}

One can easily verify that the above algorithm sorts
the input array $A$ with ternary keys, and 
moreover, the algorithm completes in $O(n)$ 
total work and $O(\log n)$ depth.
\end{proof}

Just like compaction, here it also makes sense to consider a 
reverse-routing 
capability of our ternary-key sorting algorithm.


\ignore{
Note also that 
in the above ternary sorting algorithm,
Steps 3 to 5 \elaine{hard coded} 
can be viewed as a network of $O(n)$ selector gates 
of depth $O(\log n)$, and if each selector gate 
remembered its routing decision, 
we can reverse 
route the output elements in the reverse direction back
to the input array. 
}

\section{Sorting Short Keys on an Oblivious PRAM}

\ignore{
We describe 
some new building blocks we will need; and in this section,
we focus on the oblivious PRAM model of computation.
}

Throughout, we assume that the array $A$ to be sorted
contains elements that are (key, payload) pairs. 
A key can be expressed in $k$ bits, 
and the entire element can fit in $O(1)$ memory words. 
\ignore{
Whenever we say ``an algorithm ${\sf Alg}$ 
can sort $k$-bit keys'',
we mean that ${\sf Alg}$ 
can correctly sort any input 
array containing elements whose keys are expressed in $k$ or 
fewer bits. 
}

\subsection{Slow Sorter and Slow Alignment}


\paragraph{Slow sorter.}
We show that there is a slow sorter that sorts
an array containing $n$ elements with $k$-bit keys
in $O(2^k \cdot n)$  work and $O(\log n)$ depth.

\begin{theorem}[Slow sorter]
Let $K := 2^k$. There exists a deterministic, oblivious algorithm, 
henceforth 
denoted ${\bf SlowSort}^K(A)$, 
that can correctly sort any input array $A$
of length $n$ and containing elements with $k$-bit keys in 
$O(nK)$ total work and $O(k + \log n)$ depth.
\label{thm:slowsort}
\end{theorem}
\begin{proof}
Consider the following algorithm:
\begin{mdframed}
\begin{center}
${\bf SlowSort}^K(A)$
\end{center}
\paragraph{Input:} An array $A$ whose length $n$ is a power of $2$.
Every element in $A$ has a $k$-bit key chosen 
from the domain $[0, K-1]$.

\paragraph{Algorithm:}
\begin{enumerate}
\item 
\label{stp:slowsort:count}
For each $u \in [0, K-1]$ in parallel, count
the number of occurrences of the key $u$ in $A$, and let 
$c_u$ be this count.
Using an all-prefix-sum algorithm, compute
$s_u := \sum_{i \in [0, u-1]} c_i$ for every $u \in [1, K-1]$, and 
define $s_0 := 0$.

\item 
\label{stp:slowsort:copy}
Make $K$ copies of the array $A$, denoted $B_0, \ldots, B_{K-1}$, 
respectively.
In each $B_u$, the elements whose keys are not $u$
are replaced with {\it filler}.
\item 
\label{stp:slowsort:ternary}
For $u \in [0, K-1]$: 
\begin{enumerate}[leftmargin=5mm,itemsep=1pt]
\item In array $B_u$, 
for the first $s_u$ filler elements, treat their keys as $-\infty$;
for every other filler element, treat 
its the key as $\infty$. This can be accomplished 
by invoking a generalized 
binary-to-unary conversion algorithm.
\item 
Invoke oblivious sorting for ternary keys to sort $B_u$.
In the resulting array denoted $B'_u$, the elements whose keys 
are equal to $u$ will 
appear at positions $s_u + 1, \ldots, s_{u+1}$.
\end{enumerate}

\item 
\label{stp:slowsort:populate}
In parallel, 
populate the $i$-th element in the output array for every $i \in [n]$ as
follows:  
select the element whose key is 
within the range $[0, K-1]$ 
among the elements $B'_0[i], B'_1[i], \ldots, B'_{K-1}[i]$.
The selection can be accomplished by aggregating over a binary tree
whose leaves are $B'_0[i]$, $B'_1[i]$, $\ldots$, $B'_{K-1}[i]$.
\end{enumerate}
\end{mdframed}

One can easily verify that the above algorithm indeed
correctly sorts in the input array. Moreover, its total work
is bounded by $O(nK)$ and its depth is bounded
by $O(k + \log n)$. Specifically, for the depth, the $O(k)$ part
upper bounds the depth 
of the first step 
that computes the all-prefix-sum of $K$ elements as well as the 
last step where we select among $K$ elements; 
and the $O(\log n)$ part is an upper bound on
the depth of the generalized binary-to-unary computation, as well
as the ternary-key sorting.
\end{proof}

\begin{Remark}[Reverse routing]
In the above ${\bf SlowSort}$ algorithm, there is a way
to reverse-route elements in the output array 
back into their original positions in the input.
Suppose that during Step~\ref{stp:slowsort:populate}, \elaine{hard coded}
we remember for each position of the output array, 
from which array $B'_u$ it received an element, 
In this way, we can reverse Step 4 and reconstruct
the arrays $B'_0, \ldots, B'_{K-1}$ from the output.
Now, we can reverse the routing decisions of the ternary sorter  
to reconstruct the arrays $B_0, \ldots, B_{K-1}$. 
For each $i \in [n]$, 
there is only one $B_u$ such that $B_u[i]$ 
is not a filler element, and this element $B_u[i]$ 
will be routed back to the $i$-th position of the input 
array.
Clearly, the reverse routing does not cost more
than the forward direction in terms of work and depth.
\end{Remark}

\paragraph{Slow alignment.}
We define a variant of the slow sorter algorithm, called
${\bf SlowAlign}^{K, K'}(A)$.
${\bf SlowAlign}^{K, K'}$ receives an input array $A$ in which
every element $A[i]$
is not only tagged with a key $A[i].{\sf key}$ 
from the domain $[0, K-1]$,
but also an index $A[i].{\sf idx}$ which can be expressed in 
$k' := \log K'$ 
bits.
As before, we assume that each element, including
its tagged key and index, 
can fit in $O(1)$ words. We want to output a permutation 
of $A$ such that in the ordering
the keys become consistent with the ordering of the indices
in the input array. 
In other words, suppose that $B$ is the output array in which
each element is tagged with only a key, 
then, 
\begin{equation}
\forall i, j \in [n] \text{ and } i \neq j: 
(A[i].{\sf idx} < A[j].{\sf idx})
\Longrightarrow
(B[i].{\sf key} \leq B[j].{\sf key})
\label{eqn:aligned}
\end{equation}

\begin{theorem}[${\bf SlowAlign}^{K,K'}$]
There is a deterministic, oblivious 
${\bf SlowAlign}^{K, K'}(A)$
algorithm that solves the above alignment problem
and outputs an array $B$ that is a permutation of the input  
array $A$ satisfying 
Equation~\eqref{eqn:aligned}; and moreover,
the algorithm
takes $O((K+K') n)$
total work and $O(\log K + \log K' + \log n)$
depth where $n$ is the length of the input array.
\label{thm:slowalign}
\end{theorem}
\begin{proof}
The oblivious algorithm ${\bf SlowAlign}^{K, K'}$ is described below:
\begin{mdframed}
\begin{center}
${\bf SlowAlign}^{K, K'}(A)$
\end{center}
\paragraph{Input:} An array $A$ of length $n$, 
and for every $i \in [n]$, the element $A[i]$ is tagged
with a key $A[i].{\sf key}$
and an index $A[i].{\sf idx}$.

\paragraph{Algorithm:}

\begin{enumerate}[leftmargin=5mm,itemsep=1pt]
\item 
\label{step:slowalign:sortkey}
Call ${\bf SlowSort}^K(A)$
using the ${\sf key}$
field as the key to sort the array $A$, and let $B$ 
be the outcome.
\item 
\label{step:slowalign:sortidx}
Call ${\bf SlowSort}^{K'}(A[1].\idx, A[2].\idx, \ldots, A[n].\idx)$ 
and let $\idx_1, \ldots, \idx_n$
be the resulting ordered list of indices.
\item 
\label{step:slowalign:reverse}
Reverse route $B$ by reversing 
the routing decisions made in 
Step~\ref{step:slowalign:sortidx}.
\end{enumerate}
\end{mdframed}
Correctness is easy to verify.
For performance bounds, 
observe that 
Step~\ref{step:slowalign:sortkey} takes $O(nK)$ work and $O(\log n + \log K)$ depth, 
Step~\ref{step:slowalign:sortidx} takes $O(nK')$ work and $O(\log n + \log K')$ depth, 
Step~\ref{step:slowalign:reverse}'s work and depth are not more than Step~\ref{step:slowalign:sortidx}.
\end{proof}

\subsection{Finding the Dominant Key}
Let $\epsilon \in (0, 1/2)$.
We say that an array $A$ of length $n$ is $(1 - \epsilon)$-uniform iff 
except for at most $\epsilon n$ 
elements, all other elements in $A$ have the same key --- henceforth
this key is said to be the {\it dominant} key.

We want an algorithm that can 
correctly identify the dominant
key when given an input array  
$A$ that is $(1-\epsilon)$-uniform.
If the input array $A$ is not $(1-\epsilon)$-uniform, the
output of the algorithm may be arbitrary. 

\begin{theorem}[${\bf FindDominant}$ algorithm]
Suppose that $n > 2^{k + 7}$ and moreover $n$ is a power of $2$. 
\elaine{note this condition}
Let $A$ be an array containing $n$ elements
each with a $k$-bit key, and suppose
that $A$ is $(1-2^{-8k})$-uniform.
There is a deterministic, oblivious algorithm 
that can correctly identify 
the dominant key given any such $A$;
and moreover, the algorithm 
requires 
$O(n)$ total work and 
$O(k + \log n)$ \elaine{check}
depth.
\label{thm:finddominant}
\end{theorem}
\begin{proof}
Let $K := 2^k$. We can call ${\bf FindDominant}(A, K, n)$ which
is defined below --- since $n > 2^{k + 7}$ 
and $n$ is a power of $2$, one can verify that
every recursive call to ${\bf FindDominant}(B, K, n)$
will have an input $B$ whose size is a multiple of $8$. 
\begin{mdframed}
\begin{center}
${\bf FindDominant}(B, K, n)$
\end{center}
\begin{enumerate}[leftmargin=5mm,itemsep=1pt]
\item 
\label{stp:finddom:base}
If $|B| \leq n/K$, then call ${\bf SlowSort}^K(B)$
and output either one of the median keys in the sorted array. 
Else, continue with the following.
\item 
\label{stp:finddom:group}
Divide the array into columns of size $8$.
Obliviously sort each column using AKS~\cite{aks}; 
and let $a_i, b_i$ be the two median elements in column $i$, 
i.e., the 4th and 5th smallest elements.
\item
\label{stp:finddom:recurse}
Output ${\bf FindDominant}\left(\{(a_i, b_i)\}_{i \in [|B|/8]}, K, n\right)$. 
\end{enumerate}
\end{mdframed}
Henceforth, any element whose key differs
from the dominant key is said to be a minority element.
In the above algorithm, for each column, 
if we want to make sure that both median elements
are minority, we must 
consume at least $5$ minority elements. 
If we want to make sure that one of the two 
median elements is minority, 
we must consume at least $4$ minority elements.

Suppose that 
$B$ is $(1-\mu)$-uniform.
In the array $\{(a_i, b_i)\}_{i \in [|B|/8]}$, 
the number of elements that are minority is upper bounded by 
$\frac{2 \mu \cdot |B|}{5}$;
the fraction of elements that  
are minority 
is upper bounded by 
\[
\frac{{2 \mu \cdot |B|}/{5}}{|B|/4}
= 8\mu/5
\]

After $D := \ceil{\log_4 K}$ recursive calls, the 
algorithm will encounter the base case, 
invoke ${\bf SlowSort}$ and output the median.
At this moment, the fraction of minority elements 
is upper bounded by 
\[
2^{-8k} \cdot \left(\frac{8}{5}\right)^D 
\leq 1/2 
\]
Therefore, outputting the median at this moment 
will give the correct result.
%
\end{proof}

\subsection{Sorting a Nearly Orderly Array}

Recall that using a nearly ordered segmenter
(see Section~\ref{sec:segmenter})
to partially sort the input array, such that the result
is $(\eta, p)$-nearly ordered. 
We will show that there is an efficient oblivious algorithm
that can correct the remaining errors and fully sort the array.

\begin{theorem}
Suppose that $n > 2^{4k + 7}$.  
There is a deterministic, oblivious algorithm
that fully sorts an $(2^{-8k}, 2^{3k})$-orderly array 
in $O(n)$ total work and $O(\log n)$ depth. 
\label{thm:corrector}
\end{theorem}

\begin{proofof}{Theorem~\ref{thm:corrector}}
We consider the following algorithm.
\begin{mdframed}
\begin{center}
{\bf Fully sort an $(\eta, p)$-orderly array}
\end{center}

\paragraph{Input and parameters.}
The input is an array $A$ whose length $n$ is a power of $2$.
$A$ is promised to be $(\eta, p)$-orderly   
for $\eta := 2^{-8k}$ and $p := 2^{3k}$, where $k$ is a natural
number such that $6k < \log n$.
Henceforth we write $A$ as $A := A_1 || A_2 || \ldots || A_p$
where all $A_i$s are equally sized segments. 
Let $K := 2^k$ and let $m := n/p$.

\paragraph{Algorithm.}
\begin{enumerate}[leftmargin=5mm,itemsep=1pt]
\item 
\label{step:gb}
For each segment $i \in [p]$ in parallel: 
\begin{enumerate}[leftmargin=5mm,itemsep=0pt]
\item 
Call ${\sf key}^*_i := {\bf FindDominant}(A_i, K, |A_i|)$; 
\item 
\label{step:count}
Count the number of occurrences of ${\sf key}^*_i$ in $A_i$
to decide if $A_i$ is $(1-\eta)$-uniform.
\elaine{note: dont use k as key}
\item 
\label{step:labelgb}
If $A_i$ is $(1-\eta)$-uniform, 
mark the following elements as 
``{\tt misplaced}'': 
1) all elements whose key differ from ${\sf key}^*_i$; and 2) 
exactly $\ceil{\eta m}$ number of elements with the dominant 
key ${\sf key}^*_i$. \elaine{todo: explain how this is done in 
remark later}

Else, mark all elements in $A_i$ as  
``{\tt misplaced}''.

\end{enumerate}

\item 
All elements in $A$ calculate which segment it falls in ---
note that all elements can learn its position within $A$  
through an all-prefix-sum calculation,
and the segment number can be calculated from the element's 
position within $A$.

Call oblivious compaction to move all elements in $A$
marked with ``{\tt misplaced}'' to the front of the array, and
all other elements to the end; all elements carry their segment number 
in the process. 
Let the outcome be called $X$.
\label{step:compact-correction}

\item 
\label{step:slowalign-correction}
Call ${\bf SlowAlign}^{K, K^2}(X[1: 3n/K^2])$ on the first
$3n/K^2$ elements of $X$, where the first $2k$-bits of 
each element's segment number 
is used as the $\idx$ field in the 
${\bf SlowAlign}$ algorithm.

\item 
\label{step:reverseroute-correction}
Invoke the reverse routing algorithm of the 
compactor in Step~\ref{step:compact-correction} on the outcome of 
the previous step, and let 
$Y$ be the outcome.

\item 
We will now divide $Y$ into $K^2$ 
super-segments instead where each super-segment
is composed of $K$ original segments.
Write $Y := Y_1 || Y_2 || \ldots || Y_{K^2}$ as the concatenation of $K^2$ 
equally sized super-segments.
For each $i \in [K^2]$: 
check if $Y_i$ has multiple keys;
if so, label the super-segment as ``{\tt multi-key}''.
\label{step:multikey-correction}

\item 
Invoke an oblivious compaction algorithm 
to move all the super-segments marked ``{\tt multi-key}'' to the front
of the array (here 
the compaction algorithm treats each super-segment as an element). 
Let the outcome be $Z$.
\label{step:compact2-correction}


\item 
Now, for each of the beginning $K$ super-segments of $Z$ in parallel 
(where each super-segment is $n/K^2$ in size), 
use ${\bf SlowSort}^K$ to sort within each super-segment.

\label{step:sortwithin-correction}
\item
Finally, reverse route the outcome of the previous step
by reversing the decisions made in Step~\ref{step:compact2-correction},
and output the result.
\label{step:final-correction}
\end{enumerate}
\end{mdframed}

\begin{Remark}
In the above algorithm, Steps~\ref{step:count} and \ref{step:labelgb}
can be performed obliviously as follows:
\end{Remark}
\begin{itemize}[leftmargin=5mm,itemsep=1pt]
\item 
The counting in Step~\ref{step:count} can be performed
by aggregating over a binary-tree of 
depth $O(\log n)$.
\item 
If the segment $A_i$ is not $(1-\eta)$-uniform:
all elements in $A_i$ label itself as ``{\tt misplaced}''
(else all elements in $A_i$ pretend to write a label for obliviousness).
\item 
Else, every element whose key differs from the dominate key 
${\sf key}^*_i$
marks itself ``{\tt misplaced}''; moreover, using  
a generalized binary-to-unary conversion algorithm, 
the first $\ceil{\eta m}$ elements
with the dominant key ${\sf key}^*_i$ 
label itself
as ``{\tt misplaced}''. All remaining elements pretend
to write a label for obliviousness.
\end{itemize}

\paragraph{Correctness.}
We will now argue correctness of the above algorithm, 
i.e., that the result is fully sorted
as long as the input array is $(2^{-8k}, 2^{3k})$-orderly.
Since there are at most $K$ distinct keys, 
it must be that in a fully sorted array, 
at most $K$ out of the $p = 2^{3k} = K^3$ segments have multiple keys,
all remaining segments must have a single key.
Since the input array is $(2^{-8k}, 2^{3k})$-orderly, 
it means that in the input array, all but $K$ segments  
must be $(1-\eta)$-uniform.

For any $(1-\eta)$-uniform segment in the input array,  
if we extract from it all elements whose keys differ 
from the dominant key, as well as at least $\ceil{\eta m}$ number of elements
with the dominant key,  
then all remaining elements must belong to this segment 
if the array were fully sorted.
In Step~\ref{step:labelgb}, we label all such 
elements as ``{\tt misplaced}''; as well as all elements
in segments that are not $(1-\eta)$-uniform.
The total number of elements marked as ``{\tt misplaced}''
is upper bounded by 
\[
K \cdot m + 2\eta m (K^3 - K) \leq n/K^2  + 2 \eta n   
= n/K^2 + 2n/K^8
\leq 3n/K^2
\]

Therefore, after Step~\ref{step:compact-correction}
effectively $X[1:3n/K^2]$ contains 
all elements marked ``{\tt misplaced}''
as well as some additional elements that we want to extract,
such that all remaining elements belong to their segment.
Suppose that $i_1, i_2, \ldots, i_{K^2}$  
number of elements from each super-segment are contained
in $X[1:3n/K^2]$.  
In Step~\ref{step:slowalign-correction}
and Step~\ref{step:reverseroute-correction},
we move the smallest $i_1$ extracted elements back
to the first super-segment, then next smallest  
$i_2$ extracted elements to the second super-segment, and so on.
In the outcome of Step~\ref{step:reverseroute-correction},
every element must belong to the correct super-segment.

At this moment, we only need to sort the super-segments that 
are multi-keyed. The total number of multi-keyed super-segments is at most $K$.
This is accomplished as follows: Step~\ref{step:compact2-correction}
moves all multi-keyed super-segments to the front, and then 
sorts within each of the first $K$ super-segments.
Finally, Step~\ref{step:final-correction}
reverse routes all the super-segments back to their original positions.

\paragraph{Performance bounds.}
Since by assumption, $n > 2^{4k + 7}$, then 
the length of each segment $m := n/2^{3k} > 2^{k+7}$, and therefore
other assumption of 
Theorem~\ref{thm:finddominant}
is satisfied and we can use
Theorem~\ref{thm:finddominant} to characterize the performance of the 
${\bf FindDominant}$ step.
Steps~\ref{step:compact-correction}
and \ref{step:compact2-correction}
each incurs $O(n)$ work and $O(\log n)$ depth.
Step~\ref{step:slowalign-correction}
incurs $O(3n/K^2 \cdot K^2) = O(n)$
work and $O(\log n + k)$ depth.
Step~\ref{step:sortwithin-correction}
incurs $O(K \cdot (n/K^2) \cdot K) = O(n)$
work and $O(\log n + k)$ depth.
The costs of all other steps are upper bounded by $O(n)$
and $O(\log n + k)$ too.
\end{proofof}

\subsection{Sorting Short Keys on an Oblivious PRAM}

Now, we can put everything together and obtain an oblivious 
parallel algorithm
that sorts an input array with short keys.

\begin{theorem}[Restatement of Theorem~\ref{thm:intro-sort-opram}]
There exists a deterministic oblivious parallel algorithm 
that sorts any input array 
containing $n$ elements each with a $k$-bit key 
in $O(n) \cdot \min(k, \log n)$
total work and $O(\log n)$ depth, assuming that each element
can be stored in $O(1)$ words.
\label{thm:mainsort-opram}
\end{theorem}
\begin{proof}
If $n \leq 2^{4k+7}$, we can just run AKS which takes
$O(n \log n)$ total work an $O(\log n)$.
Else, if $n > 2^{4k+7}$, 
we can accomplish the task with the following algorithm.

\begin{mdframed}
\begin{center}
{\bf Sorting short keys on an oblivious PRAM}
\end{center}

\paragraph{Input.}
An array $A$ of length $n$ each with a $k$-bit key
and a payload string. We assume that $n > 2^{4k+7}$ 
and moreover each element can be stored in $O(1)$ memory words.

\paragraph{Algorithm.}
\begin{enumerate}[leftmargin=5mm,itemsep=1pt]
\item 
Apply the $(2^{-8k}, 2^{3k})$-orderly segmenter  
construction of Theorem~\ref{thm:segmenter} to the input array $A$,
the outcome is a permutation of $A$ that is $(2^{-8k}, 2^{3k})$-orderly.
\item 
Apply the algorithm of Theorem~\ref{thm:corrector}
to correct the remaining errors and output the fully sorted result.
\end{enumerate}
\end{mdframed}
Given Theorems~\ref{thm:segmenter}
and \ref{thm:corrector}, it is not hard to see that the algorithm
takes $O(nk)$
work and $O(\log n)$ depth.
\end{proof}


\elaine{shall we get rid of generalized boolean gates}

\ignore{TO ADD: we shall first assume that selector is depth 1. 
at the end, we will actually compute all selectors in parallel.}

\section{Building Blocks for the Circuit Model}
\label{sec:circuitbldgblock}

\subsection{Our Operational Circuit Model}
Our result will be stated 
using the standard circuit model of computation~\cite{savagebook}
where the circuit consists of AND, OR, and NOT gates;
and moreover each AND and OR gate has  
constant fan-in and constant fan-out.

For convenience, we shall use an operational
model that consists of generalized boolean gates and (reverse) selector gates.
A {\it generalized boolean gate}  \elaine{do we need this?}%
has constant fan-in and constant fan-out, and implements  
any truth table between the inputs and outputs.
A {\it $w$-selector gate}
is a selector gate that 
takes in a 1-bit flag and two $w$-bit payload strings, and 
outputs one of the two payload strings determined by the flag.
A {\it reverse selector} gate is the opposite. 
A $w$-reverse selector gate takes one 
element $x$ of bit-width $w$ and a flag $b \in \{0, 1\}$ as input and 
outputs $(m, 0^w)$ if $b=0$ and $(0^w,m)$ if $b=1$. 
In our construction later, we will often use
reverse selector gates to preform ``reverse routing'', where we reverse the 
routing decisions 
made by earlier selector gates. 
Henceforth in the paper whenever we count 
selector and reverse selector gates, we do not 
distinguish between them and count
both towards selector gates. 

We say a circuit is in the \emph{indivisible} model if and only if the input 
to the circuit 
consists of 
elements with $k$-bit keys and $w$-bit payloads,
and the circuit never performs boolean computation on the payload
strings, that is, the payload strings are only moved around
using 
$w$-selector gates.

\begin{lemma}[Technical lemma about our operational circuit model] 
In the indivisible model, any circuit with $n$ generalized  
boolean gates, $n'$ number of $w$-selector gates
and of depth $d$ can be implemented as a boolean circuit  
(having constant fan-in and constant fan-out)
of size $O(n + n' \cdot w)$
and depth $O(d + \log w)$.
\label{lem:operational}
\end{lemma}
\begin{proof}
Generalized boolean gates can be easily replaced with 
AND, OR, and NOT gates without incurring additional asymptotical overhead. 
The key is how to instantiate all the $w$-selector
gates without blowing up the 
circuit's depth by a $\log w$ {\it multiplicative} factor.

First, imagine we have a ``partial evaluation'' 
circuit where payloads are fake and of the form $0^w$.
In this way, we can implement every $w$-selector gate  
as a degenerate one that takes $O(1)$ depth, 
since the outputs are always $0^w$.
Evaluating this partial evaluation circuit will populate
the flags on all selector gates. 
Notice such partial evaluation relies on the circuit being indivisible and thus
populating a flag is independent of any result of any $w$-selector.

Since we are subject to constant fan-in and constant fan-out gates, 
to implement an actual selector gate will require replicating
the gate's flag $w$ times, and then use $w$ generalized boolean
gates, one for selecting each bit. 
After the partial evaluation phase, all selector gates
can perform this $w$-way replication in 
parallel, incurring an additive rather than multiplicative $\log w$ overhead. 
At this point, we can instantiate each $w$-selector gate
using one generalized boolean gate  
for each of the $w$ bits being selected.

Therefore, the total circuit size is $(n + n' \cdot w)$
and the depth is $O(d + \log w)$.
\end{proof}

\ignore{todo: remove the appendix on circuit model, move
reserve selector gates to this section}

We define some useful circuit gadgets below.


\paragraph{Comparator.}
A $k$-bit comparator takes two values each of $k$ bits,
and outputs an answer from a constant-sized result set such as $\{>, <, =\}$, or
$\{\geq, <\}$, or $\{\leq, >\}$.
Note that the outcome can be expressed as 1 to 2 bits.

\begin{fact}
A $k$-bit comparator can be implemented with a circuit 
with $O(k)$ generalized boolean gates
and $O(\log k)$ depth.
\label{fct:comparator}
\end{fact}

\paragraph{Delayed-carry representation and shallow addition.}
Adding two $\ell$-bit numbers in binary representation 
takes $O(\log \ell)$ depth.
We will later need adders that are constant in depth.
To do this, we can use a Wallace-tree-like trick
and adopt a delayed-carry representation of numbers.

We represent an $\ell$-bit number $v$ 
as the addition of two $\ell$-bit numbers, i.e., $v := x + y$.
Here, it must be that the sum $x + y \leq 2^\ell-1$ can still
be presented as $\ell$-bits; moreover, 
the delayed-carry representation of $v$ is not unique.
Given two $\ell$-bit numbers $v_1 := x_1 + y_1$ and $v_2 := x_2 + y_2$
in this delayed-carry representation, we can compute 
the $(\ell+1)$-bit number $v_1 + v_2$
as follows where the answer is also in delayed-carry representation:
\begin{enumerate}[leftmargin=5mm,itemsep=1pt]
\item 
Compute a delayed-carry representation of $x_1 + y_1 + x_2$,
and let the result be $x' + y'$.
This can be done by summing up the $i$-th bit of
$x_1$, $y_1$, and $x_2$ respectively, for each $i \in [\ell]$.
For each $i \in [\ell]$, the sum of the three bits can be expressed
as a 2-bit number, where the first bit becomes the $i$-th bit of 
$x'$ and the other bit becomes the $(i+1)$-st bit of $y'$.
\item 
Now, using the same method, compute and output a delayed-carry 
representation of $x' + y' + y_2$.
\end{enumerate}
\ignore{
we can compute two $(\ell + 2)$-bit numbers $x', y'$ 
such that $x' + y' = v_1 + v_2$, i.e., $x', y'$ is a delayed-carry 
representation of $v_1 + v_2$.  
Further, $x'$ and $y'$ can be computed  
}
The above can be accomplished 
with $O(\ell)$ generalized boolean gates and in $O(1)$ depth. 
Henceforth this is called a {\it shallow addition}.

\paragraph{Counting.}
We will need a simple circuit gadget that counts the number
of $1$s in an input array containing $n$ bits.

\begin{fact}
Given an input array containing $n$ bits, 
counting the number of 1s in the input
array can be realized with 
a circuit of size $O(n)$ 
and depth $O(\log n)$.
\label{fct:countprelim}
\end{fact}
\begin{proof}
We can use the algorithm in Fact 4.3 of Asharov et al.~\cite{soda21}, but use 
the delayed-carry representation of numbers, and  
replace all adders with shallow adders.
Essentially, the numbers are added over a binary tree, where
in the leaf level (also called the last level), every number is 
promised to be at most $1$-bit long; in the second to last level,
every number is promised to be at most $2$-bit long; and so on. 
In this way, the total circuit size for the entire tree of adders 
is $O(n)$.
At the end of the algorithm, we perform a final addition to convert
the delayed-carry representation of the answer to a binary representation.
\end{proof}


\paragraph{All prefix sums.}
We consider an all-prefix-sum circuit gadget, which upon receiving
an input $A$ containing $n$ non-negative integers, 
outputs the sums of all $n$ prefixes, that is,  $A[:1], A[:2], 
A[:3], \ldots, A[:n]$. 
It is promised that the sum of  
the entire array $A$ can be stored in $\ell$ bits.

\begin{fact}
For any $\ell \le n$, there is a circuit composed of at most $O(n \ell)$ 
generalized boolean gates
and of depth $O(\log n)$ that 
solves the aforementioned all-prefix-sum problem.
\label{fct:prefixsumcircuit}
\end{fact}
\begin{proof}
We can use the standard prefix sum algorithm, but represent
all numbers using the delayed-carry representation,
and use shallow addition which can be computed in constant depth. 

\begin{mdframed}
\begin{center}
${\bf AllPrefixSum}(A)$
\end{center}
\paragraph{Input:} An array $A$ containing $n$ bits, 
where $n$ is a power of $2$.
We assume that each bit $A[i]$ is 
represented in a delayed-carry representation
as the sum of $A[i]$ and $0$.

\paragraph{Algorithm:}
\begin{enumerate}[leftmargin=5mm,itemsep=1pt]
\item 
If $n = 1$, return the only element of $A$.
Else proceed with the following
\item 
Let $A'$ be the array of length $n/2$ containing sums of adjacent pairs in $A$.
$A'$ can be computed from $A$ 
using $n/2$ shallow additions.
\item 
Compute $S := {\bf AllPrefixSum}(A')$.
\item 
Compute the all-prefix-sum 
array for $A$ from $S$, filling the 
gaps by performing $n/2$ shallow additions.
\end{enumerate}
\end{mdframed}

If we run the ${\bf AllPrefixSum}$ algorithm using 
the delayed-carry representation, 
the outcome will be $n$ prefix sums where the $i$-th prefix
sum is expressed the sum of two numbers, $s_i$ and $t_i$. 
Finally, we can compute $s_i + t_i$ in parallel 
for all $i \in [n]$ in parallel, taking $O(\log\ell) \le O(\log n)$ depth.
The entire circuit for computing all $n$ prefix sums
takes $O(n \ell)$ generalized boolean gates 
and $O(\log n)$ depth.
\elaine{do we need concrete constants?}
\end{proof}

\ignore{
\paragraph{Find in array.}
Given an input array containing $n$ elements
each with a $k$-bit label and a $w$-bit payload,
given also a desired $k$-bit label $L^*$.
Find the first occurrence in the input array an element whose label is $L^*$
and output any fixed canonical value if not found.
This can be implemented 
by comparing each element with the desired label $L^*$ 
and aggregating the result over a binary tree.

\begin{fact}
The above task to find an element with desired $k$-bit label
in an array of length $n$ can be achieved in a circuit 
with $O(nk)$
generalized boolean gates,  
$n$ number of $w$-selector gates, and of depth $O(\log n)$.
\label{fct:findinarray}
\end{fact}
}
\elaine{deleted find in array}

\paragraph{Generalized binary-to-unary conversion.}
\ignore{
Imagine that there are $n$ receivers where the $i$-th receiver
is labeled with an indicator bit ${\bf x}[i]$.
We are given an integer $\ell \in \{0,1,\dots,n\}$ expressed
in binary representation, 
and we want to output an array of $n$ bits where the $i$-th bit represents
the bit received by the $i$-th receiver. 
We want that the first $\ell$ receivers marked with $1$ 
receive $1$, and all other receivers marked with $1$ 
receive $0$. The receivers marked with $0$ may receive an arbitrary bit.
Note that in the special case 
that all receivers are marked with $1$, then the problem
boils down to converting an integer $\ell \in \{0, 1, \ldots, n\}$ 
expressed in binary representation to a corresponding unary string.
}
The generalized binary-to-unary conversion problem
has been defined earlier in Section~\ref{sec:prambldgblock}.
Earlier, we also showed how to solve it on an oblivious PRAM  
in linear total work and logarithmic depth. 
It turns out that it is a little trickier if we want 
to accomplish the same with a linear-sized 
and logarithmic depth circuit. This is because on a PRAM,  
arithmetic and boolean operations on $\log n$ bits can be performed in unit cost,
whereas in a circuit model, we charge $O(\log n)$.

We can solve the generalized binary-to-unary conversion problem 
with the following algorithm.
Without loss of generality, we can assume that $n$ is a power of $2$;
if not, we can round $n$ up to the nearest power of $2$.

\begin{mdframed}
\begin{center}
{\bf Generalized binary-to-unary conversion circuit}
\end{center}
\begin{enumerate}[leftmargin=5mm,itemsep=1pt]
\item 
First, we apply the counting circuit of Fact~\ref{fct:countprelim} 
to the input array ${\bf x}$. Specifically, we compute the sum
over a binary tree using the delayed-carry representation  
of numbers. At the end of this step,
every tree node stores the sum of its subtree, in 
delayed-carry representation. 

Henceforth, let $S(v)$ denote the 
sum of the subtree rooted at the node $v$.
We may assume that all numbers below use  
a delayed-carry representation.

\ignore{
\item
The root node now receives the number $\ell$. 
Let ${\sf lc}$ and ${\sf rc}$ be its left child and right child,
respectively.
The root sends $v$ to ${\sf lc}$ and sends $v - S({\sf lc})$ to ${\sf rc}$.
}
\item 
\label{step:passnumber}
For convenience, assume that the root receives $\ell$ 
from an imaginary parent.

From level $i = 0$ to $\log n - 1$: every node in level $i$ performs
the following. Let $S$ be the number received
from its parent,  and let ${\sf lc}$ and ${\sf rc}$ denote
the node's left child and right child, respectively. 
Send $S$ to ${\sf lc}$ and send $S - S({\sf lc})$
to ${\sf rc}$. 

\item 
\label{step:passlabel}
For convenience, assume that the root receives the label ``{\tt M}'' 
from an imaginary parent.

From level $i = 0$ to $\log n - 1$, every node in level $i$ does the following
where ${\sf lc}$ and ${\sf rc}$ denote its left child and right child,
respectively:
\begin{itemize}[leftmargin=5mm,itemsep=1pt]
\item
If the label received from its parent is not ``{\tt M}'', then 
pass the label to both children; 
\item 
Else, let $S$ be the number received earlier from its parent. 
\begin{itemize}[leftmargin=5mm,itemsep=1pt]
\item
if $S \geq S({\sf lc})$
then pass ``1'' to ${\sf lc}$ and 
pass ``{\tt M}'' to ${\sf rc}$; 

\item
else, pass ``{\tt M}'' to ${\sf lc}$ and pass ``0'' to ${\sf rc}$. 
\end{itemize}
\end{itemize}
\item 
If a leaf node receives ``0'' or ``1'' from the parent, then output 
the corresponding bit. 
Otherwise, let $S$ be the 1-bit number  
received from the parent, output $S$.
\end{enumerate}

\paragraph{Implementation as a circuit.}
All numbers use a delayed-carry representation. 
Let $v_1 := x_1 + y_1$ and $v_2 := x_2 + y_2$ be two $\ell$-bit numbers
in delayed-carry representation, and suppose that $v_1 \geq v_2$.
Then, $v_1 - v_2$ can be derived 
by computing $x_1 + x_2 + \overline{y}_1 + \overline{y}_2 + 2$
and keeping only the last $\ell$ bits, 
where $\overline{y}_b$ denotes the number obtained by flipping
all bits of $y_b$ for $b \in \{1, 2\}$.
Therefore, we can use the shallow addition trick to compute subtraction.
Of course, before a node receives the label  
from $\{{\tt M}, 0, 1\}$ from its parent, 
it is not guaranteed that $S \geq S({\sf lc})$, but we can just pretend
it will be the case and continue.
Therefore, Step~\ref{step:passnumber}
can be implemented in $O(\log n)$ depth.

Step~\ref{step:passlabel}
must be implemented in a pipelined manner to save depth: basically, as soon
as a node receives the number $S$ from its parent during 
Step~\ref{step:passnumber}, it immediately
starts to compute the comparison between $S$ and $S({\sf lc})$
which takes $O(\log \log n)$ depth.
In other words, the nodes do not wait for its parent
to compute this comparison before it computes its own
comparison, but rather 
pre-computes this comparison ahead of time. 
Using this pipelining trick, 
Step~\ref{step:passlabel}
can also be accomplished in $O(\log n)$ depth.

Finally, observe that $S$   
is at most $\log n + 1$ bits at the root; and at level $i$ it is at most 
$\log n + 1 - i$ bits. 
Therefore, the above can be implemented with an $O(n)$-sized circuit.
\end{mdframed}

This gives rise to the following fact.
\begin{fact}
There is a circuit with $O(n)$ generalized boolean gates 
and of $O(\log n)$ depth that solves the aforementioned generalized
binary-to-unary conversion problem.
\label{fct:binary_to_unary}
\end{fact}

\ignore{
\elaine{TO ADD: we need a generalized version of the binary-unary conv.}

\paragraph{Binary to unary conversion.}
Fix $n \in \N$. For any integer $\ell \in \{0,1,\dots,n\}$,
the binary-to-unary conversion 
takes as input $\ell$ and then outputs an $n$-bit string $u$,
where $\ell$ is represented in binary (i.e., a 
string of $\ceil{\log n} + 1$ bits),
and $u$ is the unary representation of $\ell$, i.e.,
the head $\ell$ bits of $u$ are all 0s and the tail $n -\ell$ 
bits of $u$ are all 1s. Suppose $n$ is a power of 2
and $\ell$ is represented in $\log n + 1$ bits.
Asharov et al.~\cite{soda21}
describes a circuit of size $O(n)$ and depth $O(\log n)$ that accomplishes
the aforementioned binary-to-unary conversion task, as stated
in the following fact:

\begin{fact}
The binary-to-unary conversion task
can be implemented in a circuit 
of size $O(n)$ and depth $O(\log n)$.

\elaine{do we need to state it as generalized boolean  gates in these facts?}
\label{fct:binary_to_unary}
\end{fact}

}

\section{Lossy Loose Compaction Circuit}
\label{sec:llc}

\subsection{Definitions}

\paragraph{Lossy loose compactor.}
Let $\alpha \in [0, 1)$. An $(\alpha, n, w)$-lossy loose compactor 
(also written as $\alpha$-lossy loose compactor
when $n$ and $w$ are clear from the context)
solves the following problem: 
\begin{itemize}[leftmargin=5mm,itemsep=1pt]
\item 
{\bf Input}:
an array ${\bf I}$ containing $n$ elements 
of the form $\{(b_i, v_i)\}_{i \in [n]}$, 
where each $b_i \in \{0, 1\}$ is 
a metadata bit indicating 
whether the element is {\it real} or {\it filler},
and each $v_i \in \{0, 1\}^w$ is the {\it payload}.
The input array is promised to 
{\it have at most $n/128$ real elements}.

\item 
{\bf Output}:
an array ${\bf O}$ containing $\floor{n/1.9}$ elements, 
such that ${\it mset}({\bf O}) \subseteq {\it mset}({\bf I})$,
and moreover, 
$|{\it mset}({\bf I}) - {\it mset}({\bf O})| 
\leq \alpha n$
where ${\it mset}({\bf O})$
denotes the {\it multiset} of 
real elements contained in ${\bf O}$, and 
${\it mset}({\bf I})$
is similarly defined.
\end{itemize}

In other words, lossy loose compaction takes a relatively sparse 
input array containing only a small constant fraction of real elements;
it compresses the input 
to slightly more than half its original length\footnote{It is 
not exactly
half the original length 
due to rounding issues --- See Remark~\ref{rmk:notdivisible}.}
while preserving 
all but $\alpha \cdot n$ real elements in the input.

\paragraph{Loose compactor.}
If $\alpha = 0$, i.e., there is no loss,
we also call it a loose compactor.  
More formally, an $(0, n, w)$-lossy loose compactor is also called 
an $(n, w)$-loose compactor.

\subsection{Intuition for the Next 4 Sections:
Bootstrapping an Efficient Lossy Loose Compactor}
\label{sec:roadmap:bootstrap}
Fix an arbitrary constant $C > 2$.
First, we want to construct 
a $1/(\log n)^C$-lossy loose compactor
that has $O(n \cdot w)$ generalized boolean gates, 
$O(n)$ number of $w$-selector gates (ignoring $\poly\log^*$ terms), 
and with depth $O(\log^{0.5}n)$ --- here $w$ denotes
the bit-width of an element's payload.

We could get an {\it inefficient} $1/(\log n)^C$-lossy loose compactor
(for an arbitrary constant $C > 1$)
using techniques described in Asharov et al.~\cite{paracompact}: 
specifically, the 
resulting $1/(\log n)^C$-lossy loose compactor requires
$O(n \log \log n)$ generalized boolean gates, $O(n)$ number of $w$-selector
gates, and incurs depth $O(\log \log n)$. 
If $w \geq \log \log n$, we would then be able to implement
this as a constant fan-in, constant fan-out boolean circuit 
of $O(n w)$ size and $O(\log n + \log w)$ depth.

Henceforth we focus on the case when $w = o(\log \log n)$.
In this case, the generalized boolean gates cost
asymptotically 
more than the $w$-selector gates when we fully instantiate the circuit as  
a constant fan-in, constant fan-out boolean circuit.
We want to bootstrap a more efficient
$1/(\log n)^C$-lossy loose compactor 
by balancing these two costs.
During the bootstrapping,
we can blow up the 
$\alpha$ parameter (i.e., the fraction of lost elements)
by at most a (poly-)logarithmic factor.

We are inspired by Asharov et al.~\cite{soda21}'s 
repeated bootstrapping technique:
they 
use a loose compactor to bootstrap a 
tight compactor without incurring too much overhead, and then
use the tight compactor to bootstrap a loose compactor 
much more efficient than the original one.
 This is repeated for $d := \log(\log^* n - \log^* w)$
times.
Unfortunately, even if we allow
lossiness, we cannot directly use their techniques due 
to the blowup in depth. 
One critical factor contributing to the depth blowup
comes from the bootstrapping step in which they construct 
a tight compactor given a loose compactor.
Here, they have to perform metadata computation 
that is $\Theta(\log n)$ in depth.
This would incur at least $(\Theta(\log n))^d$
total depth over all steps of the bootstrapping, 
where $d := \log(\log^* n - \log^* w)$.

Our key observation is to use a weaker intermediate 
abstraction during the bootstrapping, that is, an approximate splitter. 
Specifically, we  
use a lossy loose compactor to construct an
approximate splitter without incurring too much overhead,
and then use the resulting approximate splitter
to construct a lossy loose compactor much more efficient than the
original one.
Unlike Asharov et al.~\cite{soda21},  
the repeated bootstrapping no longer gives us
a tight compactor directly; it only  
gives an efficient lossy loose
compactor. As explained later, getting a 
tight compactor from an efficient lossy loose compactor 
requires additional novel techniques.

\paragraph{Approximate splitter from lossy loose compactor.}
In a pre-processing phase, 
we first mark misplaced elements (and some additional
elements) as either ${\tt blue}$ 
or ${\tt red}$, such that the approximate splitter
task can be expressed as 
pairing up each ${\tt blue}$ with a distinct ${\tt red}$
and swapping almost all such pairs. 
Specifically, any distinguished element
not contained in the first $\floor{\beta n + n/64}$
positions of the input are colored ${\tt blue}$.
Any non-distinguished element
contained in the first $\floor{\beta n + n/64}$
positions of the input are colored ${\tt red}$.
This makes sure that $n_{\rm red} \geq n_{\rm blue} + n/64$,
where $n_{\rm red}$ and $n_{\rm blue}$
denote the number of ${\tt red}$ and ${\tt blue}$ elements,
respectively.
Observe that the metadata computation in the pre-processing
step has constant depth (as opposed to logarithmic depth had
we used the tight compaction version of the bootstrapping~\cite{soda21}).

Next, we rely on an {\it approximate swapper}
that swaps most of the ${\tt blue}$ 
elements with their paired ${\tt red}$, except for
leaving behind at most $n/256$ ${\tt blue}$ elements that 
are unswapped.
Henceforth, we may assume that swapped elements become uncolored.
Such an approximate swapper circuit can be constructed
using a linear-sized and constant-depth circuit by combining
prior techniques~\cite{paracompact,soda21}.

Now, we want to extract 
almost all of the remaining ${\tt blue}$ 
elements except for at most $n/\poly\log n$ of them, 
as well as slightly more ${\tt red}$ elements 
than ${\tt blue}$. Further, the extracted array is a constant
factor shorter than the original array.
For technical reasons, 
we have to use a different 
algorithm for extracting the ${\tt blue}$
and ${\tt red}$ elements, respectively.
Specifically, we rely on a lossy loose compactor
to extract the ${\tt blue}$ elements;
and rely on an $\epsilon'$-near-sorter 
to extract the ${\tt red}$ elements for some sufficiently
small constant $\epsilon' \in (0, 1)$.
At this moment, the problem 
boils down to swapping almost all ${\tt blue}$ elements
in the extracted array with a distinct, paired ${\tt red}$
element, and 
reverse routing the 
result back to the original array. 
We can accomplish this by recursing on the extracted array.
The recursion stops
when the extracted array's size becomes 
$n/\poly\log n$ for some appropriate choice of $\poly\log(\cdot)$.

We defer a formal description of the scheme and the parameters
to the subsequent technical sections.
This bootstrapping step incurs the following blowup in parameters:
\begin{itemize}[leftmargin=5mm,itemsep=1pt,topsep=2pt]
\item 
Let $\alpha := 1/\poly\log n$ 
be the loss-factor of the $\alpha$-lossy loose compactor,
then the resulting approximate splitter has the 
approximation factor $8 \alpha$.
\item 
Suppose that the $1/\poly\log(n)$-lossy loose compactor
has $B_{\rm lc}(n)$ number of 
generalized boolean gates, and
$S_{\rm lc}(n)$ number of 
$w$-selector gates, and has $D_{\rm lc}(n)$
depth, then the resulting 
approximate splitter 
has $C_1 \cdot B_{\rm lc}(n)$ 
generalized boolean gates, $C_2 \cdot S_{\rm lc}(n)$
number of $w$-selector gates, 
and $C_3 \log \log n \cdot D_{\rm lc}(n)$ depth,
where $C_1, C_2, C_3 > 2$ are appropriately large constants.
\end{itemize}

\paragraph{Lossy loose compactor from approximate splitter.}
We want to construct 
a more efficient lossy loose compactor given an approximate splitter.
\ignore{
The idea here is similar to how Asharov et al.~\cite{soda21}
constructed a loose compactor from a tight compactor, except
that we instead use an approximate splitter in lieu of their
tight compactor. As a result, we need to re-parametrize, 
and moreover, there will be some elements
lost during this bootstrapping step.
}
Let $f(n) < \log n$ be some function on $n$, and let $C_{\rm sp} > 1 $ 
be some appropriate constant.
Suppose that we have an $\alpha$-approximate splitter
that costs $C_{\rm sp} \cdot n \cdot f(n)$
generalized boolean gates,  
$C_{\rm sp} \cdot n$ number of $w$-selector gates, and 
$D_{\rm sp}(n)$ depth.
We can construct a lossy loose compactor as follows:

\begin{enumerate}[leftmargin=5mm,itemsep=1pt,topsep=2pt]
\item 
Divide the input array into $f(n)$-sized chunks. 
We say that a chunk is {\it sparse} if there are at most $f(n)/32$ 
real elements in it; otherwise
it is called {\it dense}.
Since the input is promised to be $1/128$-sparse,  
we will later prove 
that at least $3/4$ fraction of the chunks are sparse.
\item 
Call an $(\alpha, 1/4)$-approximate splitter 
 to move almost all 
dense chunks to the front and almost all sparse chunks  
to the end.
Here the approximate splitter 
works on $n/f(n)$ elements each of bit-width $f(n) \cdot w$.

\item 
Apply an $(\alpha, 1/32)$-approximate splitter
to the trailing $\ceil{(\frac{3}{4} - \frac{1}{64}) 
\cdot \frac{n}{f(n)}}$ chunks to compress
each of these chunks to a length of $\floor{\frac{3f(n)}{64}}$,
losing few elements in the process.
The first $\floor{(\frac{1}{4} + \frac{1}{64}) \cdot \frac{n}{f(n)}}$
chunks are unchanged. Output the resulting array.
\end{enumerate}

The resulting lossy loose compactor  
has a lossy factor of $1.74\alpha$;
moreover, it costs  
at most $2.1 \cdot C_{\rm sp}\cdot n\cdot f(f(n))$
generalized boolean gates, at most  
$2.1 \cdot C_{\rm sp}\cdot n$
number of $w$-selector gates, and has depth $2.1 D_{\rm sp}(n)$.
Note that the total number of generalized boolean gates
reduces quite significantly in this step but the total number of $w$-selector 
gates and the depth  increase by a constant factor.

\paragraph{Repeated bootstrapping.}
We repeatedly perform the above bootstrapping.
Henceforth going from lossy loose compactor
to approximate splitter, and then back to a lossy loose compactor 
is called one step in our bootstrapping.
After $d := \log (\log^*n - \log^*w)$ steps of bootstrapping,
the cost incurred by generalized
boolean gates and $w$-selector gates will be balanced.
Specifically, 
there will be $O(nw) \cdot \poly(\log^* n - \log^* w)$
generalized boolean gates and 
$O(n) \cdot \poly(\log^* n - \log^* w)$ number of $w$-selector gates. 
Both can be instantiated with 
$O(nw) \cdot \poly(\log^* n - \log^* w)$
number of AND, OR, and NOT gates of constant fan-in.
After $d$ steps of bootstrapping,
the depth will be $\log \log n \cdot (\Theta(\log\log n))^d$
which is upper bounded by $O(\log^{0.5} n)$.
The total lossy factor 
will be $\poly(\log^* n - \log^* w) \cdot \alpha$, where $\alpha = 1/\poly\log(n)$ 
denotes the lossy factor of the initial lossy loose compactor we started out with.

\section{Inefficient Lossy Loose Compaction Circuit}
\label{sec:slowllc}
In this section, we will 
 prove the following theorem. 

\begin{theorem}
Let $c > 1$ be an arbitrary constant.
There is a circuit in the indivisible model with 
$O(n \log \log n)$ generalized boolean gates,  
$O(n)$ number of $w$-selector gates, 
and of depth $O(\log \log n)$
that realizes
an $(\frac{1}{\log^c n}, n, w)$-lossy loose compactor.
\label{thm:initllc}
\end{theorem}

To prove the above theorem, 
we describe how to implement lossy loose compaction as a low-depth circuit.
Our construction is almost the same as 
the loose compactor circuit described by 
Asharov et al.~\cite{soda21}, except that 
we now run the algorithm for fewer (i.e., $c \log \log n$) iterations rather
than $O(\log n)$ iterations).
Because we omit some iterations, 
we end up losing a small fraction of elements during the loose compaction.
\ignore{
and 2) we implement the circuit using low-depth 
circuit gadgets 
described earlier in Section~\ref{sec:circuitbldgblock},
rather than the circuit gadgets
which Asharov et al.~\cite{paracompact}
used since they did not care about depth.
}
We describe the algorithm below.

\elaine{could defer to appendix?}

\paragraph{Expander graphs.}
The construction will rely on a suitable family of bipartite 
expander graphs denoted $\{G_{\epsilon, m}\}_{m \in {\sf SQ}}$
where ${\sf SQ} \subseteq \N$ is the set of perfect squares.
The parameter $\epsilon \in (0, 1)$ is a suitable constant referred to as 
the {\it spectral expansion}. 
The graph $G_{\epsilon, m}$
has $m$ vertices on the left henceforth denoted $L$, and $m$
vertices on the right henceforth denoted $R$,  
and each vertex has $d := d(\epsilon)$ number of edges
where $d$ is a constant that depends on $\epsilon$.
We give additional preliminaries on expander graphs
in Appendix~\ref{sec:expanderprelim}.

Without loss of generality, we may assume that $d$ 
is a multiple of $8$ since we can always consider the graph
that duplicates each edge $8$ times.

\paragraph{Construction.}
The input array is grouped into chunks of $d/2$ size. 
Chunks that have at most $d/8$ elements (i.e., a quarter loaded)
are said to be {\it sparse} and chunks that have more than $d/8$ 
elements are said to be {\it dense}.
The idea is to first distribute the dense chunks 
such that 
there are only very few dense chunks after this step.
Then, we can easily compact each chunk separately.
When the remaining dense chunks are compressed, we end up  
losing some elements.

The challenge is how to distribute
the dense chunks. We can consider the chunks to be left-vertices 
in the bipartite expander graph $G_{\epsilon, m}$.
Each dense chunk wants to 
distribute its real elements to its neighbors on the right, such that each
right vertex receives no more than $d/8$ elements, i.e., each vertex
on the right is a sparse chunk too.
At this moment, we can replace dense chunks on the left 
with filler elements --- for almost all of these
dense chunks, their real elements have moved to the right.
For the remaining dense chunks, replacing them with filler 
causes some elements to be lost.
Now that all chunks are sparse, 
and we can compress each chunk on the left and the right 
to a quarter its original size. 
All compressed chunks are concatenated and output, and the output
array is a half the length of the input.  

The distribution of the real elements to its neighbors on the right requires some care, as we have to guarantee that no node on the right will become dense. We will have to compute on which subset of edges we will route the real elements. This is done via the procedure ${\sf ProposeAcceptFinalize}$ described below. 
\elaine{is our previous impl of the algo small depth?}

\paragraph{${\sf ProposeAcceptFinalize}$ subroutine.}
We now describe the ${\sf ProposeAcceptFinalize}$ subroutine 
which is the key step to achieve the aforementioned distribution 
of dense chunks.
To make the description more intuitive, henceforth
we call each vertex in $L$ a {\it factory} and each vertex in $R$ a {\it facility}.
Initially, 
imagine that the dense vertices correspond to factories that manufacture at most $d/2$
products, and the sparse vertices are factories that are unproductive.
There are at most $m/32$ productive factories, and they 
want to route all their products to facilities on the right satisfying the following
constraints: 1) each edge can route only 1 product; and 2) 
each facility can receive at most $d/8$ products.
The ${\sf ProposeAcceptFinalize}$ algorithm described below finds 
a set of edges $M$ to enable such routing, 
also called a feasible route as explained earlier. 

\begin{mdframed}
\begin{center}
${\sf ProposeAcceptFinalize}$
{\bf subroutine}
\end{center}

\vspace{3pt}
Initially, each productive factory is {\it unsatisfied} and 
each unproductive factory is {\it satisfied}.
For a productive factory $u \in L$, we use notation ${\sf load}(u)$ 
to denote the number of products it has (corresponding to the number
of real elements in the chunk).

\paragraph{Algorithm:} Repeat the following for 
${\sf iter}$
times
and output the resulting matching $M$ at the end:
\elaine{specify const}

\begin{enumerate}[label=(\alph*),itemsep=1pt]
\item {\it Propose:} 
Each unsatisfied factory 
sends a proposal (i.e., the bit 1) to each one of its neighbors. 
Each satisfied factory sends 0 to each one of its neighbors. 

\item {\it Accept:}
If a facility $v \in R$
received no more than $d/8$ proposals, 
it sends an acceptance message to each one of its $d$ neighbors;
otherwise, it sends a reject message along each of its $d$ edges. 

\item  {\it Finalize:}
Each currently unsatisfied factory $u \in L$
checks if it received at least $\frac{d}{2}$ acceptance messages.
If so, 
add the set of edges over which acceptance messages
are received to the matching $M$.
At this moment, this factory becomes satisfied.
\end{enumerate}
\end{mdframed}

Notice that for a facility $v \in R$, 
the proposals it receives in iteration $i+1$ is a subset
of the proposals it receives in iteration $i$. Therefore,
once $v$ starts accepting 
in some iteration $i$, it will also 
accept all proposals received in future rounds $i+1, i+2, \ldots$ 
too, if any proposals are received.
Moreover, the total number of product $v$ receives
will not exceed $d/8$. 
Pippenger~\cite{selfroutingsuperconcentrator}
and Asharov et al.~\cite{paracompact} showed the following fact: 

\begin{fact}[Pippenger~\cite{selfroutingsuperconcentrator} and  
Asharov et al.~\cite{paracompact}]
There exist an appropriate constant $\epsilon \in (0, 1)$
and a bipartite expander graph family $\{G_{\epsilon, m}\}_{m \in \N}$
where each vertex has $d$ edges for a constant
$d := d(\epsilon)$ assumed to be a multiple of $8$, 
such that for any $m \in {\sf SQ} \subseteq \N$,
at the end of the above ${\sf ProposeAcceptFinalize}$ procedure 
which runs for ${\sf iter}$ iterations
(and assuming it is instantiated with 
the family of graphs $\{G_{\epsilon, m}\}_{m \in \N}$), 
the following must hold:
\begin{enumerate}[leftmargin=7mm,itemsep=0pt]
\item 
at most $m/2^{\sf iter}$ vertices in $L$ remain unsatisfied;
\item 
every satisfied vertex in $u \in L$ has 
at least $d/2$
edges in the output matching $M$;
\item 
for every vertex in $v \in R$, the output 
matching $M$ has at most $d/8$ edges
incident to $v$.
\end{enumerate}
\label{fct:matching}
\end{fact}

Given the ${\sf ProposeAcceptFinalize}$ subroutine,
we can realize a $1/\log^c n$-lossy loose compaction as follows
where $c > 1$ denotes a constant.

\begin{mdframed}
\begin{center}
{\bf $1/(\log n)^c$-Lossy Loose Compaction} 
\end{center}
\begin{itemize}[leftmargin=5mm,itemsep=1pt]
\item {\bf Input:} An input array ${\bf I}$ of $n$ elements, in which at most $n/128$ are real and the rest are dummies. 
\item
{\bf Assumption}:
Without loss of generality, we assume  
that $d$ is a multiple of $8$.
Further, we assume that $m$ is a perfect square 
and that $n$ is a multiple of $d/2$; henceforth let $m := n/(d/2) = 2n/d$.
The algorithm can be easily generalized to any choice
of $n$ --- see Remark~\ref{rmk:notdivisible}.

\item 
{\bf The algorithm}:
\begin{enumerate}[leftmargin=5mm]
\item 
Divide ${\bf I}$ into $m$ chunks of size $d/2$.
If a chunk contains at most $d/8$ real elements (i.e., at most a quarter loaded), 
it is said to be {\it sparse};
otherwise it is said to be {\it dense}.
It is not hard to see that the number of dense chunks must be at most  $m/32$.
\label{step:sparsedense}

\item 
Now imagine that each chunk is a vertex in $L$ of $G_{\epsilon, m}$, 
and $D \subset L$ is a set of dense vertices (i.e., corresponding to 
the dense chunks). 
Let ${\sf edges}(D, R)$ denote all the edges in $G_{\epsilon, m}$
between $D \subset L$ and $R$. 

Let $D$ be the subset of productive factories, and 
run the 
${\sf ProposeAcceptFinalize}$ subroutine for $c \log \log n$ iterations.
The outcome is a 
subset of edges $M \subseteq {\sf edges}(D, R)$ 
that satisfy Fact~\ref{fct:matching}, where the fraction
of unsatisfied chunks is $1/\log^c n$.
\ignore{
 such that 
\begin{enumerate}[leftmargin=7mm,itemsep=1pt,label=(\alph*)]
\item 
$1-1/\log^c n$ fraction of the vertices
$u \in D$ each has ${\sf load}(u) \leq d/2$ 
edges in $M$
where ${\sf load}(u)$ denotes the number of real elements in the chunk;
and 
\item 
every vertex $R$ has at most $d/8$ incoming edges in $M$.
\end{enumerate}
}
%
%
\label{step:matchfinding}
\item 
Now, every dense chunk $u$ in $D$ 
does the following:
for each of an arbitrary subset of ${\sf load}(u) \leq d/2$ 
outgoing edges of $u$ in $M$, send one element over the edge 
to a corresponding neighbor in $R$; for all remaining out edges
of $u$, send a filler element on the edge.

Every vertex in $R$ receives
no more than $d/8$ real elements.  
Henceforth we may consider every vertex in $R$
as a sparse chunk, i.e., an array of capacity $d/2$ but containing
only $d/8$ real elements. 

\label{step:onlineroute}
\item 
At this moment, 
for each dense 
chunk in $L$, replace the entire chunk with $d/2$ filler elements.

\label{step:nodense}
\item 
Now, all chunks in $L$ and in $R$ must be sparse, that is, 
each chunk contains at most $d/8$ real elements, while its size is $d/2$. We now compress each chunk in $L$ and $R$ to a quarter of its original size (i.e., to size $d/8$ in length), 
losing few elements in the process (we will bound
the number of lost elements later).

Output the compressed array ${\bf O}$, containing of $2m \cdot \frac{d}{8} = 2 \cdot \frac{2n}{d} \cdot \frac{d}{8} = n/2$ elements. 
\label{step:compress}
\end{enumerate}
\end{itemize}
\end{mdframed}

\begin{proposition}[Pippenger~\cite{selfroutingsuperconcentrator} and Asharov et al.~\cite{paracompact}]
There exists an appropriate constant $\epsilon \in (0, 1)$
and a bipartite expander graph family $\{G_{\epsilon, m}\}_{m \in \N}$
where each vertex has
$d$ edges for a constant $d := d(\epsilon)$, 
such that for 
any perfect square $m$ and $n = md/2$,
the above lossy loose compaction algorithm, 
when instantiated with this family of
 bipartite expander graph, 
can correctly compress any input array of length $n$
to a half its original size 
losing at most $n/\log^c n$ 
real elements, 
as long as the input array has at most $n/128$ real elements.
\label{prop:llc}
\end{proposition}

\begin{Remark}
In the above, we assumed that $n$ is divisible by $d/2$. 
If $n = d m/2$ where $m$ is a perfect square. 
In case this is not the case, we can always round $n$ up to the
next integer that satisfies this requirement; 
this blows up $n$ by at most a $1+o(1)$ factor. 
This is why in 
our definition of lossy loose compactor, 
the output size is allowed
to be $\floor{n/1.9}$ rather than $\floor{n/2}$, assuming that $n$
is sufficiently large.
\label{rmk:notdivisible}
\end{Remark}

\paragraph{Implementing the algorithm in a low-depth circuit.}
Since our lossy loose compactor 
algorithm is almost the same as Asharov et al.'s loose compactor~\cite{soda21},
we can implement 
the algorithm as a circuit in almost the same way 
as described by Asharov et al.~\cite{soda21}, except that we run fewer
iterations.
It is not hard to check that
the resulting circuit has $O(n \log \log n)$
generalized boolean gates, $O(n)$
number of $w$-selector gates, and 
has depth $O(\log\log n)$.

\ignore{
The only
difference is that we now employ the low-depth counterparts of the 
circuit gadgets
described earlier in Section~\ref{sec:circuitbldgblock}
rather than the circuit gadgets Asharov et al. adopted.
}

\ignore{
\subsection{Slow Loose Compaction Circuit}

\begin{theorem}[Slow loose compaction circuit ${\bf SlowLC}$]
There is a circuit in the indivisible model with 
$O(n \log n)$ generalized boolean gates,  
$O(n)$ number of $w$-selector gates, 
and of depth $O(\log n)$ that realizes
an $(n, w)$-loose compactor.
\label{thm:slowlc}
\end{theorem}
\begin{proof}
We can construct a slow loose compaction circuit using the 
$1/(\log n)^c$-lossy loose compaction 
algorithm of Theorem~\ref{thm:initllc}, but now instead
we run the  
${\sf ProposeAcceptFinalize}$ algorithm for ${\sf iter} := 
\ceil{\log n}$ iterations.
In this way, there will not be any loss during the compression.
The performance 
bounds can be analyzed in a similar manner as in 
Theorem~\ref{thm:initllc} 
except that now we run the  
${\sf ProposeAcceptFinalize}$ algorithm for more iterations.
\end{proof}
}

\elaine{i deleted the slowlc algorithm. we don't need to refer
to it now. since for slowtc we just refer to soda'21 paper now}

\elaine{maybe we could just use poly log depth 
loose compactor to bootstrap? 
but we will still need lossy loose compactor
to bootstrap an imperfect TC
}

\section{Approximate Splitter from Lossy Loose Compaction}
\label{sec:splitterfromlc}

\subsection{Approximate Swapper Circuit}

\paragraph{Approximate swapper.}
An $(n, w)$-approximate swapper obtains an input array where
each element is marked with a label that is 
 $\bot$, ${\tt blue}$, or ${\tt red}$.
Let $n_{\rm red}$ and $n_{\rm blue}$
denote the number of red and blue elements, respectively.
The $(n, w)$-approximate swapper 
circuit swaps a subset of the ${\tt blue}$
elements with ${\tt red}$ ones and the swapped 
elements receive the label $\bot$.
We call elements marked ${\tt red}$ or ${\tt blue}$
{\it colored} and those marked $\bot$ {\it uncolored}.

Formally, an $(n, w)$-approximate swapper 
solves the following problem:
\begin{itemize}[leftmargin=5mm,itemsep=1pt]
\item {\bf Input}:
an input array containing $n$ elements
where each element contains a $w$-bit payload string
and a two-bit metadata label
whose value is chosen
from the set $\{{\tt blue}, {\tt red}, \bot\}$.
Henceforth we assume the first bit of the label
encodes whether the element is colored or not, 
and the second bit of the label
picks a color between ${\tt blue}$ and ${\tt red}$ 
if the element is indeed colored.

\item {\bf Output}:
a {\it legal swap} of the input array such that
at most $n/128 + |n_{\rm red} - n_{\rm blue}|$ 
elements remain colored, where the notion of
a legal swap is defined below.

We say that an output array ${\bf O}$ is a {\it legal swap}
of the input array ${\bf I}$ iff
there exist pairs $(i_1, j_1), (i_2, j_2), \ldots, (i_\ell, j_\ell)$
of indices that are all distinct,
such that for all $k \in [\ell]$, ${\bf I}[i_k]$
and ${\bf I}[j_k]$ are colored and have opposite colors,
and moreover ${\bf O}$ is obtained by swapping
${\bf I}[i_1]$ with ${\bf I}[j_1]$,
swapping ${\bf I}[i_2]$ with ${\bf I}[j_2]$,  $\ldots$, and
swapping ${\bf I}[i_k]$ with ${\bf I}[j_k]$; further, all swapped
elements become uncolored.
\end{itemize}

\ignore{
\paragraph{Swapper.}
If $\alpha = 0$, i.e., all colored elements are swapped
except for the residual caused by imbalance $|n_{\rm red} - n_{\rm blue}|$,
then we simply call it a {\it swapper}.
Formally, 
an $(0, n, w)$-approximate swapper is also called 
an {\it $(n, w)$-swapper} (or simply {\it swapper} for short).
}


\begin{theorem}[Approximate swapper]
There exists an $(n, w)$-approximate swapper 
circuit containing $O(n)$ generalized boolean gates 
and $O(n)$ number of $w$-selector gates, \elaine{is it w} and of depth 
$O(1)$.
\label{thm:approxswap}
\end{theorem}
\begin{proof}
We can use Algorithm 6.10 in Asharov et al.~\cite{paracompact}: their
algorithm is described for
the oblivious PRAM model, and achieves $O(n)$ work and $O(1)$ depth. 
It is straightforward to check that 
the same algorithm can be implemented 
as a circuit with $O(n)$ generalized boolean gates, 
$O(n)$ number of $w$-selector gates, 
and $O(1)$ in depth.
Note that Algorithm 6.10 in Asharov et al.~\cite{paracompact}
needs to compute the decomposed perfect matchings on the fly since their
oblivious PRAM algorithm is uniform;
however, we do not need to compute the matchings on the fly  
in the circuit model, since the circuit
model is non-uniform.
\end{proof}

\ignore{
\begin{theorem}[$1/\poly \log n$-approximate swapper]
Suppose that there is an 
$(\alpha, n, w)$-lossy loose compaction
circuit with $B_{\rm lc}(n)$
generalized boolean gates and $S_{\rm lc}(n)$ $w$-selector gates,
and of depth $D_{\rm lc}(n)$.
Then, 
for any constant $c > 1$, 
there exists a $(1/\log^c n, n, w)$-approximate swapper 
circuit containing 
\elaine{FILL}
generalized boolean gates 
and  \elaine{FILL}
number of $w$-selector gates, and of depth  
\elaine{FILL}.
\end{theorem}
}

\paragraph{Swapper.}
A swapper is defined in almost the same way
as an approximate swapper, except  
that we require that the remaining colored
elements do not exceed $\abs{n_{\rm red} - n_{\rm blue}}$.
In other words, if initially, the number of ${\tt red}$
elements equals the number of ${\tt blue}$ elements, 
then the swapper must swap every ${\tt red}$ element
with a distinct ${\tt blue}$ element, leaving 
no colored elements behind.

\begin{theorem}[Slow swapper]
There exists an $(n, w)$-swapper circuit (henceforth
denoted ${\bf SlowSwap}(\cdot)$) 
with $O(n \log n)$
generalized boolean gates, 
and $O(n \log n)$ number of $w$-selector gates \elaine{is it w},
and whose depth is $O(\log n)$.
\label{thm:slowswap}
\end{theorem}
\begin{proof}
We can use the following algorithm:
\begin{enumerate}[leftmargin=5mm,itemsep=1pt]
\item 
Use an AKS sorting circuit~\cite{aks} to sort 
the input array such that 
all the ${\tt red}$ elements are in the front;
and all the ${\tt blue}$ elements are at the end.
Let the result be ${\bf X}$.
\label{step:slowswap-aks}
\item 
For each $i \in 1, 2, \ldots, \floor{n/2}$ in parallel:
if ${\bf X}[i]$ is marked ${\tt red}$
and ${\bf X}[n + 1 -i]$ is marked ${\tt blue}$, 
then swap ${\bf X}[i]$ and ${\bf X}[n + 1 -i]$ and mark 
both elements as uncolored.
\item 
Reverse route the resulting array by reversing
the decisions made by the AKS 
network in Step~\ref{step:slowswap-aks}, and output the result.
\end{enumerate}

Since the AKS sorting network
performs comparisons on labels 
that are at most 2-bits long, 
the entire algorithm can be 
implemented as a circuit with $O(n \log n)$ 
generalized boolean gates, $O(n \log n)$ number of 
$w$-selector gates, and 
of depth $O(\log n)$. 
\end{proof}


\subsection{Approximate Splitter from Lossy Loose Compaction}

\paragraph{Approximate splitter.}
Let $\beta \in (0, 1/4]$ and let $\alpha \in (0, 1)$.
An $(\alpha, \beta, n, w)$-approximate splitter 
(also written as $(\alpha, \beta)$-approximate splitter 
when $n$ and $w$ are clear from the context)
solves the following problem:
suppose we are given an input array ${\bf I}$ containing $n$ elements
where each element has a $w$-bit payload
and a 1-bit label indicating whether the element 
is {\it distinguished} or not.
It is promised that at most $\beta \cdot n$ 
elements in ${\bf I}$ are distinguished.
We want to output 
a permutation (denoted ${\bf O}$) of 
the input array ${\bf I}$, such that 
at most $\alpha n$ distinguished elements 
are not contained in the first $\floor{\beta n + n/64}$
positions of ${\bf O}$. 

\begin{theorem}[Approximate splitter from lossy loose compaction]
Suppose that there is an 
$(\alpha, n, w)$-lossy loose compaction
circuit with $B_{\rm lc}(n)$
generalized boolean gates and $S_{\rm lc}(n)$ $w$-selector gates,
and of depth $D_{\rm lc}(n)$.
Suppose also that there is an $O(1)$-depth 
approximate swapper circuit with
$B_{\rm sw}(n)$
generalized boolean gates and $S_{\rm sw}(n)$ $w$-selector gates
for an input array containing $n$ element each of bit-width $w$.

For any constant $\beta \in (0, 1/4]$,
there is a 
$(8\alpha, \beta, n, w)$-approximate splitter circuit  
with at most $2.5 S_{\rm sw}(n) + 5 S_{\rm lc}(n) + O(n)$
number of $w$-selector gates, 
$2.5 S_{\rm sw}(n) + 
2.5 B_{\rm sw}(n) +  2.5 B_{\rm lc}(n) + 10 S_{\rm lc}(n) + O(n)$
generalized boolean gates, and has depth 
at most $2.4 \log{\frac{1}{\alpha}} \cdot (D_{\rm lc}(n) + O(1))$. 
\elaine{check the constants}
\label{thm:splitterfromlc}
\end{theorem}

\begin{proofof}{Theorem~\ref{thm:splitterfromlc}}
Consider the following algorithm.

\begin{mdframed}
\begin{center}
{\bf Approximate splitter from lossy loose compaction}
\end{center}
\begin{enumerate}[leftmargin=6mm,itemsep=1pt]
\ignore{
\item 
{\it Near-sort.}
Let $\epsilon := \beta/16$.
Apply an $\epsilon$-near-sorter to the input array ${\bf I}$
where distinguished elements are considered smaller than  
non-distinguished elements.
Let ${\bf X}$ be the outcome.
\label{step:nearsort}
}
\item 
\label{step:color}
{\it Color.}
Any distinguished element 
\emph{not} contained in the first $\floor{\beta n + n/64}$
positions of ${\bf X}$ are colored ${\tt blue}$.
Any non-distinguished element 
contained in the first $\floor{\beta n + n/64}$ 
positions of ${\bf X}$ are colored ${\tt red}$.
Observe that 
$n_{\rm red} \geq n_{\rm blue} + n/64$,
where $n_{\rm red}$ and $n_{\rm blue}$
denote the number of ${\tt red}$ and ${\tt blue}$ elements, 
respectively.
\elaine{to fix: i am assuming divisible now.}

Note that at this moment, each element 
in ${\bf X}$ is labeled with 3 bits of metadata, 
one bit of distinguished indicator
and two bits of color-indicator (indicating whether 
the element is colored or uncolored,
and if so, which color).

\item 
{\it Swap.}
\label{step:swapmain}
Call ${\bf Swap}^n({\bf X})$
defined below
to swap the ${\tt blue}$ elements with ${\tt red}$ elements 
except for a small residual fraction 
(here we use a payload of size $w+1$ and not $w$ as 
we also include the distinguished-indicator as part of the payload). 
Return the outcome.
\end{enumerate}
\end{mdframed}

We now describe the ${\bf Swap}^n(\cdot)$ subroutine.

\begin{mdframed}
\begin{center}
${\bf Swap}^n({\bf X})$
\end{center}
\begin{itemize}[leftmargin=5mm,itemsep=0pt]
\item {\bf Input:} An array {\bf X} of $m \leq n$ 
elements, each has a $w$-bit payload\footnotemark\
and a 2-bit label indicating whether the element is colored,
and if so, whether the element is ${\tt blue}$ or ${\tt red}$.
$n$ is the size of the original problem when 
${\bf Swap}$ is first called; the same $n$   
will be passed into all recursive 
calls since it is used to decide
when the recursion stops.  
It is promised that 
$m_{\rm red} \geq m_{\rm blue} + m/64$
where 
$m_{\rm red}$
and 
$m_{\rm blue}$
denote the number of ${\tt red}$ and ${\tt blue}$ elements
in the input array ${\bf X}$, respectively.

\footnotetext{Our approximate splitter algorithm 
actually requires a swapper where elements are of bit-length $w+1$, 
but for convenience
we rename the variable to $w$ in the description of the swapper.
}

\item {\bf Algorithm:}
\begin{enumerate}[leftmargin=5mm,itemsep=1pt,label=(\alph*)] 
\item 
\label{step:splitter:base}
{\it Base case.}
If $m \leq \alpha n$, then 
return ${\bf X}$;
else continue with the following steps. 

\item 
{\it Approximate swapper.}
Call an $(m, w)$-approximate swapper (see Theorem~\ref{thm:approxswap}) 
on ${\bf X}$ to swap elements of opposite colors 
and uncolor them in the process,
such that at most $m/128 + m_{\rm red} - m_{\rm blue}$ 
elements remain colored. 
Let the outcome be called ${\bf X}'$.
\label{step:swap}
\item 
{\it Lossy-extract blue.}
Call an $(\alpha, m, w+1)$-lossy loose 
compactor 
to compact ${\bf X}'$ by a half, 
where the lossy loose compactor treats 
the ${\tt blue}$
elements as real, and all other elements as fillers (i.e.,
the loose compactor treats
the second bit of the color label as a real-filler indicator,
and the first bit of the color label 
is treated as part of the payload).

\ignore{
In other words, the lossy 
loose compactor treats the 1st bit of the color label
as a filler indicator, and treats the 2nd bit of the color label and an 
element's payload string as the payload.
}
Let the outcome be ${\bf Y}_{\rm blue}$ 
whose length is $\floor{\abs{\bf X} / 1.9}$.
\elaine{fix: may not be divisible by 2}
\label{step:compactblue}
\item 
\label{step:compactred}
{\it Extract red.}
Let $\epsilon' = 1/2^{10}$. Apply an $\epsilon'$-near-sorter
(defined in Section~\ref{sec:segmenter-prelim})
to the array ${\bf X}'$ 
treating all ${\tt red}$ elements as 
smaller than all other elements.
Let ${\bf Y}_{\rm red}$   
be the first $\floor{m/32}$ elements of the 
resulting near-sorted array.
Mark every non-${\tt red}$ element in ${\bf Y}_{\rm red}$ as 
uncolored,
and let
${\bf Y} := {\bf Y}_{\rm red} || {\bf Y}_{\rm blue}$.

\item 
{\it Recurse.}
Recursively call ${\bf Swap}^n({\bf Y})$, 
and let the outcome be ${\bf Y}'$.
\label{step:recurse}
\item 
{\it Reverse route.}
Reverse the routing decisions made by all selector gates 
during Steps~\ref{step:compactblue} and \ref{step:compactred}
(see Remark~\ref{rmk:reverseroute}).
Specifically, 
\begin{itemize}[leftmargin=5mm,itemsep=1pt]
\item 
pad ${\bf Y}'[:\floor{m/32}]$ 
with a vector of fillers to a length of $m$ 
and reverse-route the padded array 
by reversing the decisions of Step~\ref{step:compactred} --- let
${\bf Z}_{\rm red}$ be the outcome; 
\item 
reverse-route 
${\bf Y}'[\floor{m/32}+1 : ]$ by reversing the decisions
of Step~\ref{step:compactblue}, resulting
in ${\bf Z}_{\rm blue}$.
\end{itemize}
Note that both ${\bf Z}_{\rm blue}$
and ${\bf Z}_{\rm red}$ have length $m$, i.e., length of the input 
to this recursive call.
\label{step:reverse}

\item 
{\it Output.}
\label{step:output}
Return 
${\bf O}$ which is formed by a performing coordinate-wise 
select operation between ${\bf X}'$, ${\bf Z}_{\rm red}$,
and ${\bf Z}_{\rm blue}$.
For every $i \in [m]$:
\begin{itemize}[leftmargin=5mm,itemsep=1pt]
\item 
if ${\bf X}'[i]$ originally had a ${\tt blue}$ element
and the element was not lost during Step~\ref{step:compactblue},
then let ${\bf O}[i] := {\bf Z}_{\rm blue}[i]$;
\item 
if ${\bf X}'[i]$ originally had a ${\tt red}$ element
and ${\bf Z}_{\rm red}[i]$ is not a filler, then 
let ${\bf O}[i] := {\bf Z}_{\rm red}[i]$;
\item 
else 
let ${\bf O}[i] := {\bf X}'[i]$;
\end{itemize}

\ignore{back to the input ${\bf X}$, such that each position 
$i$ receives either a real or a dummy element: if a real element
is received the element ${\bf X}[i]$ is overwritten with the received element;
otherwise the element ${\bf X}[i]$ is unchanged.
Finally, output the result.
}
\end{enumerate} 
\end{itemize}
\end{mdframed}

\begin{Remark}[Reverse routing details]
For every selector gate $g$ in Steps~\ref{step:compactblue} 
and \ref{step:compactred}, 
its reverse selector gate denoted $g'$ 
is one that receives a single element as input and outputs two elements; the same control bit $b$ input to the original gate $g$ 
is used by $g'$ to select which of the output 
receives the input element, and the other one will simply receive 
a filler element. \elaine{TO FIX: how is the filler encoded, 
is it part of the payload} 
If $g$ selected the first input element to route to the output, then 
in $g'$, the input element is routed to the first output.
\label{rmk:reverseroute}
\end{Remark}

\begin{fact}
Suppose that $n$ is greater than a sufficiently large constant. 
If a call to ${\bf Swap}^n({\bf X})$ 
does not hit the base case, 
then, in the next recursive call 
to ${\bf Swap}^n({\bf Y})$ in Step~\ref{step:recurse}, 
$m' := |{\bf Y}| \leq \frac{m}{1.9} + \frac{m}{32}$.
Therefore, the recursive 
call will hit the base case after 
at most $\ceil{\log_{c} \frac{1}{\alpha}}$
steps of recursion where $c := 1/(1/1.9 + 1/32) > 1.79$.
\label{fct:lenreduce}
\end{fact}

\begin{fact}
Suppose that $n$ is greater than a sufficiently large constant. 
If the condition $m_{\rm red} \geq m_{\rm blue} + m/64$
is satisfied at the beginning of 
some call ${\bf Swap}^n({\bf X})$, 
then if and when the function makes a recursive
call to ${\bf Swap}^n({\bf Y})$, 
the same condition is satisfied by the array ${\bf Y}$.
\end{fact}
\begin{proof}
If the execution does not trigger the base case,  
since $n$ is greater than a sufficiently large constant, $m$ must
be greater than a sufficiently 
large constant too.

Suppose the inequality is satisfied at the beginning of the recursive call.
Then, after Step~\ref{step:swap}, at most $m/256$ elements are ${\tt blue}$,
and at least $m/64$ elements are ${\tt red}$.
After Step~\ref{step:compactred},
due to the property of the near-sorter,
${\bf Y}_{\rm red}$ has 
at least $(1-\epsilon') \cdot (m/64)$ 
${\tt red}$ elements.
As long as $m$ is greater than some appropriate constant,
in the next recursive call 
to ${\bf Swap}^n({\bf Y})$ in Step~\ref{step:recurse}, 
$m' := |{\bf Y}| \leq \frac{m}{1.9} + \frac{m}{32}$.
Let $m'_{\rm red}$
and  $m'_{\rm blue}$
be the number of ${\tt red}$
and ${\tt blue}$ elements in ${\bf Y}$ respectively.
We have that $m'_{\rm blue} \leq m/256$
and $m'_{\rm red} \geq (1-\epsilon') \cdot (m/64)$.
Therefore, 
\[
\frac{m'_{\rm red} - m'_{\rm blue}}{m'}
\geq \frac{ (1-\epsilon') \cdot (m/64) - m/256 }{\frac{m}{1.9} + \frac{m}{32}}
> 1/64 
\]
\end{proof}

\begin{fact}
Assume that $n$ is greater than a sufficiently large constant.
The remaining number of colored elements at the end of the algorithm 
is at most 
$8 \alpha n + n_{\rm red} - n_{\rm blue}$.
\end{fact}
\begin{proof}
The number of ${\tt blue}$
elements remaining is equal to the total number of ${\tt blue}$ elements
lost during all executions of Step~\ref{step:compactblue},
plus the size of the base case $\alpha n$.
Let $c := 1/(1/1.9 + 1/32)$.
The total number of elements lost
during all executions of Step~\ref{step:compactblue}
is upper bounded by 
$\alpha n + \alpha n / c  + \alpha n / c^2 + \ldots  \leq 3 \alpha n$.
Therefore, the total number of ${\tt blue}$
elements remaining 
is upper bounded by $4\alpha n$.
This means that the total number of 
colored elements remaining 
is at most 
$8 \alpha n + n_{\rm red} - n_{\rm blue}$.
\end{proof}

Clearly, Step~\ref{step:color}
of the algorithm takes only $n$ number of generalized boolean gates.
We now discuss how to implement Step~\ref{step:swapmain}
as a circuit.

\paragraph{Implementing Step~\ref{step:swapmain} in circuit.}
This step is accomplished through recursive calls
to ${\bf Swap}$  
on arrays of length $n' := 
n, n/c, n/c^2, \ldots$, where $c := 1/(1/1.9 + 1/32)$.
The recursion stops when $n' < \alpha n$.
For each length $n'$, we consume an approximate swapper, 
a loose compactor, an $\epsilon'$-near-sorter,
and the reverse-routing circuitry of   
the loose compactor and the $\epsilon'$-near-sorter.
Thus for each problem size $n' = n, n/c, n/c^2, \ldots$, we need 
\begin{MyItemize}
\item 
$S_{\rm sw}(n')$ number of $(w + 1)$-selector gates 
and $B_{\rm sw}(n')$ number of generalized boolean gates
corresponding to Step~\ref{step:swap}; 
\item 
$2 S_{\rm lc}(n')$ number of $(w+2)$-selector gates 
(one for the forward direction
and one for the reverse direction)
and $B_{\rm lc}(n')$
generalized boolean gates 
 corresponding to Step~\ref{step:compactblue}; 

\item 
$O(n')$ number of generalized boolean gates 
and $O(n')$ number of $(w+2)$-selector gates
due to Step~\ref{step:compactred}
and its reverse routing;
and  
\item 
$O(n')$ generalized boolean gates and 
$O(n')$ number of 
$w$-selector gates due to 
Step~\ref{step:output}. 
\end{MyItemize}

Note that each $(w+1)$-selector gate can be realized 
with one $w$-selector gate that operates on the $w$-bit payload 
and one generalized boolean gate that computes on the extra metadata bit;
further, during the reverse routing, 
the metadata generalized boolean gate 
can also be used to save whether each output is a filler.
Thus each problem size $n'$ can be implemented with 
$S_{\rm sw}(n') + 2 S_{\rm lc}(n') + O(n')$
number of $(w+1)$-selector gates and 
$B_{\rm sw}(n') +  B_{\rm lc}(n') + 2 S_{\rm lc}(n') + O(n')$
generalized boolean gates.
Replacing each $(w+1)$-selector gate with a $w$-selector
gate and a generalized boolean gate, we have that 
each problem size $n'$ can be implemented with 
$S_{\rm sw}(n') + 2 S_{\rm lc}(n') + O(n')$
number of $w$-selector gates and 
$S_{\rm sw} (n') + B_{\rm sw}(n') +  B_{\rm lc}(n') + 
4 S_{\rm lc}(n') + O(n')$
generalized boolean gates.

Summing over all $n' = n, n/c, n/c^2, \ldots$, 
we have the follow fact:

\begin{fact}
In the above approximate splitter algorithm, 
the total number of $w$-selector gates 
needed is upper bounded by 
$2.5 S_{\rm sw}(n) + 5 S_{\rm lc}(n) + O(n)$
and the total number of generalized boolean gates
is upper bounded by 
$2.5 S_{\rm sw}(n) + 
2.5 B_{\rm sw}(n) +  2.5 B_{\rm lc}(n) + 10 S_{\rm lc}(n) + O(n)$;
furthermore, the depth 
is upper bounded 
by 
$\log_{1.79} \frac{1}{\alpha} \cdot (2 D_{\rm lc}(n) + O(1))
\leq 2.4 \log \frac{1}{\alpha} \cdot (D_{\rm lc}(n) + O(1))$.
\label{fct:swap}
\end{fact}
\end{proofof}

\section{Lossy Loose Compaction from Approximate Splitter}
\label{sec:lcfromsplitter}

In this section, we show how to construct a circuit for lossy 
loose compaction from an approximate splitter. 

\begin{theorem}
Let $f(n)$ be some function 
in $n$ such that $1 < f(n) \leq \log_2 n$ holds for 
every $n$ greater than an appropriate constant; let $C_{\rm sp} > 1$
be a constant.  Fix some $\alpha \in (0, 1)$ 
which may be a function of $n$.
Suppose that for any $\beta \in (0, 1/4]$, 
for any $n$ that is greater 
than an appropriately large constant,  
$(\alpha, \beta, n, w)$-approximate splitter can be solved
by a circuit 
with $C_{\rm sp} \cdot n \cdot f(n)$ 
generalized boolean 
gates, $C_{\rm sp} \cdot n$ selector gates, 
and of depth $D_{\rm sp}(n)$. 
Then, for any $n$ greater than an appropriately large constant, 
$(1.74 \alpha, n, w)$-lossy loose compaction can be solved by a circuit 
with $2.07 C_{\rm sp} \cdot n \cdot f(f(n)) + O(n)$ generalized boolean
gates, $2.07 C_{\rm sp} \cdot n$ number of $w$-selector gates, 
and of depth $2.07 D_{\rm sp}(n) + O(\log f(n))$.
\elaine{check expressions}
\label{thm:lcfromsplitter}
\end{theorem}

The remainder of this section will be dedicated to proving the above theorem.

\medskip
\begin{proofof}{Theorem~\ref{thm:lcfromsplitter}}
For simplicity, we first consider the case when $n$ is divisible
by $f(n)$. Looking ahead, we will use $f(n)$ to 
be $\log^{(x)}n$ for some $x$ that is power of $2$. 
We will later extend our theorem statement 
to the case when $n$ is not divisible by $f(n)$.  
Consider the following algorithm:

\medskip
\begin{mdframed}
\begin{center}
{\bf Lossy loose compaction from approximate splitter}
\end{center}
\begin{enumerate}[leftmargin=5mm,itemsep=1pt]
\item 
Divide the input array into $f(n)$-sized chunks. 
We say that a chunk is {\it sparse} if there are at most $f(n)/32$ 
real elements in it; otherwise
it is called {\it dense}.
Now, count the number of elements in every chunk, 
and mark each chunk as either ${\tt sparse}$ or ${\tt dense}$.
We will show later in Fact~\ref{fct:sparsechunk} 
that at least $3/4$ fraction of the chunks are sparse.
\label{step:countperchunk}
\item 
Call an $(\alpha, 1/4, n/f(n), w\cdot f(n))$-approximate splitter 
 to move  
almost all 
the dense chunks to the front and almost all the sparse chunks  
to the end.
\label{step:movechunks}

\item 
Apply an $(\alpha, 1/32, f(n), w)$-approximate splitter
to the trailing $\ceil{(\frac{3}{4} - \frac{1}{64}) 
\cdot \frac{n}{f(n)}}$ chunks to compress
each of these chunks to a length of $\floor{\frac{3f(n)}{64}}$,
losing few elements in the process.
The first $\floor{(\frac{1}{4} + \frac{1}{64}) \cdot \frac{n}{f(n)}}$
chunks are unchanged. Output the resulting array.
\elaine{note: now there are loose-end output wires... 
use the other notion of tight compaction}
\label{step:compactperchunk}
\end{enumerate}
\end{mdframed}

At the end of the algorithm, the output array has length at most
\begin{equation}
(\frac{3}{4} -\frac{1}{64}) \cdot \frac{n}{f(n)}
\cdot \frac{3f(n)}{64} 
+ (\frac{1}{4} + \frac{1}{64}) \cdot \frac{n}{f(n)} \cdot f(n) 
\leq 0.32 n < n/1.9 
\label{eqn:compressionrate}
\end{equation}

\begin{fact}
At least ${\frac{3}{4} \cdot \frac{n}{f(n)}}$ chunks are sparse.
\label{fct:sparsechunk}
\end{fact}
\begin{proof}
Suppose not, this means that more than  
${\frac{1}{4} \cdot \frac{n}{f(n)}}$
have more than $f(n)/32$ real elements.
Thus the total number of elements is more than 
$n/128$ \elaine{note more than} 
which contradicts the input sparsity assumption of loose compaction.
\end{proof}

\begin{fact}
The above algorithm loses at most 
$1.74 \alpha n$ real elements. 
\end{fact}
\begin{proof}
If a {\tt dense} chunk is not contained within the first 
 $\floor{(\frac{1}{4} + \frac{1}{64}) \cdot \frac{n}{f(n)}}$ chunks, 
we may assume that all elements in it will be lost.
Due to the property of the approximate splitter,
at most $\alpha n/f(n)$ {\tt dense} chunks 
are not contained within the first
 $\floor{(\frac{1}{4} + \frac{1}{64}) \cdot \frac{n}{f(n)}}$ chunks.
Further, when we apply an approximate splitter to compress 
the trailing  $\ceil{(\frac{3}{4} - \frac{1}{64}) 
\cdot \frac{n}{f(n)}}$ chunks to 
each to a length of $\floor{\frac{3f(n)}{64}}$,
for each chunk, we may lose at most $\alpha f(n)$ real elements.

Therefore, the number of real elements
lost is upper bounded by  
the following as long as $n$ is greater than an appropriate constant:
\[
\alpha \cdot (n/f(n)) \cdot f(n)  
+ \alpha f(n) \cdot \ceil{(\frac{3}{4} - \frac{1}{64}) 
\cdot \frac{n}{f(n)}} 
\leq 1.74 \alpha n 
\]
\end{proof}

\paragraph{Implementing the above algorithm in circuit.}
We now analyze the circuit size of the above algorithm.
For simplicity, we first assume that $n$ is divisible by $f(n)$
and we will later modify our analysis to the more general case when $n$ is not
divisible by $f(n)$.

\begin{enumerate}[leftmargin=5mm,itemsep=1pt]
\item 
Due to Fact~\ref{fct:countprelim}, 
Step~\ref{step:countperchunk}
of the algorithm 
requires 
at most $O(n)$ generalized boolean gates as we have $n/f(n)$ counters 
each for a chunk of size $f(n)$. 
The counting for all chunks are performed in parallel, and thus
the depth is $O(\log f(n))$.


\item 
Step~\ref{step:movechunks} is a single invocation of 
an $(\alpha, 1/4, 
n/f(n),w\cdot f(n))$-approximate splitter. 
Assuming that $(\alpha, 1/4, n,w)$-approximate splitter 
can be realized with $C_{\rm sp}\cdot n \cdot f(n)$ 
generalized boolean gates and $C_{\rm sp} \cdot n$ 
selector gates, this step requires at most 
$C_{\rm sp} \cdot (n/f(n)) \cdot f(n/f(n)) \leq C_{\rm sp} 
\cdot (n/f(n)) \cdot f(n)$
generalized boolean gates and $C_{\rm sp} \cdot n/f(n)$
number of $w\cdot f(n)$-selector gates.
Each such selector gate can in 
turn be realized with $f(n)$ number of $w$-selector gates; 
moreover, the flag bit needs to be replicated $f(n)$ times 
over a binary tree, requiring $O(\log f(n))$ depth and $O(f(n))$ generalized
boolean gates per chunk.
\elaine{note the replication of the flag}
Thus, in total, 
Step~\ref{step:movechunks}
requires $C_{\rm sp} \cdot n + O(n)$ generalized boolean gates,
$C_{\rm sp} \cdot n$ number of $w$-selector gates, 
and requires at most $D_{\rm sp}(n) + O(\log f(n))$ depth.
\item 
Step~\ref{step:compactperchunk}
of the algorithm requires applying 
$\ceil{(\frac{3}{4} - \frac{1}{64}) \cdot \frac{n}{f(n)}}$
number of 
$(\alpha, 1/32, f(n), w)$-approximate splitters, 
where, according to our assumption in Theorem~\ref{thm:lcfromsplitter}, 
each such approximate splitter 
consumes $C_{\rm sp} \cdot f(n) \cdot f(f(n))$ generalized boolean
gates and $C_{\rm sp} \cdot f(n)$ number of $w$-selector gates. 
For sufficiently large $n$ and $f(n) \leq \log_2 n$, we have that 
$\ceil{(\frac{3}{4} - \frac{1}{64}) \cdot \frac{n}{f(n)}} \cdot f(n) \leq  n$.
Therefore, 
in total there are at most 
$C_{\rm sp} \cdot n \cdot f(f(n))$ generalized boolean
gates and $C_{\rm sp} \cdot n$
number of $w$-selector gates.
The depth of this step is upper bounded by $D_{\rm sp}(f(n))$.
\end{enumerate}

Summarizing the above, we have the following fact:

\begin{fact}
Assume the same assumptions as in Theorem~\ref{thm:lcfromsplitter}, and moreover
$n$ is divisible by $f(n)$.
The lossy loose compaction algorithm 
above can be realized with a 
circuit consisting of $C_{\rm sp} 
\cdot n \cdot (f(f(n)) + 1) + O(n)$
generalized boolean gates, 
$2 C_{\rm sp} \cdot n$ number of $w$-selector gates,
and of depth $D_{\rm sp}(n) + D_{\rm sp}(f(n)) + O(\log f(n))$.
\end{fact}

\paragraph{When $n$ is not divisible by $f(n)$.}
When $n$ is not divisible by $f(n)$, we can pad the 
last chunk with filler elements to a 
length of a multiple $f(n)$.  After the padding the total number of elements 
is upper bounded by $n + f(n)$. 
As long as $n$ is greater than an appropriately large constant,
even with the aforementioned padding, we would have the 
following fact:


\begin{fact}
Assume the same assumptions as in Theorem~\ref{thm:lcfromsplitter}.
Then, 
for sufficiently large $n$,
the above lossy loose compaction algorithm 
can be realized with a 
circuit consisting of $2.07 C_{\rm sp} \cdot n \cdot f(f(n)) + O(n)$
generalized boolean gates, 
$2.07 C_{\rm sp} \cdot n$ number of $w$-selector gates, 
and in depth $2.07 D_{\rm sp}(n) + O(\log f(n))$.
\end{fact}

\end{proofof}


\section{Linear-Sized, Low-Depth 
$1/\poly\log(\cdot)$-Lossy Loose Compactor}
\label{sec:optllc}

In this section, we shall prove the following theorem. 

\begin{theorem}[Linear-sized loose compactor]
Let $\widetilde{C}> 1$ be an arbitrary constant.
There exists a 
circuit 
in the indivisible model 
that solves $(1/\log^{\widetilde{C}}(n), n, w)$-lossy loose compaction, and moreover 
the circuit depth is $O(\log^{0.5} n)$,  \elaine{is this tight enough}
the total number of generalized boolean gates
is upper bounded by 
 $O(n \cdot w) \cdot 
\max\left(1,\poly(\log^* n - \log^* w)\right)$, and the number 
of 
$w$-selector gates is upper bounded by  
 $O(n) \cdot 
\max\left(1,\poly(\log^* n - \log^* w)\right)$.

As a direct corollary, for any arbitrarily large constant $c \geq 1$, 
if $w \geq \log^{(c)} n$, it holds that 
the number of generalized boolean gates 
is upper bounded by $O(nw)$, and the number of  
$w$-selector gates is  
upper bounded by $O(n)$.
\label{thm:linearllc}
\end{theorem}

The case when $w > \log \log n$ is easier (see Footnote~\ref{fnt:bigw}), 
so in the remainder
of this section, unless otherwise noted, 
we shall assume that $w \leq \log \log n$.


\medskip
\begin{proofof}{Theorem~\ref{thm:linearllc}}
We will construct tight compaction through 
repeated bootstrapping and boosting.
Without loss of generality, we may assume that $n$ 
is greater than an appropriately large constant.
We have two steps:

\begin{itemize}[leftmargin=5mm,itemsep=1pt]
\item {\bf \boldmath $\LC_i \Longrightarrow \SP_{i+1}$ 
(Theorem~\ref{thm:splitterfromlc}):}
from lossy loose compactor to 
approximate splitter. 
Due to 
Theorem~\ref{thm:approxswap}
and 
Theorem~\ref{thm:splitterfromlc}, we get the following, 
where we use different subscripts in the big-O notations
to hide different constants.
\begin{quote}
Assuming $(\alpha, n,w)$-lossy loose compactor with:
\[
\begin{array}{rl}
\text{\rm \# generalized boolean gates}: &  B_{\rm lc}(n) 
\\
\text{\rm \# selector gates}: &  S_{\rm lc}(n) \\
\text{\rm depth}: &  D_{\rm lc}(n) \\
\end{array}
\] 
Then, for any $\beta \in (0, 1/4]$, 
there exists $(8\alpha, \beta, n,w)$-approximate splitter with:
\[
\begin{array}{rl}
\text{\rm \# generalized boolean gates}: &  
2.5 B_{\rm lc}(n) + 10 S_{\rm lc}(n) + O_1(n)
\\
\text{\rm \# selector gates}: &  
5 S_{\rm lc}(n) + O_2(n) \\
\text{depth}: & 
2.4 \log{\frac{1}{\alpha}} \cdot (D_{\rm lc}(n) + O_3(1))
\end{array}
\] 
\end{quote}
\item {\bf \boldmath $\SP_{i+1} \Longrightarrow \LC_{i+1}$ (Theorem~\ref{thm:lcfromsplitter})}: from approximate splitter to 
lossy loose compactor. Simplifying  Theorem~\ref{thm:lcfromsplitter} 
we have:
\begin{quote}
Fix some $\alpha \in (0, 1)$.
Assuming that for 
any $\beta \in (0, 1/4]$,
$(\alpha, \beta, n,w)$-approximate splitter 
can be realized 
in a circuit with the following cost
for some constant $C_{\rm sp}$ and function $f(n)$:
\[
\begin{array}{rl}
\text{\rm \# generalized boolean gates}: &  C_{\rm sp}\cdot n\cdot f(n)
\\
\text{\rm \# selector gates}: &  C_{\rm sp}\cdot  n  \\
\text{depth}: & D_{\rm sp}(n)\\
\end{array}
\] 
Then there exists a $(1.74 \alpha, n,w)$-lossy 
loose compactor such that:
\[
\begin{array}{rl}
\text{\rm \# generalized boolean gates}: &  2.07\cdot C_{\rm sp}\cdot n\cdot f(f(n)) + O_4(n)
\\
\text{\rm \# selector gates}: &  2.07\cdot C_{\rm sp} \cdot n 
\\
\text{depth}: & 2.07 \cdot D_{\rm sp} (n) + O_5(\log f(n))
\end{array}
\] 
\end{quote}
\end{itemize}

Choose $C_0 := \widetilde{C}+ 1$. 
Our starting point is Theorem~\ref{thm:initllc}, 
which gives as a circuit $\LC_0$
that realizes $(1/\log^{C_0} n, n, w)$-lossy loose compaction
for the constant $C_0 > 1$.
Using the 
above two steps, we bootstrap and boost the circuit:
\begin{itemize}
\item[$\LC_0$:]
By Theorem~\ref{thm:initllc}, there exists
a constant $C > 1$ 
such that we can solve $(1/\log^{C_0} n, n, w)$-lossy 
loose compaction with
\[
\begin{array}{rl}
\text{\rm generalized boolean gates}: &  C n \log \log n
\\
\text{\rm selector gates}: &  C n \\
\text{\rm depth}: & C \log \log n\\
\end{array}
\]

\item[$\SP_1$:] By Theorem~\ref{thm:splitterfromlc}, 
for any $\beta \in (0, 1/4]$, 
we can 
construct an $(8/\log^{C_0} n, \beta, n, w)$-approximate 
splitter circuit $\SP_1$ from $\LC_0$.
$\SP_1$'s size is upper bounded by the 
expressions\footnote{\label{fnt:bigw} When 
$w > \log \log n$, $\LC_0$ gives Theorem~\ref{thm:linearllc}. 
Therefore, the rest of this section assumes $w \leq \log \log n$.
}:
\[
\begin{array}{rl}
\text{\rm generalized boolean gates}: &  
2.5 C n \log n + 10 C n + O_1(n) \leq 5.1 C n \log n 
\\
\text{\rm selector gates}: &  5 C n + O_2(n) \leq 5.1 C n \\
\text{\rm depth}: & 2.4 C_0 \log \log n \cdot (C \log \log n + O_3(1))
\leq 2.5 C_0 \log \log n \cdot (C \log \log n)
\\
\end{array}
\] 
In the above, 
the inequalities hold as long as $n$ is greater than an appropriately
large constant.

\item[$\LC_1$:] 
Due to Theorem~\ref{thm:lcfromsplitter},
we build
a $(8 \cdot 1.74/\log^{C_0}n, n, w)$
 lossy loose compaction circuit $\LC_1$ from $\SP_1$.
 $\LC_1$'s size is upper bounded by the expressions:
\[
\begin{array}{rl}
\text{\rm generalized boolean gates}: &  
2.07 \cdot 5.1 C n \log \log n + O_4(n) 
\leq 2.1 \cdot 5.1 Cn \log \log n \\[2pt]
\text{\rm selector gates}: &  
2.07 \cdot 5.1 Cn  \leq 2.1 \cdot 5.1 Cn\\[2pt]
\text{depth}: & 
\begin{array}{ll}
& 2.07 \cdot (2.5 C_0 \log \log n) \cdot (C \log \log n) + O_5(\log \log n)
\\
\leq & 2.1 \cdot (2.5 C_0 \log \log n) \cdot (C \log \log n) 
\end{array}
\end{array}
\] 
\item[$\SP_2$:] 
Due to Theorem~\ref{thm:splitterfromlc},
for any $\beta \in (0, 1/4]$, 
we can construct 
a $(8^2 \cdot 1.74 / \log^{C_0} n, \beta, n, w)$-approximate splitter
circuit $\SP_2$ from $\LC_1$.
$\SP_2$'s size is upper bounded by the expressions:
\[
\begin{array}{rl}
\text{\rm generalized boolean gates}: &  
2.5 \cdot 2.1 \cdot 5.1 Cn \log \log n  + 
10 \cdot 2.1 \cdot 5.1 Cn + O_1(n) \leq 
2.1 \cdot (5.1)^2 C n \log \log n
\\[2pt]
\text{\rm selector gates}: &  
5 \cdot 2.1 \cdot 5.1 Cn  + O_2(n) \leq 2.1 \cdot (5.1)^2 Cn 
\\[2pt]
\text{depth}: & 
\begin{array}{ll}
& 2.4 C_0 \log \log n \cdot ( 2.1 \cdot (2.5 C_0 \log \log n) \cdot (C \log \log n) 
+ O_3(1))\\
\leq 
& 2.1 \cdot (2.5 C_0 \log \log n)^2 \cdot (C \log \log n) 
\end{array}
\end{array}
\] 
\item[$\LC_2$:]
Due to Theorem~\ref{thm:lcfromsplitter},
we build
a $((8 \cdot 1.74)^2 /\log^{C_0}n, n, w)$-lossy 
loose compaction circuit $\LC_2$ from $\SP_2$.
 $\LC_2$'s size is upper bounded by the expressions:
\[
\begin{array}{rl}
\text{\rm generalized boolean gates}: &  
2.07 \cdot 2.1 \cdot (5.1)^2 C n \log^{(4)}n + O_4(n) 
\leq (2.1 \cdot 5.1)^2 Cn \log^{(4)}n 
\\
\text{\rm selector gates}: &   
2.07 \cdot 2.1 \cdot (5.1)^2 Cn  \leq (2.1 \cdot 5.1)^2 Cn
\\
\text{\rm depth}: & 
\begin{array}{ll}
& 2.07 \cdot 2.1 \cdot (2.5 C_0 \log \log n)^2 \cdot (C \log \log n) 
+ O_5(\log^{(3)}(n))  
\\
\leq & (2.1 \cdot 2.5 C_0 \log \log n)^2 \cdot (C \log \log n)   
\end{array}
\end{array}
\] 
\end{itemize}

Let $d \in \N$ be the smallest integer
such that $\log^{(2^d)}n \leq w$, i.e., 
$d = \ceil{\log(\log^* n - \log^* w)} \leq \log(\log^* n - \log^* w) + 1$.
Continuing for $d$ iterations, we get:
\begin{itemize}
\item[$\LC_{d}$:]  
$\LC_d$ is a 
$((8\cdot 1.74)^d / \log^{C_0} n, n, w)$-lossy loose compactor, 
and $\LC_{d}$'s size is upper bounded by the expressions:
\[
\begin{array}{rl}
\text{\rm generalized boolean gates}: &  
(2.1 \cdot 5.1)^{d} Cn \log^{(2^{d})}n 
= O(nw) \cdot \poly(\log^*n - \log^* w) \\
\text{\rm selector gates}: &   
(2.1 \cdot 5.1)^{d} Cn
= O(n) \cdot \poly(\log^*n - \log^* w)\\
\text{\rm depth}: & (2.1 \cdot 2.5 C_0 \log \log n)^{d} 
 \cdot (C \log \log n) 
\leq O(\log^{0.5} n)
\end{array}
\] 
\end{itemize}
This gives rise to Theorem~\ref{thm:linearllc}.
\end{proofof}

\section{Approximate Tight Compaction}
\label{sec:approxtc}

\paragraph{Definition.}
Let $\alpha \in (0, 1)$.
An $(\alpha, n, w)$-approximate tight compactor
(also written as $\alpha$-approximate tight compactor
when $n$ and $w$ are clear from the context)
solves the following problem:
given an input array ${\bf I}$ containing $n$ elements, 
each containing  
a $1$-bit key and a $w$-bit payload, 
we want to output 
a permutation (denoted ${\bf O}$) of the input array ${\bf I}$, such that 
at most $\alpha \cdot n$ elements 
in ${\bf O}$
are misplaced  --- here, an element ${\bf O}[i]$ is said
be misplaced iff ${\bf O}[i]$ is marked
with the key $b \in \{0, 1\}$; however, 
the sorted array ${\sf sorted}({\bf I})$
wants the key $1-b$ in position $i$. 

\begin{theorem}[Approximate tight compaction]
Fix an arbitrary constant $\widetilde{C} > 1$.
There is an $(1/(\log n)^{\widetilde{C}}, n, w)$-approximate 
tight compaction circuit  
that has 
 $O(n \cdot w) \cdot 
\max\left(1,\poly(\log^* n - \log^* w)\right)$
generalized boolean gates, 
$O(n) \cdot 
\max\left(1,\poly(\log^* n - \log^* w)\right)$
number of $w$-selector gates, and 
with depth at most $O(\log n)$.
\label{thm:approxtc}
\end{theorem}

\begin{proofof}{Theorem~\ref{thm:approxtc}}
Given an $(\alpha, n, w)$-lossy loose compactor, we can obtain a
$(8\alpha, n, w)$-approximate tight compactor  \elaine{check param}
using 
an algorithm that is similar to the one described
in the proof of Theorem~\ref{thm:splitterfromlc}.
For convenience, below 
we shall refer to the elements with the $0$-key
in the input array ${\bf I}$ as {\it distinguished}.

\begin{mdframed}
\begin{center}
{\bf Approximate tight compaction from lossy loose compaction} 
\end{center}
\begin{enumerate}[leftmargin=5mm,itemsep=1pt]
\item
{\it Count.}
Compute the total number (denoted ${\sf cnt}$) of 
distinguished elements 
in the input array ${\bf I}$.
\label{step:count-tc}
\item
{\it Color.}
For any $i \leq {\sf cnt}$, if ${\bf I}[i]$ 
is not distinguished, mark the element ${\tt red}$;
for any $i > {\sf cnt}$, if ${\bf I}[i]$ is
distinguished, mark the element ${\tt blue}$; every other element is marked $\bot$.
Let the outcome be ${\bf X}$.

Note that at this moment, each element
is labeled with 3 bits of metadata, one bit of distinguished indicator
and two bits of color-indicator (indicating whether the element is colored,
and if so, which color).
\label{step:color-tc}
\item 
{\it Swap.} \label{step:swap-tc}
Call 
$\widehat{\bf Swap}^n({\bf X})$
(to be defined below)
to swap almost all the ${\tt blue}$ elements
each with a ${\tt red}$ element --- 
here we use a payload of size $w+1$ and not $w$ as we also include 
the distinguished-indicator as part of the payload.
Return the outcome.
\end{enumerate}
\end{mdframed}

$\widehat{\bf Swap}^n({\bf X})$
is defined 
in a very similar to the ${\bf Swap}^n$
algorithm of Theorem~\ref{thm:splitterfromlc}; except
that now, we simply use a lossy loose compactor to extract
the residual ${\tt red}$ and ${\tt blue}$ elements, 
and then recurse
on the extracted array. In comparison, in the earlier 
${\bf Swap}^n$ algorithm, we used a lossy loose compactor
to extract ${\tt blue}$ elements and used 
a near-sorter to extract the ${\tt red}$ elements.

\begin{mdframed}
\begin{center}
$\widehat{\bf Swap}^n({\bf X})$
\end{center}
\begin{itemize}[leftmargin=5mm,itemsep=0pt]
\item {\bf Input:} An array {\bf X} of $m \leq n$ 
elements, each has a $w$-bit payload\footnotemark  
and a 2-bit label indicating whether the element is colored,
and if so, whether the element is ${\tt blue}$ or ${\tt red}$.
$n$ is the size of the original problem when 
${\bf Swap}$ is first called; the same $n$   
will be passed into all recursive 
calls since it is used to decide
when the recursion stops.  
\footnotetext{Our approximate tight compaction algorithm 
actually requires a swapper where elements are of bit-length $w+1$, 
but for convenience
we rename the variable to $w$ in the description of the swapper.
}

\item {\bf Algorithm:}
\begin{enumerate}[leftmargin=5mm,itemsep=1pt,label=(\alph*)] 
\item 
{\it Base case.}
Same as Step~\ref{step:splitter:base} of the earlier ${\bf Swap}^n$ of Theorem~\ref{thm:splitterfromlc}.
\item 
{\it Approximate swapper.}
Same as Step~\ref{step:swap} of the earlier ${\bf Swap}^n$ of Theorem~\ref{thm:splitterfromlc};
recall that the resulting array is denoted as ${\bf X}'$.
\ignore{
Call an $(m, w)$-approximate swapper (see Theorem~\ref{thm:approxswap}) 
on ${\bf X}$ to swap elements of opposite colors 
and uncolor them in the process,
such that at most $m/128 + m_{\rm red} - m_{\rm blue}$ 
elements remain colored. 
Let the outcome be called ${\bf X}'$.
\label{step:swap}
}
\item 
{\it Lossy-extract colored.}
Call an $(\alpha, m, w+1)$-lossy loose 
compactor 
to compact ${\bf X}'$ by a half, 
where the lossy loose compactor treats 
the 
colored 
elements as real, and all other elements as fillers 
(i.e., the loose compactor treats
the first bit of the color label as a real-filler indicator,
and the second bit of the color label 
is treated as part of the payload).

Let the outcome be ${\bf Y}$ 
whose length is half of ${\bf X}$.
\label{step:compactcolored}
\ignore{
\item 
\label{step:compactred}
{\it Extract red.}
Let $\epsilon' = 1/2^{10}$. Apply an $\epsilon'$-near-sorter
to the array ${\bf X}'$ 
treating all ${\tt red}$ elements as 
smaller than all other elements.
Let ${\bf Y}_{\rm red}$   
be the first $\floor{m/32}$ elements of the 
resulting near-sorted array.
Mark every non-${\tt red}$ element in ${\bf Y}_{\rm red}$ as 
uncolored,
and let
${\bf Y} := {\bf Y}_{\rm red} || {\bf Y}_{\rm blue}$.
}
\item 
{\it Recurse.}
Recursively call $\widehat{\bf Swap}^n({\bf Y})$, 
and let the outcome be ${\bf Y}'$.
\label{step:recurse-tc}
\item 
{\it Reverse route.}
Reverse the routing decisions made by all selector gates 
during Steps~\ref{step:compactcolored} 
(see Remark~\ref{rmk:reverseroute} in the
proof of Theorem~\ref{thm:splitterfromlc}).
In this way, we can reverse-route elements in ${\bf Y}'$ 
to an array (denoted $\widetilde{\bf X}$) 
whose length is $m$.
\label{step:reverse-tc}

\item 
{\it Output.}
\label{step:approxtc:output}
Return 
${\bf O}$ which is formed by a performing coordinate-wise 
select operation between ${\bf X}'$
and $\widetilde{\bf X}$.
For every $i \in [m]$:
\begin{itemize}[leftmargin=5mm,itemsep=1pt]
\item 
if ${\bf X}'[i]$ originally had a colored element
and the element was not lost during Step~\ref{step:compactblue},
then let ${\bf O}[i] := \widetilde{\bf X}[i]$;
\item 
else let ${\bf O}[i] := {\bf X}'[i]$;
\end{itemize}

\ignore{back to the input ${\bf X}$, such that each position 
$i$ receives either a real or a dummy element: if a real element
is received the element ${\bf X}[i]$ is overwritten with the received element;
otherwise the element ${\bf X}[i]$ is unchanged.
Finally, output the result.
}
\end{enumerate} 
\end{itemize}
\end{mdframed}

Suppose that $n$ is sufficiently large. Then, 
the recursive call will hit the base case after
at most $\ceil{\log_{1.9} \frac{1}{\alpha}}$
steps of recursion.

\begin{fact}
Assume that $n$ is greater than a sufficiently large constant.
The remaining number of colored elements at the end of the algorithm 
is at most $8 \alpha n$.
\end{fact}
\begin{proof}
The total number of elements lost during
Step~\ref{step:compactcolored}
of the algorithm 
is upper bounded by $\alpha n + \alpha n / 1.9 + \alpha n / 1.9^2 
\ldots  \leq 3 \alpha n$.
Also, the recursion stops when $m \leq \alpha n$, 
all remaining colored elements 
will not get swapped.
Therefore, the total number of colored elements
remaining at the end is upper bounded by 
$2 \cdot 3\alpha n + \alpha n < 8\alpha n$,
where the factor 2 comes from the fact that we may lose all $3 \alpha n$ 
in {\tt blue} color and thus there are another $3 \alpha n$ in {\tt red}.
\end{proof}

\paragraph{Implementing Steps~\ref{step:count-tc} and
\ref{step:color-tc} in circuit.} 
Due to Fact~\ref{fct:countprelim}, 
Step~\ref{step:count-tc}
can be accomplished with $O(n)$ generalized boolean gates
and in depth $O(\log n)$.
When the count ${\sf cnt}$ is computed 
from Step~\ref{step:count-tc}, we can implement
Step~\ref{step:color-tc} as follows. 
Recall that ${\sf cnt} \in \{0,1, \dots,n\}$ 
is a $((\log_2 n) + 1)$-bit number.
Imagine that there are $n$ receivers numbered 
$1, 2, \ldots, n$. Each receiver is waiting to receive  
either ``$\leq$'' or ``$>$''.
Those with indices $1, \ldots, {\sf cnt}$ should receive ``$\leq$'' and 
those with indices ${\sf cnt} + 1, \ldots, n$ should receive ``$>$''.
We can accomplish this 
using the binary-to-unary conversion circuit of 
 Fact~\ref{fct:binary_to_unary},
i.e., convert ${\sf cnt}$ 
into an $n$-bit string so that the head ${\sf cnt}$ bits are 0
and the tail $n-{\sf cnt}$ bits are 1. 
Due to Fact~\ref{fct:binary_to_unary},
Step~\ref{step:color-tc} can be implemented
as a circuit consisting of at most $O(n)$  
generalized boolean gates and in depth $O(\log n)$.
Once each of the $n$ receivers receive either ``$\leq$'' or ``$>$'',  
it takes a single generalized boolean gate 
per receiver 
to write down either ${\tt blue}$, ${\tt red}$, or uncolored.

\paragraph{Implementing Steps~\ref{step:swap-tc} in circuit.} 
The approach and analysis are 
similar to
the 
${\bf Swap}^n$ circuit 
in the proof of Theorem~\ref{thm:splitterfromlc}.

Summarizing the above, and plugging in a 
$(1/8\log^{\widetilde{C}} n, n, w)$-lossy loose compactor
as stated in Theorem~\ref{thm:linearllc}, 
we will get Theorem~\ref{thm:approxtc}.
\end{proofof}

\ignore{TODO: explain somewhere, maybe in roadmap
why we do not use the lc to tc to lc bootstrap}

\section{Sparse Loose Compactor}
\label{sec:sparselc}

\subsection{Building Blocks: Slow Tight Compaction and Distribution}
\label{sec:slowtc-distr}
\begin{lemma}[Slow tight compaction circuit ${\bf SlowTC}$]
There is an $(n, w)$-tight compaction 
circuit of depth $O(\log^2 n)$,
and requiring 
$O(nw + n\log n)$
generalized boolean gates
and  $O(n)$
number of $w$-selector gates.
Henceforth we will use ${\bf SlowTC}$ to denote this circuit.
\label{thm:slowtc}
\end{lemma}
\begin{proof}
We can use the tight compactor
circuit constructed in Asharov et al.~\cite[Theorems 4.8 and 5.1]{soda21}.
In particular, wherever they employ an approximate swapper
(called a loose swapper in their paper~\cite{soda21}), we
replace its implementation with a constant-depth one as described
in Theorem~\ref{thm:approxswap}.
Asharov et al.~\cite{soda21}
did not analyze the depth of the circuit; however, with this
modification, it is not hard
to show that the resulting circuit has depth
upper bounded by  
$O(\log^2 n)$.
\elaine{todo: double check this}
\ignore{
Use the approximate tight compaction algorithm
of Theorem~\ref{thm:approxtc} but with the following modifications:
\begin{enumerate}[leftmargin=5mm,itemsep=1pt]
\item 
Now, we instead use ${\bf SlowLC}$ of Theorem~\ref{thm:slowlc} 
to extract the colored elements without incurring any loss.
\item  

\end{enumerate}
}
\end{proof}


Recall
that in Section~\ref{sec:prambldgblock}, we showed how to construct
an algorithm that accomplishes distribution 
from tight compaction.
The same algorithm applies in the circuit model.
This gives rise to the following corollary:

\begin{corollary}[Slow distribution circuit ${\bf SlowDistr}$]
\label{cor:slowdist}
There is a circuit that solves the aforementioned distribution
problem, henceforth denoted ${\bf SlowDistr}$; further,
the number of generalized boolean 
gates, $w$-selector gates, and depth asymptotically match
the ${\bf SlowTC}$ circuit of 
Theorem~\ref{thm:slowtc}.
\end{corollary}
\begin{proof}
Use the above algorithm where tight compaction
is instantiated with ${\bf SlowTC}$. 

\elaine{use the same notation earlier}
\end{proof}

\subsection{Sparse Loose Compactor}

\paragraph{Sparsity of an array.}
Let $A$ be an array in which each element has a $w$-bit payload, and is tagged
with a bit denoting  
whether the element is {\it real} or a {\it filler}.
Let $\alpha \in (0, 1)$. An array $A$ of length $n$ is said to 
be $\alpha$-sparse
if there are at most $\alpha n$ real elements in it.

\paragraph{Sparse loose compactor.}
A sparse loose compactor is defined almost in the same way
as a loose compactor (see Section~\ref{sec:llc}), except
that 1) it works only on $1/\poly\log(n)$-sparse arrays;
and 2) it compresses the array by $1/\log n$ factor  
without losing any real elements.

More formally, let $C_{\smalltriangledown} > 8$
be a sufficiently large universal constant.
Given an input array ${\bf I}$ of length $n$ 
that is promised to be $1/(\log n)^{C_\smalltriangledown}$-sparse,
an $(n, w)$-sparse loose compactor 
outputs an array ${\bf O}$ whose length is $\floor{n/\log n}$,
and moreover, the multiset of real elements
in ${\bf O}$ must be equal to the multiset of real elements
in ${\bf I}$.

In the remainder of the section, we will prove the following theorem:
\begin{theorem}[Sparse loose compactor]
There is an $(n,w)$-sparse loose compactor 
circuit, 
with 
$O(nw)$ 
generalized boolean gates  
and 
$O(n)$ 
number of $w$-selector gates, and of depth $O(\log n)$.
\elaine{check expr}
\label{thm:lcsparse}
\end{theorem}

\subsection{Intuition}
\label{sec:roadmap:sparseloose}
Now that we have a $1/(\log n)^C$-approximate tight compactor
with $O(nw) \cdot \poly(\log^* n - \log^* w)$ 
boolean gates and $O(\log n)$ depth,
we can apply it to the input array, such that all but 
$1/\poly\log n$ fraction of the elements are in the correct place.
Next, we 
want to extract the $1/\poly\log n$ fraction of misplaced elements
to an array of length at most $\Theta(n/\log n)$. 
If we can accomplish this, we can then use
AKS to swap every misplaced $0$ 
with a distinct misplaced $1$ in the extracted short array, 
and reverse route the result back.

Therefore, the crux is how to solve the sparse loose compaction problem,
that is, we want to extract 
the $1/\poly\log n$ fraction of misplaced elements
to an output array of a fixed length of $\floor{n/\log n}$; 
besides containing the misplaced elements, the output 
array is otherwise padded with filler elements.

\paragraph{Bipartite expander graphs with poly-logarithmic degree.}
We are inspired by the loose compactor construction
of Asharov et al.~\cite{soda21}
which in turn builds on Pippenger's self-routing
superconcentrator~\cite{selfroutingsuperconcentrator}.
Asharov et al.~\cite{soda21}'s construction relies on 
$d$-regular bipartite expander graph with constant degree $d$ and 
constant spectral expansion $\epsilon \in (0, 1)$.
We will instead need a bipartite expander
graph with $m$ vertices on the left
and $m$ vertices on the right,  
where each vertex has degree $d = \log^{c_1} m$.
The spectral expansion of the graph is $\epsilon := 1/\log^{c_2} m$. 
In the above, $c_1 > c_2 > 1$, and both $c_1$ and $c_2$ are 
suitable constants. 
Such a bipartite expander graph can be constructed
using standard techniques.
As we shall see later, using a polylogarithmic degree
bipartite expander graph introduces additional 
complications to the algorithm in comparison with  
earlier works~\cite{soda21,selfroutingsuperconcentrator}.

\paragraph{Intuition.}
Given such a polylogarithmic-degree bipartite expander graph, 
where $L$ denotes the left vertices and $R$ denotes the right vertices,
we construct a sparse loose compactor as follows.
Throughout, our algorithm will operate on {\it super-elements}
rather than elements, where each super-element
contains $\log n$ consecutive elements in the input array.
Each super-element is real if it contains at least one real element.
If the fraction of real elements 
in the input is at most $1/(\log n)^C$, then 
the fraction of real super-elements is at most $1/(\log n)^{C-1}$.
Henceforth let $n' := n/\log n$ denote the number of super-elements.

We divide the input array
into chunks each containing only $d/2$ super-elements. 
Henceforth let $m = 2n' /d$ 
be the number of chunks. For simplicity, we assume that 
the numbers $\log n$, $n/\log n$, and $2n'/d$ 
are integers in this informal overview, and we will deal
with rounding issues in the formal technical sections.
We will think of each of the $m$ chunks as a left vertex
in the bipartite expander graph.
If the chunk contains 
at most $d/(2\log^2 m)$ real super-elements, it is said to be sparse;
else it is said to be dense. 

At a very high level, the idea is for all the 
dense vertices on the left 
to distribute its load to 
the right vertices, such that each right vertex
receives no more than $d/(2\log^2 m)$ real super-elements.
After the load distribution step, we empty  
all real super-elements from the dense chunks; and now all vertices
on the left and right are sparse chunks. We now  
compress each  
left and right chunk to $1/\log^2 m$ of its original size without
losing any real super-elements. 
This would compress the array by a $\Theta(1/\log^2 m)$ factor.
\elaine{is this the right calculation?}

\paragraph{Offline phase.}
The load distribution step consists of an offline phase and 
an online phase. The offline phase looks at only 
the real/filler indicator
of each super-element, and does not 
look at the payloads. 
The goal of the offline phase is to output 
a matching $M$ between the left vertices $L$
and the right vertices $R$, such that each dense chunk on the left
has $d/2$ neighbors in the matching $M$, and each right 
vertex has no more than $d/2 \log^2 m$ neighbors in $M$.
If such a matching can be found, 
then during the online phase, each dense chunk can 
route up to $d/2$ super-elements each along a distinct edge
in the matching $M$ to a right vertex. 

To find the matching, we use the 
${\sf ProposeAcceptFinalize}$ algorithm first proposed 
by Pippenger~\cite{selfroutingsuperconcentrator}.
For convenience, a left vertex is called
a {\it factory} and a right vertex
is called a {\it facility}.

\vspace{3pt}
Initially, each factory corresponding to a dense chunk 
is {\it unsatisfied} and 
each factory corresponding to a sparse chunk is {\it satisfied}.
Each productive factory $u \in L$ has 
at most $d/2$ real super-elements. 
Now, repeat the following for 
${\sf iter} := \log n'/\log \log n'$ times
and output the resulting matching $M$ at the end:

\begin{enumerate}[label=(\alph*),itemsep=0pt,topsep=2pt]
\item {\it Propose:} 
Each unsatisfied factory 
sends a proposal (i.e., the bit 1) to each one of its neighbors. 
Each satisfied factory sends 0 to each one of its neighbors. 

\item {\it Accept:}
If a facility $v \in R$
received no more than $d/(2\log^2 m)$ proposals, 
it sends an acceptance message to each one of its $d$ neighbors;
otherwise, it sends a reject message along each of its $d$ edges. 

\item  {\it Finalize:}
Each currently unsatisfied factory $u \in L$
checks if it received at least $\frac{d}{2}$ acceptance messages.
If so, for each edge over which an acceptance message is received,
mark it as part of the matching $M$.
At this moment, this factory becomes satisfied.
\end{enumerate}

In our subsequent formal sections, we will use the Expander Mixing Lemma 
(see Lemma~\ref{lem:expander}
of Appendix~\ref{sec:expanderprelim})
to prove that 
in each iteration of the above 
${\sf ProposeAcceptFinalize}$ algorithm,
at most $32/\log^4 m$ fraction of the unsatisfied factories
remain unsatisfied at the end of the iteration 
(Lemma~\ref{lem:remainunsat}). 
Therefore, one can show that after $\log n'/\log \log n'$
iterations, all factories become satisfied. 
Note that each iteration takes $O(\log d) = O(\log \log n)$
depth (this is needed
for tallying how many proposals
or acceptance messages a vertex has received), 
and therefore the total depth is only $O(\log n)$.
One crucial observation is that 
the number of edges in the bipartite group
is within a constant factor of 
the number of super-elements, which is $O(n/\log n)$.
In this way, 
over all $\log n /\log \log n$ iterations of the offline phase,
the number of generalized boolean gates is upper bounded by $O(n)$.

Finally, like in prior work~\cite{alm90,selfroutingsuperconcentrator}, 
it is not hard to show that each facility on the right
will 
be matched with at most $d/(2\log^2 m)$ factories.

\paragraph{Online routing phase.}
Each dense chunk  
wants to route each of its up to $d/2$ real super-elements  
along a distinct edge in the matching $M$ to the right.
The challenge is that we need to accomplish this using a linear number of
gates, i.e., each chunk is allowed 
to consume $O(d \cdot w)$ gates (ignoring $\poly\log^*$ terms). 
In comparison, 
in prior works~\cite{soda21,selfroutingsuperconcentrator}, this was
a non-issue because their chunks were constant in size.

We accomplish this by leveraging a tight 
compaction circuit\footnote{In fact, in our formal technical
sections, we will define a slight variant of tight compaction
called ``distribution'' to accomplish the online routing --- 
see Sections~\ref{sec:prambldgblock} and \ref{sec:slowtc-distr}.}
that is optimal in size, but not so optimal in depth --- since each chunk
is small. In fact, to achieve this, we can 
use the tight compaction circuit by Asharov et al.~\cite{soda21},
but replace some its building blocks 
with parallel versions (see Theorem~\ref{thm:slowtc} for more details).
The resulting tight compaction circuit has depth 
that is super-polylogarithmic in the input length,  
but when applied to a chunk of $\poly\log n$ size,
the depth would be 
upper bounded by $O(\log n)$.

\paragraph{Compressing all chunks.}
Now that we have finished the load distribution phase,
all chunks on the left and right
must be sparse. We therefore compress each chunk to $1/\log^2 m$ of
its original size. This can be done by applying 
to each chunk a tight compaction 
circuit that is optimal in work but not optimal in depth (same as 
the building block we used in the online routing phase).

After this, the input 
is compressed to $1/\log^2 m$ of its original size, without losing any
real elements.

\subsection{Proof of Theorem~\ref{thm:lcsparse}}
We will run a variant of the lossy loose compactor  
algorithm described in the proof of Theorem~\ref{thm:initllc}
in Section~\ref{sec:llc}.

\paragraph{Bipartite expander graphs with polylogarithmic degree.} 
Recall the bipartite graph of Margulis~\cite{margulis1973}. Fix a positive
  $t \in \N$.  The left and right vertex sets are 
$L=R :=[t]\times [t]$. A left vertex 
  $(x,y)$ is connected to the right vertices 
$(x,y),(x,x+y),(x,x+y+1),(x+y,y),(x+y+1, y)$ where all
  arithmetic is modulo $t$. 
We let $H_m$ be the resulting graph that has $m = t^2$
  vertices on each side.

  It is known (Margulis~\cite{margulis1973}, Gabber and Galil\cite{GabberG81},
  and Jimbo and Maruoka~\cite{JimboM87}) that for every $m$ which is a perfect square
  (i.e., of the form $m=i^2$ for some $i\in \N$), $H_m$ is $5$-regular and
the second largest eigenvalue  
of its normalized adjacency matrix 
$\lambda_2(H_m) \in (1/5,1)$ is a constant. 
Let $\epsilon := 1/\log^4 m$. 
We will use a graph $G_{\epsilon, m} := H_m^\gamma$ 
that is the $\gamma$-th power of $G_m$, where 
$\gamma$ is the smallest odd integer such that 
$\lambda_2(G_{\epsilon, m}) = \lambda_2(H_m)^\gamma \leq \epsilon$.
In other words, in $G_{\epsilon, m}$, the edges
are the length-$\gamma$ paths in $H_m$.
Therefore, 
$G_{\epsilon, m}$ is a $5^\gamma$-regular bipartite graph.
Note that the degree 
$5^\gamma \in [\log^c m, 25 \log^c m]$ 
for some constant $c > 4$ (where any constant $c>4$ works later).


\paragraph{Sparse loose compactor algorithm.}
We first describe the modifications to the  
meta-algorithm 
on top of the lossy loose compactor algorithm 
in the proof of Theorem~\ref{thm:initllc}.
We then described the modified circuit implementation 
of the meta-algorithm.

\begin{mdframed}
\begin{center}
{\bf Sparse loose compactor}
\end{center}

\paragraph{Expander graph family and parameters.}
We use a family of bipartite expander graphs
$\{G_{\epsilon, m}\}_{m}$
whose special expansion $\epsilon \leq 1/\log^4 m$.
The expander graph family $\{G_{\epsilon, m}\}_{m}$ 
can be constructed in the aforementioned manner.

\paragraph{Input.}
The input is an array ${\bf I}$ of length $n$ which is promised
to be $1/(\log n)^{C_\smalltriangledown}$-sparse.
Interpret $\I$ as an array of $n'$ \emph{super-elements} 
where $n' := n/\floor{\log n}$,
each super-element consists of $\floor{\log n}$ 
consecutive elements in $\I$,
and a super-element is real if it consists of 
at least one real element.
Assume that $n' = m \cdot \floor{d/2}$ 
for some perfect square $m$
and $d = \Theta(\log^c m)$ is the degree
of the aforementioned bipartite expander 
graph $G_{\epsilon, m}$ where $c > 4$
is an appropriate constant.
For now, we assume that $n$ is divisible by $\floor{\log n}$, and 
that $n'$ is divisible by $\floor{d/2}$ --- 
see Remark~\ref{rmk:notdiv-sparse} regarding how to deal
with general parameters. \elaine{i tweaked this remark}


\paragraph{Algorithm.}
Similar to the lossy loose compactor
algorithm described in the proof of Theorem~\ref{thm:initllc},
except that we now parametrize the expander graph
family differently as explained above, 
we run the algorithm on {\it super-elements} throughout, 
and moreover,
we introduce the following parameter modifications:

\begin{itemize}[leftmargin=5mm,itemsep=1pt]
\item 
The array of length $n' = m \cdot \floor{d/2}$ super-elements
is divided into $m$ chunks of $\floor{d/2}$ super-elements. 
We redefine sparse and dense chunks as follows:
a {\it sparse} chunk is one that has at most $d/(2\log^2 m)$ 
real super-elements. Any chunk that is not sparse is said to be dense.

\item 
We will run the ${\sf ProposeAcceptFinalize}$
subroutine for ${\sf iter} := \log n'/\log \log n'$ 
iterations.   
Moreover, in every iteration, 
each right vertex sends a rejection
if it receives more than $d/(2\log^2 m)$ proposals; 
otherwise it sends an acceptance message. 
Each left vertex become satisfied
if it receives at least 
$\floor{d/2}$ acceptance messages.  

\item 
After 
the dense chunks distribute their real super-elements to the right vertices,  
we compress all chunks such that each chunk contains
$\floor{d/(2\log^2 m)}$ 
super-elements, 
without losing any 
real super-elements in the process (see Fact~\ref{fct:allsat}). 
\end{itemize}

Last but not the least, the circuit implementations
of the ${\sf ProposeAcceptFinalize}$
subroutine and the online routing phase 
are somewhat non-trivial, and needs to use the ${\bf SlowDistr}$ 
and ${\bf SlowTC}$ primitives --- we will explain  
these details later.
\end{mdframed}

Note that for sufficiently large $n$, 
$\log m = \Theta(\log n)$ and therefore the above
algorithm produces an output that is 
$\Theta(1/\log^2 n) < 1/\log n$
fraction of the original length.

\begin{lemma}
In each iteration of the ${\sf ProposeAcceptFinalize}$
algorithm, at most 
$32/\log^4 m$ fraction of the 
remaining unsatisfied left vertices remain unsatisfied.  
\label{lem:remainunsat}
\end{lemma}
\begin{proof}
Let $B := \floor{d/2}$ be the number of super-elements of a chunk.
The fraction of dense chunks 
is at most $1/(\log n)^{C_\smalltriangledown - 3}$, since otherwise
the total number of 
real elements in the input array would be 
greater than 
\[B \cdot \frac{1}{\log^2 m}  \cdot 
\frac{1}{(\log n)^{C_\smalltriangledown -3}}
\cdot \frac{n'}{B} \geq \frac{n}{(\log n)^{C_\smalltriangledown}}\]

\ignore{
Let $S \subset L$ such that $\abs{S} \leq m/(32B)$, 
and let  $\epsilon = \frac{1}{64}$.
In each iteration of Algorithm~\ref{algo:slowMatch}, the number of unsatisfied vertices $|L'|$ decreases by a factor of~$2$.
}

Let $U \subseteq L$ 
be the set of unsatisfied vertices at beginning of 
any given iteration,  
let $R_{\rm neg} \subseteq {\sf neighbors}(U) \subseteq R$ 
be the set of neighbors that respond with a rejection.
Then, $e(U, R_{\rm neg}) > |R_{\rm neg} | \cdot d/(2 \log^2 m)$. 
From the 
expander mixing lemma (Lemma~\ref{lem:expander} 
of Appendix~\ref{sec:expanderprelim}),
we obtain 
$$\frac{|R_{\rm neg}|\cdot d}{2\log^2 m}
< e(U,R_{\rm neg}) \leq \frac{d \abs{U}\abs{R_{\rm neg}}}{m}
+\epsilon d \sqrt{\abs{U}\abs{R_{\rm neg}}}.$$ 

Dividing by $\abs{R_{\rm neg}}d$ and rearranging, 
we have that 
$\epsilon \sqrt{\abs{U}/\abs{R_{\rm neg}}} > 1/(2\log^2 m)-\abs{U}/m$. 
Since $\abs{U}/m \leq 1/(\log n)^{C_\smalltriangledown - 3}$ 
(recall that $U$ is initially all the dense chunks
on the left), 
we have that 
$$
\sqrt{\abs{U}/\abs{R_{\rm neg}}} > 
\frac{1}{\epsilon}\cdot \frac{1}{2\log^2 m}-\frac{1}{\epsilon}\cdot \frac{\abs{U}}{m} > \frac{1}{\epsilon}\cdot 
\left(\frac{1}{2\log^2 m} - \frac{1}{(\log n)^{C_\smalltriangledown - 3}}\right),
$$
Since $C_\smalltriangledown > 8$, 
and $\epsilon \leq \frac{1}{\log^4 m}$, we have that
$\sqrt{\abs{U}/\abs{R_{\rm neg}}} \geq 0.25 \log^2 m$, 
i.e., ${\abs{U}/\abs{R_{\rm neg}}} \geq \frac{1}{16} \cdot \log^4 m$, 
that is, 
$\abs{R_{\rm neg}} \leq 16 \abs{U}/\log^4 m$.

We conclude that the number of vertices in $R$ that respond 
with a rejection is at most $16 \abs{U}/\log^4 m$. 
Therefore, the number of edges
that receive a rejection is at most  
$16 \abs{U} d /\log^4 m$.
For a left vertex to remain unsatisfied,
it must receive at least $d/2$ rejections.
This means that at most 
$32 \abs{U} /\log^4 m$ left vertices can remain unsatisfied.

\ignore{
As $L'$ has $d_\epsilon \abs{L'}$ outgoing edges, and $R'_{\rm neg}$ has at most $d_\epsilon \abs{R'_{\rm neg}} \leq d_\epsilon \abs{L'}/4$ incoming edges, at most one quarter of edges in $L'$ lead to $R'_{\rm neg}$ and yield a negative reply. Since $d_\epsilon=B/2$ and every vertex in $L'$ sends $d_\epsilon$ requests and all negative are from $R'_{\rm neg}$, there are at most $d_\epsilon \abs{L'}/4$ negative replies, and therefore at most  $\abs{L'}/2$ nodes in $L'$ get more than $B=d_\epsilon/2$ negatives. We conclude that at least $\abs{L'}/2$ nodes become satisfied.
}
\end{proof}

\begin{fact}
Suppose 
that $n'$ is sufficiently large. Then, 
after ${\sf iter} := \log n'/ \log \log n'$ iterations, all left vertices become satisfied.
\label{fct:allsat}
\end{fact}
\begin{proof}
We only need to make sure that 
$\left(\frac{32}{\log^4 m}\right)^{\sf iter} \cdot m < 1$, 
that is, 
$${\sf iter} > \log m/\log\left(\frac{\log^4 m}{32}\right)
= \frac{\log m}{4 \log \log m - 5}
.$$
Therefore, 
for sufficiently large $n$, it suffices
to make sure that ${\sf iter} > \frac{\log n'}{\log \log n'}$.
\end{proof}


\ignore{TODO: redo the analysis of the offline circuit size}

\paragraph{Circuit implementation.}
We now discuss how to implement the above meta algorithm 
in circuit.
\begin{itemize}[leftmargin=5mm,itemsep=1pt]
\item 
To determine whether each super-element is real or not,
all super-elements in parallel run the counting circuit of Fact~\ref{fct:countprelim}
and then call a comparator circuit of Fact~\ref{fct:comparator}.
In total, this step takes $O(n)$ generalized boolean gates and $O(\log \log n)$ depth.
\item 
To determine whether each chunk is sparse or dense,
all chunks in parallel run the counting circuit 
of Fact~\ref{fct:countprelim}
and then call a comparator circuit of Fact~\ref{fct:comparator}.
In total, this step takes 
$O(n')$ generalized boolean gates 
and $O(\log d) = O(\log \log n)$ depth.

\item 
Next, we invoke the ${\sf ProposeAcceptFinalize}$ algorithm.
In each iteration: 
\begin{itemize}[leftmargin=5mm,itemsep=1pt]
\item Every facility (i.e., right vertex) need to tally how many
proposals it received, and decide
whether it wants to send rejections or acceptance
messages. For each facility, this requires a counting circuit
of Fact~\ref{fct:countprelim}, and a comparator
circuit of Fact~\ref{fct:comparator}.
Then, the decision can be propagated over a 
binary tree to all $d$ edges.
Accounting for all facilities, this step 
in total requires 
$O(n')$ generalized boolean gates and 
$O(\log d) = O(\log \log n)$ depth.

\item 
Every factory (i.e., left vertex) needs
to tally how many acceptance messages it has received, and 
decide if it wants to mark itself as satisfied.
If it marks itself as satisfied, it will also  
mark all edges over which an acceptance 
message is received as being part of the matching $M$.
This can be done in a similar fashion as how facilities
tally their proposals, in total taking $O(n')$
generalized boolean gates and $O(\log \log n)$ depth.
\end{itemize}
Accounting for all 
$\log n' / \log \log n'$ iterations, the total
depth is at most $O(\log n)$, and the total
number of generalized boolean gates 
is at most $O(n') \cdot \log n' / \log \log n' = O(n)$.
\elaine{check this calculation}

\item 
Next, each dense chunk $u \in L$ 
must send one real super-element over each 
of an arbitrary subset of ${\sf load}(u)\leq d/2$ 
edges outgoing from $u$ in the matching $M$.
This can be accomplished by invoking an instance
of ${\bf SlowDistr}$ (Corollary~\ref{cor:slowdist}) for each chunk, such that 
in each dense chunk, each real super-element is sent over an edge in $M$.
Thus, each chunk takes 
$O(\floor{d/2}\cdot (w\log n + \log\floor{d/2}))
= O(\floor{d/2}\cdot w\log n)$ number of generalized boolean gates
and $O(\floor{d/2})$ total number of $(w \cdot \log n)$-selector gates.
Accounting for all chunks,
the total number of generalized boolean gates 
is at most 
$O(m \cdot \floor{d/2}\cdot w\log n) = O(nw)$, 
the total number of $(w \cdot \log n)$-selector gates  
is at most 
$O(n/\log n)$, 
and the depth is at most
$O(\log^2 \floor{d/2}) =O(\log n)$.
Recall that each $(w \cdot \log n)$-selector gate
can be implemented as 
$O(\log n)$ number of 
$w$-selector gates, and using $O(\log n)$
generalized boolean gates to propagate the flag
over a binary-tree of $\log n$ leaves and depth $\log \log n$. 

Therefore, in total, this step can be implemented 
with $O(nw)$
generalized boolean gates, 
$O(n)$ number of 
$w$-selector gates, and 
in depth $O(\log n)$.

\elaine{i modified the counting here}

\ignore{
suppose that $M_u \in \{0, 1\}^d$ 
is the bit-vector saved 
that denotes whether each of $u$'s outgoing edges   
is in the matching $M$ or not.
Suppose that ${\bf I}_u$ denotes the part of the input array, 
containing $\floor{d/2}$
elements, assigned to the vertex $u$. 
}

\item  
Now, all dense chunks mark all its super-elements as fillers.
This can be done by having each chunk broadcast
its dense/sparse
indicator bit over a binary tree to all $d$ positions of the chunk.  
In total, we can implement it 
with a circuit of $O(n')$
generalized boolean gates and $O(\log d) = O(\log \log n)$ 
depth.

\item 
Finally, we need to compress all chunks
on the left and the right to $\floor{d/(2\log^2 m)}$ super-elements.
This can be accomplished
by applying a ${\bf SlowTC}$ circuit to each chunk (Lemma~\ref{thm:slowtc}), 
and the number of generalized boolean gates, $w$-selector gates,
and depth are asymptotically the same as the 
earlier step in which we invoke a ${\bf SlowDistr}$ instance per chunk.
\end{itemize}

Summarizing the above, we get that 
the entire sparse loose compactor algorithm 
requires 
$O(nw)$
generalized boolean gates, 
$O(n)$ number of $w$-selector gates, 
and $O(\log n)$ depth.


\begin{Remark}
So far, we have assumed that $n$ is divisible by $\floor{\log n}$
and $n' := n/\floor{\log n}$ is equal to $m \cdot \floor{d/2}$
for some perfect square $m$,
and $d = \Theta(\log^c m)$ is the degree
of the aforementioned bipartite expander
graph $G_{\epsilon, m}$ where $c > 4$ is an appropriate constant.

If the above is not satisfied, 
we can let $n' := \ceil{n/\floor{\log n}}$. 
If $n'$ does not satisfy 
the above, we can find the largest 
$m^*$ such that $n' \geq m^* \cdot \floor{d^*/2}$ (note that $d^*$
is a function of $m^*$ for a fixed $\epsilon$).
Now, we can round $m^*$ up to the next 
perfect square $m$, and still use $d^*$ as the degree 
of the bipartite expander graph.
We can pad the array with fillers such that contains $m \cdot d^*$ 
super-elements, and then run the sparse loose compactor algorithm.
With this modification, 
one can check that Lemma~\ref{lem:remainunsat}
and Fact~\ref{fct:allsat} still hold. Therefore, our 
earlier analyses hold.
The padding incurs only $1 + o(1)$ blowup 
in the array's length, i.e., $n'' = (1 + o(1)) n'$.
Our algorithm compresses the array to 
$n'' /(\log m)^2$ in the number of super-elements,
for sufficiently large $n$ and thus 
sufficiently large $n'=\ceil{n/\log n}$, 
the output length  
is upper bounded by $\floor{n/\log n}$.
\label{rmk:notdiv-sparse}
\end{Remark}

\section{Linear-Sized, Logarithmic-Depth Tight Compaction Circuit}
\label{sec:tc}

Putting it all together, we can now realize
a linear-sized, logarithmic-depth tight compaction circuit, as stated
in the following theorem:

\begin{theorem}[Linear-sized, logarithmic-depth tight compaction circuit] 
There is an $(n,w)$-tight compaction circuit 
with $O(nw) \cdot \max(\poly(\log^* n - \log^* w), 1)$ 
generalized boolean gates, 
$O(n) \cdot \max(\poly(\log^* n - \log^* w), 1)$ 
number of $w$-selector gates, and of depth $O(\log n)$.
\label{thm:tccircuit}
\end{theorem}
Note that the above theorem 
and Lemma~\ref{lem:operational} together would 
imply
the following corollary.
\begin{corollary}
There is a circuit of size $O(nw) \cdot \max(\poly(\log^* n - \log^* w), 1)$
and depth $O(\log n + \log w)$ 
that can sort any array containing elements with 1-bit keys
and $w$-bit payloads.
\end{corollary}

\medskip
\begin{proofof}{Theorem~\ref{thm:tccircuit}}
We construct a linear-sized, logarithmic-depth tight
compaction circuit as follows.
\begin{mdframed}
\begin{center}
{\bf Tight compaction}
\end{center}

\paragraph{Input.}
An array ${\bf I}$ containing $n$ elements each with a $w$-bit payload 
and a 1-bit key.

\paragraph{Algorithm.}
\begin{enumerate}[leftmargin=5mm,itemsep=1pt]
\item
{\it Approximate tight compaction.}
Apply a $(1/(\log n)^{C_\smalltriangledown}, n, w)$-approximate 
tight compactor to the input array ${\bf I}$; let ${\bf X}$ denote
the outcome.
\label{step:approxtc}
\item 
{\it Count and label.}
Count how many $0$-keys 
there are in the array ${\bf I}$, let the result be ${\sf cnt}$.
For each $i \in [n]$ in parallel: 
\begin{itemize}[leftmargin=5mm,itemsep=1pt,topsep=3pt]
\item if ${\bf X}[i]$
has the key $1$ and $i \leq {\sf cnt}$, mark it as ${\tt red}$;  
\item else if ${\bf X}[i]$ has the key $0$ and $i > {\sf cnt}$, mark
it as ${\tt blue}$; 
\item else the element ${\bf X}[i]$ is uncolored.
\end{itemize}
\label{step:countlabel}
\item 
{\it Sparse loose compaction.}
Apply a sparse loose compactor 
to the outcome of the previous step; the outcome
is an array ${\bf Y}$ whose length is 
$\floor{n/\log n}$ containing all colored elements 
in ${\bf X}$, padded with filler elements
to a length of $\floor{n/\log n}$.
\label{step:sparselc}
\item 
{\it Slow swap.} Let ${\bf Y}' := {\bf SlowSwap}({\bf Y})$.
\label{step:slowswap}
\item 
{\it Reverse route.}
Reverse route the array ${\bf Y}'$ by reversing
the routing decisions made in Step~\ref{step:sparselc},
and let the outcome be ${\bf Z}$ which has length $n$. 

\label{step:reverse-final}

\item 
{\it Output.}
The output ${\bf O}$ is obtained by performing
a coordinate-wise select operation 
between ${\bf Z}$ and ${\bf X}$:
\[
\forall i \in [n]: \ \ 
{\bf O}[i] := \begin{cases}
{\bf Z}[i] & 
\text{if } {\bf X}[i] \text{ was marked ``{\tt misplaced}''} \\ 
{\bf X}[i] & \text{o.w.} 
\end{cases}
\]
\label{step:output-final}
\end{enumerate}
\end{mdframed}


\paragraph{Implementing the algorithm in circuit.}
Step~\ref{step:approxtc}
is implemented with the approximate tight compaction
circuit of Theorem~\ref{thm:approxtc}.

Step~\ref{step:countlabel}
is implemented as follows.
First, 
use the counting circuit of Fact~\ref{fct:countprelim} to compute
${\sf cnt}$.
Then, use the binary-to-unary circuit of 
Fact~\ref{fct:binary_to_unary}
to write down a string of $n$ bits where the beginning 
${\sf cnt}$ bits are $0$ and all other bits are $1$.
Next, all positions $i \in [n]$ uses the comparator
circuit of Fact~\ref{fct:comparator}
to compute its ``{\tt misplaced}'' label.

Step~\ref{step:sparselc}
is implemented with the sparse loose compactor circuit
of Theorem~\ref{thm:lcsparse}.
Step~\ref{step:slowswap}
is implemented using the ${\bf SlowSwap}$
circuit of Theorem~\ref{thm:slowswap}.
Step~\ref{step:reverse-final}'s costs are absorbed
by Step~\ref{step:sparselc}.
Finally, Step~\ref{step:output-final}
can be accomplished with $n$ generalized boolean gates.

Summarizing the above, the entire 
tight compaction circuit requires  
$O(nw) \cdot \max(\poly(\log^* n - \log^* w), 1)$
generalized boolean gates, 
$O(n) \cdot \max(\poly(\log^* n - \log^* w), 1)$
number of $w$-selector gates, and has depth $O(\log n)$.
\end{proofof}

\section{Sorting Circuit for Short Keys}
\label{sec:sort-circuit}

\subsection{Circuit Implementations of Additional Building Blocks}
Earlier, we described various building blocks 
for an Oblivious PRAM model. We now discuss  
the size and depth 
bounds for these building blocks in the circuit model.

\paragraph{Sorting elements with ternary keys.}
Given Theorem~\ref{thm:tccircuit}, 
and Fact~\ref{fct:countprelim},
we can implement
the algorithm of Theorem~\ref{thm:ternarysort}
using a circuit 
with $O(n w) \cdot \max(1, \poly (\log^* n - \log^* w))$
generalized boolean gates, $O(n)  \cdot 
\max(1, \poly (\log^* n - \log^* w))$ 
number of $w$-selector gates, and of
depth $O(\log n)$.
This leads to the following fact:

\begin{fact}
There exists a circuit  
with $O(n w) \cdot \max(1, \poly (\log^* n - \log^* w))$
generalized boolean gates, $O(n)  \cdot 
\max(1, \poly (\log^* n - \log^* w))$
number of $w$-selector gates, and of
depth $O(\log n)$, 
capable of sorting any input array 
containing $n$ elements 
each with a key from the domain $\{0, 1, 2\}$
and a payload of $w$ bits.
\label{fct:ternarysort-circ}
\end{fact}

\paragraph{Slow sorter and slow alignment.}
We now discuss how to implement the 
earlier ${\bf SlowSort}^K(\cdot)$ 
and ${\bf SlowAlign}^{K, K'}(\cdot)$ algorithms in circuit. 

\begin{fact}[${\bf SlowSort}^K(\cdot)$ circuit]
Let $n$ be the length of the input array and $w$ be the length
of each element's payload. Recall that each element has a key
from the domain $[0, K-1]$, and let $k := \log K$. 
The ${\bf SlowSort}^K(\cdot)$ algorithm 
of Theorem~\ref{thm:slowsort}
can be implemented as a circuit with 
$O(n K \cdot (w+k)) \cdot \max(1, \poly (\log^* n - \log^* (w+k)))$
generalized boolean gates, $O(n K)  \cdot 
\max(1, \poly (\log^* n - \log^* (w+k)))$
number of $(w + k)$-selector gates, and of depth $O(\log n + k)$.
\label{fct:slowsort-circ}
\end{fact}
\begin{proof}
Recall the ${\bf SlowSort}^K(\cdot)$ algorithm of 
Theorem~\ref{thm:slowsort} where $K := 2^k$:
\elaine{hard coded refs}
\begin{enumerate}[leftmargin=5mm,itemsep=1pt]
\item 
Step~\ref{stp:slowsort:count} can be 
implemented using $K$ parallel instances of the counting
circuit of Fact~\ref{fct:countprelim} on arrays
of length $n$, 
and then using the all-prefix-sum circuit 
of Fact~\ref{fct:prefixsumcircuit} on an array of length $K$
where the entire sum is promised to be at most $O(\log n)$ bits long.
In total, Step~\ref{stp:slowsort:count} requires 
a circuit with $O(nK + K \log n) = O(nK)$
generalized boolean gates 
and of depth $O(\log K + \log n) = O(k + \log n)$.

\item 
Step~\ref{stp:slowsort:copy} can be implemented 
by broadcast each element of $A$ over a binary tree of $K$ leaves, 
and then having all $nK$ elements perform 
a comparison in parallel using Fact~\ref{fct:comparator}.
This requires $O(n K)$ number of $(w + k)$-selector gates, 
$O(n K)$ generalized boolean gates, 
and at most $O(k)$ depth.

\item 
Step~\ref{stp:slowsort:ternary} invokes $K$ parallel instances of the 
generalized binary-to-unary conversion
circuit on arrays of length $n$, 
and $K$ parallel instances of the ternary-key sorting circuit.
This requires  
$O(n K (w + k)) \cdot \max(1, \poly (\log^* n - \log^* (w+k)))$
generalized boolean gates, $O(n K)  \cdot 
\max(1, \poly (\log^* n - \log^* (w + k)))$
number of $(w+k)$-selector gates, and has  
depth $O(\log n)$.
\item 
Step~\ref{stp:slowsort:populate} can be accomplished  
in a circuit with 
$O(n K)$ number of generalized boolean gates, 
$O(n K)$ number of $(w + k)$-selector gates
and of depth $O(\log K) = O(k)$.
Note that we can use a single bit
to mark whether each element in each of $B'_0, B'_1, \ldots,
B'_{K-1}$
has a real key in the range $[0, K-1]$ or not.
\end{enumerate}

Summarizing the above, we have that the entire 
${\bf SlowSort}^K(\cdot)$ algorithm 
can be implemented as a circuit with 
$O(n K (w + k)) \cdot \max(1, \poly (\log^* n - \log^* (w + k)))$
generalized boolean gates, $O(n K)  \cdot 
\max(1, \poly (\log^* n - \log^* (w + k)))$
number of $(w + k)$-selector gates, and of depth $O(\log n)$.
\end{proof}

We now discuss the circuit implementation of the 
${\bf SlowAlign}^{K,K'}(\cdot)$ algorithm of 
Theorem~\ref{thm:slowalign}.

\begin{fact}[${\bf SlowAlign}^{K,K'}(\cdot)$ circuit]
Let $n$ be the length of the input array and $w$ be the length
of each element's payload. Recall that each element has a key
from the domain $[0, K-1]$, and an index from the 
domain $[0, K'-1]$.  Let $k = \log K$ and $k' = \log K'$.
The ${\bf SlowAlign}^{K,K'}(\cdot)$ algorithm
of Theorem~\ref{thm:slowsort}
can be implemented as a circuit with
$O(n \cdot (K + K') \cdot (w + k + k')) 
\cdot \max(1, \poly (\log^* n - \log^* (w + k + k')))$
generalized boolean gates, $O(n (K + K'))  \cdot 
\max(1, \poly (\log^* n - \log^* (w + k + k')))$
number of $(w + k + k')$-selector gates, and of 
depth $O(\log n + k + k')$.
\label{fct:slowalign-circuit}
\end{fact}
\begin{proof}
Recall that ${\bf SlowAlign}^{K, K'}$ invokes
one instance of ${\bf SlowSort}^K$ on an array of length $n$
containing $(w + k')$-bit payloads, 
and one instance of ${\bf SlowSort}^{K'}$ 
on an array of length $n$ containing $(w + k)$-bit payloads and 
its reverse routing circuit. Therefore, the fact follows from 
Fact~\ref{fct:slowsort-circ}.
\end{proof}

\elaine{in this case the selector gate's input depend
on boolean bits. can we still 
say we do all the selectors at the end?}

\paragraph{Finding the dominant key.}
We now analyze the complexity of the ${\bf FindDominant}$
algorithm (Theorem~\ref{thm:finddominant}) when implemented in circuit.
Note that the ${\bf FindDominant}$ algorithm
need not look at the elements' payload strings. Therefore,
we may plug in an arbitrary $w \geq 0$ as the fake payload length.
\elaine{hard coded refs}
\begin{enumerate}[leftmargin=5mm,itemsep=1pt]
\item 
Step~\ref{stp:finddom:base}, i.e., the base case calls the ${\bf SlowSort}^K$ algorithm 
on an array of length at most $n/K$ where $K := 2^k$. 
Therefore, this step  requires 
$O(n \cdot (w + k)) 
\cdot \max(1, \poly (\log^* n - \log^* (w + k)))$
generalized boolean gates, $O(n)  \cdot 
\max(1, \poly (\log^* n - \log^* (w + k)))$
number of $(w + k)$-selector gates, and of
depth $O(\log n + k)$. 
\item 
In each of the $O(k)$ recursive calls to ${\bf FindDominant}$, 
the array length would reduce by a factor of $4$, 
and during each recursive call, we divide the array into groups
of $8$ and run an AKS circuit on each group. 
In total over all levels of recursion, this requires
$O(n)$ number of $(w + k)$-selector gates, $O(n)$ 
generalized boolean gates, and $O(k)$ depth.
\end{enumerate}

Therefore, we have the following fact.

\begin{fact}[${\bf FindDominant}$ circuit]
Suppose that $n > 2^{k + 7}$ and moreover $n$ is a power of $2$. 
\elaine{note this condition}
Let $A$ be an array containing $n$ elements
each with a $k$-bit key, and suppose
that $A$ is $(1-2^{-8k})$-uniform.
Fix some arbitrary $w \geq 0$ (which need not
be the element's payload length\footnote{Note that the algorithm
need not look at the elements' payload strings.}). 
Then, there is a circuit 
that can correctly identify the dominant key given any such $A$;
and moreover, the circuit contains 
$O(n \cdot (w + k)) \cdot \max(1, \poly (\log^* n - \log^* (w + k)))$
generalized boolean gates, $O(n)  \cdot 
\max(1, \poly (\log^* n - \log^* (w + k)))$
number of $(w + k)$-selector gates, and of
depth $O(\log n + k)$. 
\label{fct:finddominant-circuit}
\end{fact}

\subsection{Putting Everything Together: Sorting 
Short Keys in the Circuit Model}

We now finish it off and discuss how to implement the algorithm of 
Theorem~\ref{thm:mainsort-opram}
in the circuit model.
To do this, it suffices to describe how to 
implement a nearly orderly segmenter in circuit, 
and how to sort a nearly orderly  
array in circuit. 

\paragraph{Nearly orderly segmenter.}
Recall that for $k \leq \log n$, 
the algorithm of Theorem~\ref{thm:segmenter}
is a comparator-based circuit 
 with $O(nk)$ comparators and
of $O(k)$ depth.
We would like to convert this comparator-based
circuit to a circuit with generalized boolean gates
and $w$-selector gates.

\begin{fact}[$(2^{-8k}, 2^{3k})$-orderly segmenter circuit]
Suppose that $k \leq \log n$. \elaine{note the assumption}
There exists a $(2^{-8k}, 2^{3k})$-orderly-segmenter
circuit with $O(nk^2)$ generalized boolean gates, $O(nk)$ 
number of $(w + k)$-selector gates,
and of depth $O(k)$.
\label{fct:segmenter-circ}
\end{fact}
\begin{proof}
If we used a na\"ive method for converting
the comparator-based circuit in Theorem~\ref{thm:segmenter}
to a circuit with generalized boolean gates and
$w$-selector gates, 
the resulting circuit depth would have depth $O(k \log k)$
because every comparator can be implemented as an
$O(k)$-sized and $O(\log k)$-depth
boolean circuit due to Fact~\ref{fct:comparator}.

Fortunately, we can rely on a {\it pipelining} technique to make
the depth smaller.
\begin{itemize}[leftmargin=5mm]
\item 
In the beginning, all input bits of the input layer 
are ``ready''. 
All other comparators not in the input layer
see all bits of their inputs as ``not ready''.
\item 
Whenever a comparator detects 
a new $i \in [k]$ such that both of its inputs 
have the $i$-th bit ready, it can compare the $i$-th bits
of the two inputs, and as a result, the $i$-th bits
of the two outputs of the gate will be ready. 
\end{itemize}

Using this pipelining technique, we can 
first compute all the generalized boolean gates 
which will populate the flags 
of all selector gates.
This step takes $O(k)$ depth and $O(nk^2)$
generalized boolean gates.
Next, we can evaluate all $O(nk)$ number of 
$(w + k)$-selector gates 
in a topological order; this can be 
accomplished in $O(k)$ depth. 
\end{proof}

\paragraph{Sorting a nearly orderly array.}
We now describe how to implement the algorithm of 
Theorem~\ref{thm:corrector}  in circuit.

\begin{itemize}[leftmargin=5mm,itemsep=1pt]
\item 
Step~\ref{step:gb} calls the ${\bf FindDominant}$ 
circuit of Fact~\ref{fct:finddominant-circuit}, 
and then for each segment, invokes one copy of the counting circuit
of Fact~\ref{fct:countprelim}
and the generalized binary-to-unary conversion circuit
of Fact~\ref{fct:binary_to_unary}.
Therefore, this step can be accomplished 
with a circuit 
containing 
$O(n \cdot (w + k)) \cdot \max(1, \poly (\log^* n - \log^* (w + k)))$
generalized boolean gates, $O(n)  \cdot 
\max(1, \poly (\log^* n - \log^* (w + k)))$
number of $(w + k)$-selector gates, and of
depth $O(\log n + k)$.
\item 
For Step~\ref{step:compact-correction}: to mark each element
with its segment index, we can simply hard-wire the segment indices 
in the circuit.  
Then, we invoke the oblivious compaction  
circuit of Theorem~\ref{thm:tccircuit} 
which requires
$O(n(w+k)) \cdot \max(\poly(\log^* n - \log^* (w+k)), 1)$
generalized boolean gates,
$O(n) \cdot \max(\poly(\log^* n - \log^* (w+k)), 1)$
number of $(w+k)$-selector gates, and $O(\log n)$ depth.

\item 
Step~\ref{step:slowalign-correction}
invokes the ${\bf SlowAlign}^{K, K^2}$ 
circuit of Fact~\ref{fct:slowalign-circuit} on $3n/K^2$ elements.
This requires 
$O(n \cdot (w + k)) 
\cdot \max(1, \poly (\log^* n - \log^* (w + k)))$
generalized boolean gates, $O(n)  \cdot 
\max(1, \poly (\log^* n - \log^* (w + k)))$
number of $(w + 3k)$-selector gates, and has depth $O(\log n + k)$.

\item 
Step~\ref{step:reverseroute-correction}
is a reverse routing step whose costs are absorbed by 
Step~\ref{step:compact-correction}.
\item 
Step~\ref{step:multikey-correction}
invokes $K^2$ 
instances of the counting circuit of Fact~\ref{fct:countprelim}
each on an array of length $n/K^2$.
This takes $O(n)$ generalized boolean gates.

\item 
Step~\ref{step:compact2-correction}
invokes the compaction circuit 
of Theorem~\ref{thm:tccircuit} on an array containing
$K^2$ elements, where each element is of length $W := O(n(k+w)/K^2)$.
\elaine{note: need wide payload}
Since $K^2 \cdot n(k+w)/K^2 = O(n(k+w))$,
this step requires 
$O(n(k+w)) \cdot \max(\poly(\log^* n - \log^* (w+k)), 1)$ 
generalized boolean gates, 
$O(K) \cdot \max(\poly(\log^* n - \log^* (w+k)), 1)$
number of $W$-selector gates, and of depth $O(\log K)$.
\item 
Step~\ref{step:sortwithin-correction}
invokes $K$ instances the ${\bf SlowSort}^K$
circuit of Fact~\ref{fct:slowsort-circ}
each on an array of length $n/K^2$.
This cost of this step is dominated by that of 
Step~\ref{step:slowalign-correction}.

\item 
Step~\ref{step:final-correction}
is a reverse routing step whose costs are dominated
by Step~\ref{step:compact2-correction}.
\end{itemize}

Due to Lemma~\ref{lem:operational}, the above can be implemented
as a constant fan-in, constant fan-out 
a boolean circuit of size $O(n(w+k)) \cdot 
\max(1, \poly(\log^* n - \log^*(w+k)))$
and depth $O(\log n + \log w)$, assuming that $n > 2^{4k + 7}$.


\begin{fact}[Sorting a $(2^{-8k}, 2^{3k})$-orderly array in circuit]
Suppose that $n > 2^{4k + 7}$.
There is 
a constant fan-in, constant fan-out 
boolean circuit that fully sorts an $(2^{-8k}, 2^{3k})$-orderly array
containing $n$ elements each with a $k$-bit key and a $w$-bit payloads, 
whose size is $O(n(w+k)) \cdot \max(1, \poly(\log^* n - \log^* (w+k)))$
and whose depth is $O(\log n + \log w)$.
\label{fct:corrector-circ}
\end{fact}

\paragraph{Sorting short keys in the circuit model.}
Summarizing the above, we get the following theorem:

\begin{theorem}[Restatement of Theorem~\ref{thm:intro-sort-circ}]
Suppose that $n > 2^{4k + 7}$.
There is a constant fan-in, constant fan-out 
boolean circuit that correctly sorts any array 
containing $n$ elements each with a $k$-bit key and a $w$-bit payloads, 
whose size is $O(n k (w+k)) \cdot \max(1, \poly(\log^* n - \log^* (w+k)))$
and whose depth is $O(\log n + \log w)$.
\label{thm:mainsort-circ}
\end{theorem}
\begin{proof}
Follows directly due to the algorithm of 
Theorem~\ref{thm:mainsort-opram} where we implement
the nearly orderly segmenter and the sorter for a nearly orderly array
using the circuits of Facts~\ref{fct:segmenter-circ}
and \ref{fct:corrector-circ}, respectively.
Further, we use Lemma~\ref{lem:operational}
to convert each circuit gadget 
in our operational model to a constant fan-in,
constant fan-out boolean circuit gadget.
\end{proof}

\section*{Acknowledgments}
This work is in part supported by an NSF CAREER Award under the 
award number CNS1601879, a Packard Fellowship, and an ONR YIP award.
We would like to thank Silei Ren for discussions 
and help in an early stage of the project.
Elaine Shi would like to thank Bruce Maggs for explaining the AKS
algorithm, Pippenger's self-routing super-concentrator, the Wallace-tree trick, 
and the elegant work by Arora, 
Leighton, and Maggs~\cite{alm90}, as well as for 
his moral support of this work.

\nocite{optorama}
\bibliographystyle{alpha}
\bibliography{refs,ref,crypto,cache_obliv,bibdiffpriv}

\appendix

\section{Expander Graphs and Spectral Expansion}
\label{sec:expanderprelim}

\ignore{
\begin{definition}[The parameter $\lambda(G)$ {\cite[Definition 21.2]{AroraBarak}}]
  Given a $d$-regular graph $G$ on $n$ vertices, we let $A = A(G)$ be the matrix
  such that for every two vertices $u$ and $v$ of $G$, it holds that $A_{u,v}$
  is equal to the number of edges between $u$ and $v$ divided by $d$. (In other
  words, $A$ is the adjacency matrix of $G$ multiplied by $1/d$.)
  The parameter $\lambda(A)$, denoted also as  $\lambda(G)$, is:
  \begin{align*}
    \lambda(A) =  \max_{{\mathbf v}\in \mathbf{1}^\bot, \| {\mathbf  v}\|_2=1 } \|
    A{\mathbf {v}}\|_2,
  \end{align*}
  where $\mathbf{1}^\bot = \{\mathbf v \mid \sum \mathbf{v}_i = 0\}$.
\end{definition}

\begin{lemma}[Expander mixing lemma {\cite[Definition 21.11]{AroraBarak}}]\label{lem:expander}
  Let $G=(V,E)$ be a $d$-regular $n$-vertex graph. Then, for all sets
  $S,T\subseteq V$, it holds that
  \begin{align*}
    \left|e(S,T) - \frac{d}{n}\cdot |S|\cdot |T| \right| \leq \lambda(G)\cdot
    d\cdot \sqrt{|S|\cdot |T|},
  \end{align*}
\end{lemma}
}

\begin{lemma}[Expander mixing lemma for bipartite graphs~\cite{Haemers_bipartite}]
\label{lem:expander}
  Let $G=(L \cup R,E)$ be a $d$-regular bipartite graph such that $\abs{L}=\abs{R}=n$. 
  Then, for all sets $S \subseteq L$ and $T\subseteq R$, it holds that
  \begin{align*}
    \left|e(S,T) - \frac{d}{n}\cdot |S|\cdot |T| \right| \leq \lambda_2(G)\cdot
    d\cdot \sqrt{|S|\cdot |T|},
  \end{align*}
where $\lambda_2(G)$ is defined as the second largest eigenvalue
of the normalized adjacency matrix $A$ of $G$. In other words,  
$A$ is the adjacency matrix of $G$ multiplied by $1/d$; 
let $\lambda_1 \ge \lambda_2 \ge \dots \ge \lambda_{2n}$ be the eigenvalues of $A$,
then $\lambda_2(G) := \lambda_2$.
The eigenvalue  
$\lambda_2(G) \in (1/d, 1)$ 
is also called the {\it spectral expansion}
of the bipartite graph $G$.
\end{lemma}





\section{Summary and Improvement of Kouck{\'{y}} and Kr{\'{a}}l~\cite{KK21}}
\label{sec:KK21}

Kouck\`y and Kr\`al show a circuit that sorts $n$ integers each is $k$-bit
(without payload),
taking circuit size $O(nk^2)$ and depth $O(\log n + k\log k)$.
In this section, we improve the circuit depth to $O(\log n + k)$
(with the same circuit size).

If $k = \Omega(\log n)$, then the standard AKS sorting network~\cite{aks} 
is applied to $n$ integers directly:
recall that the depth of AKS is $O(k+\log n)$ using the pipelining technique
in Fact~\ref{fct:segmenter-circ} while the circuit size is still $O(kn\log n)$.

For $k \le 0.1 \cdot \log n$, Kouck\`y and Kr\`al takes a counting approach as below
(for readability, integer rounding is omitted).

\begin{enumerate}
\item 
Divide input into chunks, each consists of $2^{5k}$ integers.
For each chunk, sort integers in the chunk using AKS sorting network.
This takes size $O(nk^2)$ and depth $O(k)$.

\item
\label{stp:KK_count}
For each chunk, count the number of occurrences for all integer $i\in[2^k]$,
which yields a short list of $2^k \cdot 5k$ bits
(compared to chunk size, $2^{5k} \cdot k$).
To do so, for each chunk, the following is performed.
  \begin{enumerate}
  \item 
  The sorted chunk is sub-divided into $2^{3k}$ pieces,
  each consists of $2^{2k}$ integers.
  Because it is sorted, the chunk has at most $2^k$ pieces that are non-uniform
  (that is, having more than 1 distinct integers).
  \item
  Counting is straightforward for each uniform piece.
  \item
  The non-uniform pieces are collected into a short list of $2^{3k}$ integers
  using AKS sorting network.
  Then on the short list, a counting is performed for each integer $i \in [2^k]$
  (e.g., Fact~\ref{fct:prefixsumcircuit}),
  resulting the counts of the short list.
  \item
  The counts of all pieces and the short list are summed up,
  resulting the counts of this chunk.
  \end{enumerate}
Notice that the cost of AKS is the same as previous step,
and that counting the short list takes circuit size $O(2^k \cdot 2^{3k})$ and depth $O(k)$.
Hence for all chunks, this takes size $O(nk^2)$ and depth $O(k)$.

\item
Sum up the counts from all $n / 2^{5k}$ chunks, and then calculate
the ``desired counts'' for each chunk when all $n$ integers are sorted.
The counts from all chunks are just $(n / 2^{5k})\cdot (2^k \cdot 5k)$ bits,
so this can be implemented in circuit size $O(nk)$ and depth $O(k+\log n)$,
e.g., using delayed-carry addition and all prefix sums (Facts~\ref{fct:prefixsumcircuit}).

\item
With the desired counts, each chunk restore integers from its counts
using a reversed procedure of Step~\ref{stp:KK_count}.
\item
The concatenation of all restored chunks is the sorted output.
\end{enumerate}
This concludes the improved circuit depth $O(k+\log n)$;
notice that the depth is achieves using the above pipelined AKS,
compared to depth $O(k \log k + \log n)$ in Kouck\`y and Kr\`al.


\section{Epilogue: Reducing the $\poly\log^*$ to $\log^*$}

While Kouck\`y and Kr\'al achieved a sub-optimal depth $O(\log^3 n)$,
their circuit size is $O(nk(w+k)\cdot(1+\log^*n - \log^*(w+k)))$,
which is slightly better than our circuit size,
$O(n k (w+k)) \cdot (1 + \poly(\log^* n - \log^* (w+k)))$.
It turns out that we can combine the techniques in the two papers, and 
achieve optimal depth 
and $O(nk(w+k)\cdot(1+\log^*n - \log^*(w+k)))$ 
circuit size for sorting $k$-bit keys. 
Note that besides this section which is added in hindsight, the rest
of the paper is concurrent and independent work to 
Kouck{\'{y}} and Kr{\'{a}}l~\cite{KK21}.


Specifically, we will prove the following slightly improved theorem in this section. 
\begin{theorem}
\label{thm:better_size_sort}
Suppose that $n > 2^{4k + 7}$.
There is a constant fan-in, constant fan-out 
boolean circuit that correctly sorts 
any array containing $n$ elements each with a $k$-bit key 
and a $w$-bit payloads, whose 
size is $O(n k (w+k)) \cdot (1+ \log^* n - \log^* (w+k))$
and whose depth is $O(\log n + \log w)$.
\label{thm:bettersortcircuit}
\end{theorem}

To get the above theorem, it suffices to replace
the (tight) compaction circuit 
in the proof of Theorem~\ref{thm:mainsort-circ}
with one that achieves logarithmic depth and 
$O(nw (1 + \log^* n - \log^* w))$ size. Recall that 
the proof of Theorem~\ref{thm:mainsort-circ} 
performs a $1$-bit to $k$-bit upgrade using the new techniques developed earlier
in our paper.
Henceforth, we focus on constructing a compaction circuit satisfying the above
requirements.
More specifically, to prove 
the above Theorem~\ref{thm:bettersortcircuit},
it suffices to prove the following:

\begin{theorem}
There exists a tight compaction circuit
of size $O(nw \cdot (1+\log^*(n)-\log^*(w)))$
and depth $O(\log n)$, where $w$ is the width of the payload.
\end{theorem}

Inspired by Kouck{\'{y}} and Kr{\'{a}}l~\cite{KK21}, 
we will start with a logarithmic-depth compaction
circuit that is a $\log\log n$ 
factor non-optimal in size, and 
then use a recursive bootstrapping technique 
to 
compress the circuit size 
while preserving the asymptotical depth.
\ignore{
With the upgraded tight compaction, the sorting 
circuit of Theorem~\ref{thm:better_size_sort}
is obtained directly from our 1-bit to $k$-bit upgrade, 
as described in Section~\ref{sec:sort-circuit}.
Notice that the large-size tight compaction uses our sparse loose compactor from
Section~\ref{sec:sparselc}, while the ``upgrade'' is based on 
Kouck\`y and Kr\'al~\cite[Lemma 19]{KK21} (Figure~\ref{fig:bettersize}).
}


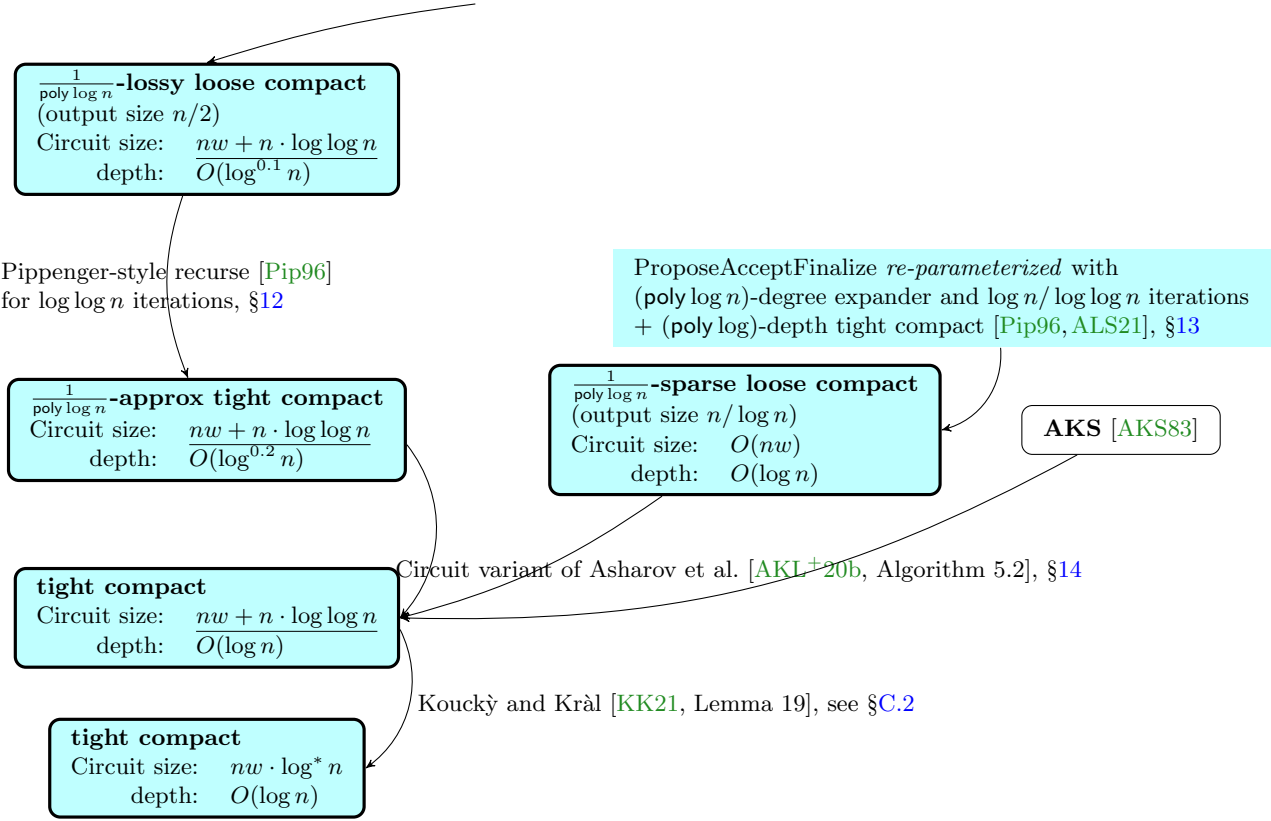
\begin{figure}[t]
\begin{center}
\begin{tikzpicture}[->,>=stealth']
\footnotesize

\definecolor{aqua}{rgb}{0.75, 1.0, 1.0}

\tikzstyle{state} = [rectangle,
           rounded corners,
           draw=black, very thick,
           minimum height=2em,
           inner sep=2pt,
           text centered,]

\tikzstyle{box} = [rectangle,
           minimum height=2em,
           inner sep=2pt,
           text centered,]

\node[box,
] (PAF) 
{
\begin{tabular}{l}
 ProposeAcceptFinalize, const-degree expander, 
 $\log\log n$ iterations (Pippenger~\cite{selfroutingsuperconcentrator})
\end{tabular}
};

\node[state,
  below of=PAF,
  node distance=2cm,
  anchor=east,
  xshift=-1cm,
  fill=aqua,
] (LOSSYLOOSE) 
{
\begin{tabular}{rl}
  \multicolumn{2}{l}{
    \textbf{$\frac{1}{\poly\log n}$-lossy loose compact}
  }\\
  \multicolumn{2}{l}{
    (output size $n/2$)
  }\\
  Circuit size: & \underline{$nw + n \cdot \log\log n$}\\
  depth: & $O(\log^{0.1} n)$\\
\end{tabular}
};

\node[state,
  below of=LOSSYLOOSE,
  node distance=4cm,
  fill=aqua,
] (APPROXTIGHT) 
{
\begin{tabular}{rl}
  \multicolumn{2}{l}{
    \textbf{$\frac{1}{\poly\log n}$-approx tight compact}
  }\\
  Circuit size: & \underline{$nw + n \cdot \log\log n$}\\
  depth: & $O(\log^{0.2} n)$\\
\end{tabular}
};

\node[box,
  below of=APPROXSPARSE,
  node distance=2.25cm,
  xshift=-1.5cm,
  fill=aqua,
] (PAFNEW) 
{
\begin{tabular}{l}
 ProposeAcceptFinalize \emph{re-parameterized} with \\
 ($\poly \log n$)-degree expander and $\log n / \log\log n$ iterations\\
  + ($\poly\log$)-depth tight compact~\cite{selfroutingsuperconcentrator,soda21},  
  \S\ref{sec:sparselc}\\
\end{tabular}
};

\node[state,
  below of=PAFNEW,
  node distance=1.75cm,
  anchor=east,
  fill=aqua,
] (SPARSELOOSE) 
{
\begin{tabular}{rl}
  \multicolumn{2}{l}{
    \textbf{$\frac{1}{\poly\log n}$-sparse loose compact}
  }\\
  \multicolumn{2}{l}{
    (output size $n/\log n$)
  }\\
  Circuit size: & $O(nw)$\\
  depth: & $O(\log n)$\\
\end{tabular}
};

\node[state,
  thin,
  right of=SPARSELOOSE,
  node distance=5cm,
] (AKS) 
{
\begin{tabular}{l}
 \textbf{AKS}~\cite{aks}\\
\end{tabular}
};

\node[state,
  below of=APPROXTIGHT,
  node distance=2.5cm,
  fill=aqua,
] (LTIGHT) 
{
\begin{tabular}{rl}
  \multicolumn{2}{l}{
    \textbf{tight compact}
  }\\
  Circuit size: & \underline{$nw + n \cdot \log\log n$}\\
  depth: & $O(\log n)$\\
\end{tabular}
};

\node[state,
  below of=LTIGHT,
  node distance=2cm,
  fill=aqua,
] (TIGHT) 
{
\begin{tabular}{rl}
  \multicolumn{2}{l}{
    \textbf{tight compact}
  }\\
  Circuit size: & $nw\cdot\log^* n$\\
  depth: & $O(\log n)$\\
\end{tabular}
};

\path 
(PAF.south)      edge[bend right=5] (LOSSYLOOSE.north)

(PAFNEW.320)      edge[bend left=40] (SPARSELOOSE.east)
(LOSSYLOOSE)      edge[bend right=20]
  node[]{
    \begin{tabular}{l}
    Pippenger-style recurse~\cite{selfroutingsuperconcentrator} \\
    for $\log\log n$ iterations, \S\ref{sec:approxtc}  
    \end{tabular}
  } 
  (APPROXTIGHT)
(APPROXTIGHT)      edge[bend left=40]  (LTIGHT.east)
(SPARSELOOSE)      edge[bend left=10] 
  node[anchor=west,xshift=-2cm]{
    Circuit variant of Asharov et al.~\cite[Algorithm 5.2]{paracompact}, \S\ref{sec:tc}
  }
  (LTIGHT.east)
(AKS)      edge[bend left=15] (LTIGHT.east)
(LTIGHT)      edge[bend left=40]
  node[anchor=west]{
    Kouck\`y and Kr\`al~\cite[Lemma 19]{KK21}, see \S\ref{sec:KK_upgrade}
  }
  (TIGHT.east)
;
\end{tikzpicture}
\end{center}
\caption{Blueprint of the improved circuit size,
putting together our techniques with that of Kouck\`y and Kr\`al.
Notice that 
the underlined circuit sizes are larger than
the corresponding ones in Figure~\ref{fig:compaction}.
 }
\label{fig:bettersize}
\end{figure}

\subsection{Compaction Circuit Optimal in Depth but Slightly Non-Optimal in Size}
As a starting point, we will use a compaction
circuit that is optimal in depth but a $\log \log n$ factor
non-optimal in size, as stated in the following theorem: 

\begin{theorem}[Tight compaction: optimal-depth, slightly non-optimal in size] 
\label{thm:largeTC}
There is an $(n,w)$-tight compaction circuit 
with $O(nw + n \log\log n)$ 
generalized boolean gates, 
$O(n)$ 
number of $w$-selector gates, and of depth $O(\log n)$.
\end{theorem}
\begin{proof}
To construct such circuit, 
we need an approximate tight compaction circuit 
that swaps the all but $1/\poly\log n$ fraction of the misplaced
elements in the original array.
We want that this approximate tight compaction circuit
to achieve sub-logarithmic depth, but we allow
the circuit size to be a $\log\log n$ factor non-optimal.
More specifically, we need the following:

\begin{lemma}
\label{lem:approxtc_large}
Fix an arbitrary constant $\widetilde{C} > 1$.
There is an $(1/(\log n)^{\widetilde{C}}, n, w)$-approximate 
tight compaction circuit that has
$O(n \cdot \log\log n)$ generalized boolean gates,
$O(n)$ number of $w$-selector gates, and
with depth at most $O((\log\log n)^2)$.
\end{lemma}
\begin{proof}
The construction is similar to that of Theorem~\ref{thm:approxtc},
the only difference is at Step~\ref{step:compactcolored}
of $\widehat{\bf Swap}^n$:
when performing the lossy loose compaction,
we use the large-size lossy loose compaction
from Theorem~\ref{thm:initllc}
instead of the small circuit from Theorem~\ref{thm:linearllc}.
Symmetrically at Step~\ref{step:reverse-tc} of $\widehat{\bf Swap}^n$,
we also use the large circuit from Theorem~\ref{thm:initllc}.

With these modifications, the depth is $O( (\log\log n)^2 )$
because each lossy loose compaction takes $O(\log\log n)$ detph
and it is recursively applied for $O(\log\log n)$ times 
in $\widehat{\bf Swap}^n$.
The circuit size follows similarly.
\end{proof}

\begin{proofof}{Theorem~\ref{thm:largeTC}}.
We use the same meta-algorithm
as in Section~\ref{sec:tc}
(which was first proposed by
Pippenger~\cite{selfroutingsuperconcentrator} and later used 
in Asharov et al.~\cite{paracompact}).
The algorithm proceeds as follows.
First, sort all but $(1/\poly\log n)$-fraction using 
the low-depth approximate tight compaction (Lemma~\ref{lem:approxtc_large}).
Second,
collect the $(1/\poly\log n)$-fraction misplaced elements
into a short list of $n / \log n$ elements using 
the sparse loose compactor in Theorem~\ref{thm:lcsparse}.
Third, sort the short list using the AKS sorting network.
Finally, reversely route the sorted short list back to the original array.

Using Lemma~\ref{lem:approxtc_large} and~\ref{thm:lcsparse}, 
the performance bound analysis is direct.
Notice that the three steps are
very similar to that of Theorem~\ref{thm:tccircuit},
and the only difference is that we perform a less efficient 
approximate tight compaction at the first step 
(Lemma~\ref{lem:approxtc_large} instead of Theorem~\ref{thm:approxtc}).
\end{proofof}

\end{proof}

\subsection{Improving Circuit Size through Recursive Bootstrapping}
\label{sec:KK_upgrade}
\newcommand{\LargeTC}{\mathbf{LargeTC}}

Next, we  
show how to use the recursive bootstrapping technique 
of Kouck{\'{y}} and Kr{\'{a}}l~\cite[Lemma 19]{KK21}
to compress the circuit size without blowing up  
the asymptotical depth.
The meta-algorithm is identical that of Kouck\`y and Kr\`al,
but we present it in a top-down recursion (compared to their bottom-up)
and parameterize the algorithm using the size and depth of the given larger circuit.


Let $\LargeTC_{n,w}$ be the tight compaction circuit that sorts $n$ elements
each with 1-bit key and $w$-bit payload,
let $O(nw + n\cdot \log\log n)$ be the circuit size
and $O(\log n)$ be the depth.
We construct $\TC_{n,w}$ recursively as below, where
$n$ is the number of elements 
in the input array $A$ and 
and $w$ is the width of each payload string.
For simplicity we also suppose the division and $\log(x)$ always map to 
proper integers in the algorithm.


\paragraph{$\TC_{n,w}(A)$:}  \ \ //Assume: $A$ consists of $n$ elements of width $w$, each with a 1-bit key.
\begin{enumerate}

\item
(Base case.)
If $\log\log n \le w$, invoke $\LargeTC_{n,w}(A)$,
and then output the result.
Otherwise, continue with the following.

\item 
\label{stp:KK:recurse}
Let $n' := \log n$.
Interpret $A$ as $n / n'$ super-pieces,
each super-piece consists of $n'$ elements.
For each super-piece, recursively call $\TC_{n',w}$ 
to sort $n'$ elements in the super-piece.
Let $B$ be the concatenation of all resulting super-pieces.

\item
Let $m' := (\log n)^{1/3}$.
Interpret $B$ as $n / m'$ pieces,
each piece consists of $m'$ elements. 
For each piece, identify itself as 0-, 1-, or mixed-piece,
where 0-piece consists of only 0-elements, 
and 1-piece consists of only 1-elements.
Notice that there are at most $n/n'$ mixed-pieces. 

\item
\label{stp:KK:move_piece}
Invoke $\LargeTC_{n/m',w\cdot m'}$ so that all 0-pieces are moved to the front,
and similarly all 1-pieces are moved to the back.
\weikai{note that $n'=(f^{\circ d}n)^3$ for all $d \ge 1$ so that $n'/m_{d+1} = O(n' / \log n')$}
Invoke $\LargeTC_{n/m',w\cdot m'}$ so that all mixed-pieces are moved to a short scratch array
which consists of $(n/n') \cdot m' = n/(m')^2$ elements.

\item
\label{stp:KK:sort_mixed}
Invoke $\LargeTC_{n/(m')^2, w}$ so that all elements in the scratch array are sorted.

\item
Merge the sorted scratch array elements with those 0- and 1-pieces
using another $\LargeTC_{n/m',w\cdot m'}$ on pieces.
Output the merged array.
\end{enumerate}

We claim the following.

\begin{theorem}
\label{thm:KK_upgrade}
Suppose $\LargeTC_{n,w}$ is a correct tight compaction and
takes circuit size
$O(nw + n\cdot \log\log n)$ and depth $O(\log n)$.
Then, $\TC_{n,w}$ is a correct tight compaction
that sorts $n$ elements each with $w$-bit payload,
takes circuit size $O(nw \cdot (1+\log^*n-\log^*w))$
and depth $O(\log n)$.
\end{theorem}

\begin{proof}

Correctness follows inductively:
the base case is correct by $\LargeTC$,
and then other cases follows as we correctly sort 0- and 1-pieces
in Step~\ref{stp:KK:move_piece} 
and mixed-pieces in Step~\ref{stp:KK:sort_mixed}.
Next, we focus on the circuit size and depth.

The circuit size is
\[
S(n,w) = \begin{cases}
(n/\log n) \cdot S(\log n) 
  + 4S_{\LargeTC_{n/m',w\cdot m'}}
  + S_{\LargeTC_{n/(m')^2,w}}
  & \log\log n > w\\
O(n \cdot w) & \log\log n \le w
\end{cases}.
\]
By Theorem~\ref{thm:largeTC}, we have 
\[
4S_{\LargeTC_{n/m',w\cdot m'}} + S_{\LargeTC_{n/(m')^2,w}}
= 4\cdot O(nw + 2n) + O(nw + 2n)
= O(nw)
\]
for all $n>2^{2^3}$
since $m' = (\log n)^{1/3} > (1/2) \log\log n$ for all $n>2^{2^3}$.
Because the recursive call to $\TC$ itself reduces the problem size
to $\log n$ at Step~\ref{stp:KK:recurse},
the recursion reaches the base case at depth $\log^* n - \log^* w$.
Then the total circuit size $O(nw\cdot (1+\log^*(n)-\log^*(w)))$ 
follows by a simple summation over each recursion depth.


To calculate the circuit depth, observe that ``for each piece/super-piece'' steps
in the procedure are all performed in parallel,
and that only $\LargeTC$ and identifying 0, 1, or mixed pieces take depth $O(\log n)$.
Moreover, for each recursion depth $i$, 
$\LargeTC$ and the identification work on at most $n_i$ items,
where $n_i = \log^{(i)} n$ denotes the number of elements 
in the input of the recursive call $\TC$.
Hence, the total depth is 
\[
D(n) = \sum_{i=0}^{\log^*(n)-\log^*(w)} O(\log n_i) = O(\log n).
\]
\end{proof}

\paragraph{Putting everything together.}
By plugging Theorem~\ref{thm:largeTC}
into $\LargeTC$,
we obtain Theorem~\ref{thm:KK_upgrade} as claimed.
Then, plugging Theorem~\ref{thm:KK_upgrade}
into our 1-bit to $k$-bit upgrade 
(described in Section~\ref{sec:sort-circuit}),
we get Theorem~\ref{thm:better_size_sort}.

\end{document}